\documentclass[11pt, a4paper]{article}
\usepackage[utf8]{inputenc}
\usepackage[utf8]{inputenc}
\usepackage[margin=1in]{geometry}
\usepackage{mathtools}
\usepackage{amsmath}
\usepackage{amsthm}
\usepackage{amssymb}
\usepackage{setspace}
\usepackage{booktabs}
\usepackage{tabularx}
\usepackage{amsmath} 
\usepackage{amssymb}
\usepackage{bm}
\usepackage{ascmac}
\usepackage{setspace}
\usepackage[at]{easylist}
\usepackage[normalem]{ulem}
\usepackage{comment}

\usepackage{multirow}

\mathtoolsset{showonlyrefs}  

\makeatletter

\usepackage{amsthm}
\usepackage{amsfonts}
\usepackage{bbm}
\usepackage{graphics}
\usepackage{enumerate}
\usepackage{float}
\usepackage{caption}
\usepackage{subcaption}
\usepackage{natbib}
\usepackage{epigraph}
\theoremstyle{definition}
\newtheorem{theorem}{Theorem}
\newtheorem{lemma}{Lemma}
\newtheorem{definition}{Definition}
\newtheorem{proposition}{Proposition}

\newtheorem{example}{Example}[section]
\newtheorem{assumption}{Assumption}
\newtheorem{condition}{Condition}

\newtheorem{observation}{Observation}

\theoremstyle{remark}

\usepackage[hyphens]{url}
\usepackage[colorlinks,urlcolor=blue, citecolor=blue, menucolor=blue]{hyperref}

\theoremstyle{plain}

\providecommand{\keywords}[1]
{
  \small	
  \textbf{Keywords:} #1
}

\usepackage{algorithm,algpseudocode}

\usepackage{amsmath}

\allowdisplaybreaks

\makeatother
\title{Dynamic Pooling and Regional Participation in Deceased-Donor Organ Allocation\thanks{I am grateful to Yu Awaya, Michihiro Kandori, Fuhito Kojima, Shunya Noda, Bruno Strulovici, Yuichi Yamamoto, and all participants of the 18th Meeting of the Society for Social Choice and Welfare for helpful comments. All remaining errors are my own.}}
\author{Genta OKADA\thanks{Graduate School of Economics, the University of Tokyo. 7-3-1 Hongo, Bunkyo-ku, Tokyo, 113-0033, Japan. Email: \href{mailto:ogenta0628@g.ecc.u-tokyo.ac.jp}{\nolinkurl{ogenta0628@g.ecc.u-tokyo.ac.jp}}.}}
\hypersetup{pdftitle={Dynamic Pooling and Regional Participation in Deceased-Donor Organ Allocation}, pdfauthor={Genta OKADA}}
\date{September 16, 2026}
\begin{document}

\maketitle

\begin{abstract}
Moving from geographically fragmented to pooled waiting lists in deceased-donor organ transplantation can improve efficiency, but it raises concerns about regional fairness and participation incentives.
This paper studies Pareto gains from such transitions in a multi-class queueing model with impatient agents, perishable items, and a tractable homogeneous compatibility friction.
In contrast to standard approaches that emphasize static match quality or overlook regional incentives in dynamic contexts, this analysis identifies dynamic risk pooling as the primary source of Pareto improvements for regions.
Pooling reduces organ waste resulting from stochastic supply-demand fluctuations, temporal mismatches, and random compatibility outcomes.
Although this dynamic surplus theoretically enables Pareto improvements, a purely utilitarian policy that treats all patients equally across regions often fails to achieve these improvements when regions are asymmetric.
In these situations, utilitarian redistribution removes an excessive number of organs from high-supply regions and violates their participation constraints.
To address this conflict, it is demonstrated that the set of achievable utilities can be represented by mixtures of static priority index policies.
Building on this structure, the Constrained Pareto Frontier Search (CPFS) is proposed as a polynomial-time procedure that computes an approximately optimal allocation policy with arbitrarily small participation-constraint violations, implemented as a single priority list with a lottery at registration.
The results indicate that, by adjusting priority indices to ensure redistribution remains within the gains from pooled supply, demand, and compatibility risk, policymakers can achieve global efficiency while maintaining local participation incentives.
\end{abstract}

{\noindent\keywords{Deceased-donor organ allocation; Fragmentation; Dynamic matching; Regional participation; Queueing models; Constrained optimization}\par}

\newpage
\section{Introduction}
Organ allocation is a high-stakes problem characterized by scarcity and urgency, where systems must balance medical efficiency with equity across diverse patient groups.
Traditionally, allocation has relied on geographically defined boundaries that produce fragmented waiting lists.
This arrangement simplifies logistics, but it produces ``thin'' markets in which organs are discarded or allocated inefficiently because no suitable recipient is available locally when a donation occurs.
In response, several systems have moved toward broader, ``pooled'' allocation.
Pooling aims to create a ``thicker'' market, thereby reducing waste and improving match quality by expanding the candidate pool.
The United States, for example, has shifted from fixed geographic regions to a distance-based sharing framework \citep{cron2022new}, and Eurotransplant operates a centralized list for multiple countries \citep{langer2012history}.

The transition from fragmented to pooled markets is more than a logistical change; it raises distributional questions about regional fairness.
Pooling delivers aggregate efficiency gains, but it can leave individual regions worse off, particularly those with high local supply relative to demand.
Such regions may worry that a centralized system will divert locally donated organs to more congested areas and reduce the supply available to their own residents.
Recent policy debates illustrate the tension: in the United States, concerns have arisen about organs leaving particular geographic areas, such as the Midwest \citep{medpage_hearing_2025}.
These concerns give rise to a \textit{participation constraint}: for a pooled system to be sustainable, it should constitute a Pareto improvement over the status quo (autarky), so that no region is worse off than it would be in isolation.

A substantial literature applies dynamic matching and queueing models to organ allocation \citep{agarwal2021equilibrium, akbarpour2020thickness, ashlagi2025optimal, su2004patient, su2006recipient, zenios1999modeling, zenios2000dynamic}, but these papers typically optimize efficiency or fairness within a single centralized system.
Work that explicitly addresses participation constraints and incentives for pooling, on the other hand, has focused largely on static settings, particularly kidney exchange \citep{agarwal2019market, ashlagi2014free, biro2019generalized, hajaj2015strategy, klimentova2021fairness}.
The intersection of the two strands has received little attention: few papers analyze regional participation incentives inside the dynamic, stochastic framework of deceased-donor allocation.
The distinction matters because the central friction in deceased-donor allocation concerns \textit{when} patients and organs arrive, not only \textit{which patient} is matched with \textit{which organ}.
This gap is addressed using a multi-class queueing framework.
To maintain tractability while incorporating biological frictions, the analysis restricts attention to a single organ type and introduces a homogeneous random compatibility shock. Conditional on an organ arrival, each waiting patient is independently compatible with a common probability.
This specification rules out class-dependent compatibility matrices, which would generally break the static index structure that underpins our analysis, but still captures local mismatches and waste: each arriving organ may be incompatible with a given waiting patient, so organs can go unused even with a single organ type.
The framework lets us isolate the throughput effects of stochastic arrivals, patient impatience, random compatibility, and regional incentives that static matching models tend to obscure.

The dynamic perspective shows that pooling delivers two distinct benefits.
The first is static redistribution, which moves organs from high-supply to high-demand regions; although globally efficient, this redistribution is often zero-sum at the regional level and drives the conflict over participation.
The second, and more important, is dynamic risk pooling: in a fragmented system, organs are wasted whenever no compatible local recipient is available at the time of arrival, even though a compatible patient may be waiting in another region.
Pooling captures these otherwise discarded organs, generating a net throughput surplus that does not come at the expense of any region.
This dynamic surplus is what makes Pareto-improving policies feasible.

The tension between global efficiency and local participation is sharpest when regional supply-demand ratios differ substantially.
We show that a standard \textit{utilitarian policy}, which prioritizes patients purely by medical value, often fails to satisfy participation constraints: it draws organs from high-supply regions by more than those regions gain from pooling, leaving them worse off than under autarky.
To preserve participation, policy design must keep redistribution inside the bounds of the dynamic surplus.

This paper studies how to design such policies.
Sections \ref{sec:index_optimality} and \ref{sec:Pareto_frontier} show that the optimal policy for a weighted welfare maximization is a static \textit{index policy}, which assigns a fixed priority score to each patient class and allocates each arriving organ to the highest-index compatible waiting class, and that the achievable utility set is a convex polytope whose vertices correspond to these index policies.
Section \ref{sec:utilitarian_improvement} shows that, despite the enlarged achievable set, the utilitarian outcome need not Pareto-improve upon autarky when regions are asymmetric.
Section \ref{sec:algorithm} solves the participation-constrained problem in two steps.
The \textit{Constrained Pareto Frontier Search (CPFS)}, a polynomial-time algorithm that exploits the index-policy structure, computes a utility vector that is $\epsilon$-optimal among the achievable utility vectors that satisfy the regional participation constraints, as a convex combination of the utility vectors of index policies.
A second polynomial-time algorithm then converts this combination, to any desired accuracy, into one fixed priority list: the only randomization is a lottery drawn once at each patient's registration, which assigns the patient to a tier of the list, and organs are thereafter allocated by the list with no further randomization.

\section{Related Literature}

This paper lies at the intersection of dynamic stochastic allocation and participation in organ allocation.
A first strand studies deceased-donor allocation and waiting-list design as dynamic queueing or dynamic matching systems
\citep{agarwal2021equilibrium, su2004patient, su2006recipient, zenios1999modeling, zenios2000dynamic}.
Related waiting-list models study dynamic assignment more broadly
\citep{ashlagi2025optimal, bloch2017dynamic, che2021optimal, leshno2022dynamic}.
These papers are dynamic and stochastic, but they typically take the pooled or centralized system as given and do not impose regional participation constraints relative to autarky.

A second strand studies participation and strategic incentives in organ allocation, especially multi-agent kidney exchange
\citep{agarwal2019market, ashlagi2014free, biro2019generalized, hajaj2015strategy, klimentova2021fairness}.
These papers address participation incentives in organ allocation, but they are primarily static matching models in which hospitals or countries may withhold pairs or decide whether to join an exchange.
They therefore do not capture the dynamic deceased-donor problem in which organs and patients arrive over time, patients abandon, and organs are perishable.

Our work is also related to cooperative queueing and resource-pooling games in operations research.
This literature studies whether independent service providers can pool queues, capacity, or demand in a way that satisfies participation incentives, often through core or cost-allocation arguments
\citep{anily2010cooperation, karsten2015resource, liu2022incentives, ozen2011core}.
Closest in spirit, partial-pooling and stratified-pooling models show that full pooling need not benefit every participant, while restricted pooling can restore stability or mutual benefit
\citep{nandigam2019sharing, schlicher2020core}.
These papers are dynamic/stochastic and participation-based, but they are not organ-allocation models: they typically rely on transferable cost allocations in service, loss, or spare-parts systems rather than priority-based allocation of perishable organs subject to biological compatibility and patient abandonment.

We contribute by combining these elements in a deceased-donor setting.
We study a dynamic stochastic queueing model of pooled organ allocation with random compatibility, impatient patients, and regional autarky constraints.
Because monetary transfers are not the relevant policy instrument, participation is enforced through the allocation rule itself.
Technically, our index-policy result also relates to multi-class queueing with abandonment.
Comparable optimality results appear in \citet{down2011dynamic} for two-class systems and in \citet{ansari2025scheduling} for truncated state spaces. 
The differences are that we work in an untruncated setting with multiple classes, and service opportunities are generated by a common organ-arrival process and filtered through homogeneous random compatibility.

\section{Model}

We formulate the problem as a multi-class queueing system with abandonment, in the spirit of standard models (e.g., \citet{ansari2025scheduling}), but with multiple regions, a pooled organ-arrival process, and a compatibility friction.

\subsection{Setup and Notation}
We consider a system with a finite set of regions, $\mathcal{R}$, and a finite set of patient types, $\Theta$. A patient class is defined as a region-type pair $i=(r,\theta)\in\mathcal{R}\times\Theta$. There is a single organ type.

Patients of class $i$ arrive according to an independent Poisson process with rate $\lambda_i>0$ and join a class-specific queue. For $i=(r,\theta)$, we also write $\lambda^r_\theta:=\lambda_{(r,\theta)}$ and $\lambda^r:=(\lambda^r_\theta)_{\theta\in\Theta}$. Patients are impatient; the time until a class-$i$ patient abandons the queue is exponentially distributed with rate $d_i>0$. For most of our analysis, we assume abandonment rates to be uniform across classes.

Organs arrive from region $r$ according to a Poisson process with rate $\mu^r>0$, yielding a total arrival rate $\mu=\sum_{r\in \mathcal R}\mu^r$. Organs are perishable and must be allocated immediately upon arrival or discarded.

To incorporate incompatibility while preserving tractability, we assume an ex-ante homogeneous random compatibility structure. Conditional on an organ arrival and the current queue state, each waiting patient is independently compatible with the arriving organ with a common probability $\rho\in(0,1]$. Compatibility realizations are independent across organs and across patients. Thus, incompatibility is organ-specific and idiosyncratic, but ex-ante symmetric across classes.

A class $i=(r,\theta)$ patient receives utility $v_\theta>0$ upon receiving an organ and zero utility if they abandon the queue. For notational simplicity, we also write $v_i:=v_\theta$ for class $i=(r,\theta)$.

\subsection{CTMDP Formulation}
The state space is $S=\mathbb{Z}_+^K$, where $K=|\mathcal{R}||\Theta|$. It consists of vectors $q=(q_i)_i$, where $q_i$ denotes the number of waiting patients of class $i$. We define $N(q)=\sum_i q_i$ as the total number of waiting patients and $J(q)=\{i\mid q_i>0\}$ as the set of classes with waiting patients. The notation $e_i$ represents the standard unit vector for class $i$. The empty state is denoted by $q_0=(0,\dots,0)$, and $\mathbf{1}_{c}$ denotes the indicator of event $c$.

Let
$$
g_\rho(n)=1-(1-\rho)^{n}
$$
denote the probability that, among $n$ waiting patients, at least one is compatible with an arriving organ.
Because patient-level compatibilities are independent, the random set of compatible classes is a subset $A\subseteq J(q)$ with probability
$$
p_q(A)=\prod_{i\in A} g_\rho(q_i)\prod_{j\in J(q)\setminus A}\bigl(1-g_\rho(q_j)\bigr).
$$

A stationary policy $\sigma$ specifies allocation probabilities after both the queue state and the compatibility realization are observed. Formally, for each state $q$ and each realized compatible set $A\subseteq J(q)$, the policy specifies
$$
\sigma(q,A)=\bigl(\sigma_i(q,A)\bigr)_i,
$$
where $\sigma_i(q,A)$ is the probability of allocating the arriving organ to class $i$.
The policy must satisfy
$$
\sigma_i(q,A)=0 \quad \text{for all } i\notin A,
\qquad
\sum_{i\in A}\sigma_i(q,A)\le 1
$$
for all $q\in S$ and all $A\subseteq J(q)$.
If $A=\emptyset$, the organ must be discarded.%
\footnote{\label{fn:within_class}Allocation rules within a class (e.g., FCFS or LCFS among compatible patients) affect individual waiting times but not the class-level aggregate utility in this Markovian framework, so we abstract from within-class allocation details. Although organ origin is not explicit in the state, a policy equivalent to autarky (allocating organs only to their region of origin) is feasible because organ arrivals are independent Poisson processes: an arriving organ is assigned to region $r$ with probability $\mu^r / \mu$.}

For later use, define the ex-ante probability that an organ arriving in state $q$ is allocated to class $i$ under policy $\sigma$ by
$$
\bar{\sigma}_i(q)=\sum_{A\subseteq J(q):\, i\in A} p_q(A)\sigma_i(q,A).
$$
The queue-length process alone is then a continuous-time Markov decision process (CTMDP), because compatibility shocks are i.i.d. across organ arrivals and can be integrated out conditional on $q$.

The transition rates $\nu^\sigma(q,q')$ from state $q$ to state $q'$ under policy $\sigma$ are
$$
\nu^\sigma(q,q')=
\begin{cases}
\lambda_i & \text{if } q'=q+e_i,\\
d_i q_i+\mu \bar{\sigma}_i(q) & \text{if } q'=q-e_i,\\
-Q^\sigma(q) & \text{if } q'=q,\\
0 & \text{otherwise,}
\end{cases}
$$
where
$$
Q^\sigma(q)=\sum_i \lambda_i+\sum_i d_i q_i+\mu\sum_i \bar{\sigma}_i(q)
$$
is the total rate of state-changing transitions out of state $q$.

\subsection{Ergodicity and Stationary Distribution}\label{subsec:stationary_dist}
The process is ergodic because of the abandonment process.

\begin{lemma}\label{lmm:ergodic}
    The Markov process induced by any stationary policy $\sigma$ is ergodic.
\end{lemma}
The proof is in Appendix~\ref{app:proof_ergodic}.

Because the process is ergodic, there exists a unique stationary distribution $\pi^{\sigma}$ for any policy $\sigma$ \citep{bremaud2013markov}. This distribution is the solution to the global balance equations:
$$
\sum_{q' \in S} \pi^{\sigma}(q') \nu^{\sigma}(q', q) = 0
\qquad \forall q \in S,
$$
or equivalently,
$$
\pi^{\sigma}(q) Q^\sigma(q)=\sum_{q'\neq q}\pi^{\sigma}(q')\nu^{\sigma}(q',q).
$$

Furthermore, by the ergodic theorem for Markov processes, the long-run time average of any function $f(q)$ that has a finite expectation converges almost surely to its expected value under the stationary distribution $\pi^\sigma$:
$$
\lim_{T\to\infty}\frac{1}{T}\int_0^T f(q(t))\,dt
=
E_{\pi^\sigma}[f(q)]
\qquad \text{a.s.}
$$
This property allows us to evaluate performance metrics, such as long-run average costs or utilities, by calculating expectations with respect to $\pi^\sigma$.

\subsection{Average Utility and Flow Conservation}\label{subsec:flow_conv}
The average utility $u_i$ per arriving class-$i$ patient is defined as the long-run average utility accumulated by class $i$, normalized by the long-run average number of class-$i$ arrivals:
$$
u_i=
\frac{
\lim_{T\to\infty}\frac{1}{T}E\left[\sum_{n=1}^{N_O(T)} v_i \mathbf{1}_{a_n=i}\right]
}{
\lim_{T\to\infty}\frac{1}{T}E\left[N_P^i(T)\right]
}.
$$
Here, $N_O(T)$ and $N_P^i(T)$ are the counting processes for organ arrivals and class-$i$ patient arrivals by time $T$, respectively, and $a_n\in(\mathcal{R}\times\Theta)\cup\{\emptyset\}$ denotes the class selected at the $n$-th organ-arrival epoch, with $a_n=\emptyset$ if the organ is discarded.

By the ergodic theorem, the numerator converges to the stationary reward rate $v_i E_{\pi^\sigma}\bigl[\mu \bar{\sigma}_i(q)\bigr]$, and the denominator converges to the arrival rate $\lambda_i$.
Thus,
$$
u_i=\frac{v_i E_{\pi^\sigma}\bigl[\mu \bar{\sigma}_i(q)\bigr]}{\lambda_i}.
$$

We can derive a flow conservation law from the global balance equations by taking the expected drift of $q_i$. In steady state, the inflow rate equals the outflow rate (allocation plus abandonment) for each class $i$:
$$
\lambda_i=E_{\pi^\sigma}\bigl[\mu \bar{\sigma}_i(q)\bigr]+d_i E_{\pi^\sigma}[q_i].
$$
Substituting $E_{\pi^\sigma}\bigl[\mu \bar{\sigma}_i(q)\bigr]=\lambda_i-d_iL_i^\sigma$, where $L_i^\sigma=E_{\pi^\sigma}[q_i]$ is the expected queue length, into the utility expression yields
$$
u_i
=
v_i-\frac{d_i v_i}{\lambda_i}L_i^\sigma.
$$
This relationship shows that maximizing average utility is equivalent to minimizing the weighted sum of stationary queue lengths, i.e., the holding cost.

\subsection{Scope of the Modeling Assumptions}\label{subsec:assumptions_scope}
Three restrictions carry the analysis: a single organ type, ex-ante homogeneous compatibility, and, from Section \ref{sec:Pareto_frontier} onward, a common abandonment rate.
Actual transplantation satisfies none of them exactly; compatibility varies with blood type, sensitization level, organ quality and donor characteristics, and impatience varies with diagnosis and comorbidity.
The first two make organs interchangeable once the compatibility draw is realized, which leaves the queue-length vector a controlled Markov process with static priority rankings among its optimal policies; the third establishes the optimality of index policies for all welfare weights, enabling our characterization of achievable utility vectors.
Relax any of them, and we do not know how to recover either the index characterization or the algorithm.
One conclusion survives regardless: a pooled system can always reproduce the autarky allocation and, on top of that, reuse the organs autarky discards, so the dynamic surplus that makes a Pareto improvement possible does not depend on homogeneity.
What heterogeneity removes is the tractable description of the achievable frontier, and the algorithm built on it.

Given these restrictions, the model is best read as a benchmark rather than as a description of an allocation system: it approximates one aspect of allocation, the organs lost because no compatible recipient happens to be waiting nearby when a donation occurs.
One setting under which homogeneous compatibility is defensible is allocation organized within medical strata: if offers are made within a stratum defined by blood type, sensitization level, and organ quality, the single organ type of the model is one such stratum, and $\rho$ is the probability that an organ arising there suits a patient waiting there.

\section{Optimality of the Index Policy}\label{sec:index_optimality}
To reconcile efficiency with regional constraints, we first need to characterize which policies maximize social welfare.
This section shows that, for any set of nonnegative welfare weights, the optimal allocation rule retains a simple form even under homogeneous random incompatibility: a static \textit{index policy}.
Beyond providing a tractable solution to the optimization problem, this result is the building block for the Pareto frontier characterization in the next section.

\subsection{Problem Formulation: Social Welfare Maximization}
We maximize a weighted social welfare function,
$$
W=\sum_i w_i u_i,
$$
with nonnegative weights $w_i\ge 0$.
Using the flow conservation result from Section \ref{subsec:flow_conv}, this is equivalent to minimizing a weighted holding cost:
$$
W
=
\sum_i w_i\left(v_i-\frac{d_i v_i}{\lambda_i}L_i^\sigma\right)
=
\left(\sum_i w_i v_i\right)-\sum_i \frac{w_i d_i v_i}{\lambda_i}L_i^\sigma.
$$
Since the first term is policy-independent, maximizing $W$ is equivalent to minimizing the long-run average cost
$$
C(\sigma)
=
\sum_i \frac{w_i d_i v_i}{\lambda_i}L_i^\sigma
=
\sum_i b_i d_i L_i^\sigma
=
E_{\pi^\sigma}\left[\sum_i b_i d_i q_i\right],
$$
where the weighted priority parameter for class $i$ is
$$
b_i=\frac{w_i v_i}{\lambda_i}.
$$

\subsection{Index Policy and Its Optimality}
Without loss of generality, reindex classes as $i=1,\dots,K$ so that
$$
b_1\ge b_2\ge \cdots \ge b_K.
$$
Ties are broken in favor of the class with the higher abandonment rate, so that $d_k\ge d_{k+1}$ whenever $b_k=b_{k+1}$, or arbitrarily if both parameters are equal.

\begin{definition}[Index Policy]
Fix a priority ordering $1, \dots, K$.
The corresponding index policy $\sigma^*$ is the deterministic stationary policy that, whenever an organ arrives in state $q$, and the realized set of compatible classes is $A\subseteq J(q)$, allocates the organ to the highest-priority class in $A$:
$$
\sigma_i^*(q,A)=1
\quad \text{if } A\neq\emptyset \text{ and } i=\min A.
$$
If $A=\emptyset$, the organ is discarded.
\end{definition}

Optimality holds under the following assumption on impatience:
\begin{assumption}\label{asmp:d_order}
For all classes $i,j$, if $i<j$ (implying $b_i\ge b_j$), then $d_i\ge d_j$.
\end{assumption}
This assumption holds when higher-medical-value classes have shorter or equal expected lifetimes without transplantation compared to lower-medical-value classes.
It holds trivially when abandonment rates are identical across classes.

\begin{theorem}\label{thm:index_optimality}
Under Assumption \ref{asmp:d_order}, the index policy with priority ordering $1,\dots, K$ minimizes $C(\sigma)$ over all stationary policies that may condition on both the queue state and the realized compatible set.
\end{theorem}
The proof is in Appendix~\ref{app:proof_indexopt}. It builds on the frameworks of \citet{down2011dynamic} and \citet{ansari2025scheduling}, but replaces the service term in the Bellman operator with its expectation over compatibility realizations.

This simplification matters for the rest of the paper.
Instead of searching over a vast space of state- and compatibility-dependent policies, we can restrict attention to static index rankings and their mixtures.


\section{Achievable Region and Pareto Frontier}\label{sec:Pareto_frontier}
Section \ref{sec:index_optimality} established that static index policies are optimal for maximizing weighted social welfare under Assumption \ref{asmp:d_order}.
The result has a useful geometric implication: it lets us characterize the entire set of feasible utility outcomes.
In this section, we describe the shape of this achievable utility set and define the benchmark for regional participation, the autarky utility.

For the analysis of the achievable region and regional utilities in the following sections, we specialize the model by assuming that all classes share the same abandonment rate, which yields a tractable one-dimensional representation of cumulative throughput under any fixed priority ordering.
\begin{assumption}\label{asmp:same_d}
For all classes $i$, $d_i=d$ for a common constant $d>0$.
\end{assumption}

Under Assumption \ref{asmp:same_d}, Assumption \ref{asmp:d_order} (and thus Theorem \ref{thm:index_optimality}) trivially holds.

\subsection{The Achievable Utility Set}
Let $\Sigma$ be the set of all stationary policies. Each $\sigma\in\Sigma$ induces a unique stationary distribution $\pi^\sigma$ and a utility vector $u(\sigma)=(u_i(\sigma))_i$, where $u_i(\sigma)$ is the average utility for class $i$ under policy $\sigma$. Define the achievable set as
$$
\mathcal U=\{u(\sigma)\mid \sigma\in\Sigma\}\subset \mathbb R_+^K.
$$

Using the result from Section \ref{sec:index_optimality}, we can characterize $\mathcal U$ via index policies.
Let $\Sigma_I=\{\sigma_1^I,\ldots,\sigma_{K!}^I\}$ denote the set of all static index policies induced by the $K!$ priority orderings, where each such policy allocates an arriving organ to the highest-priority compatible class.

\begin{proposition}\label{prop:achievable_set}
$\mathcal U$ is the set of nonnegative utility vectors dominated by some convex combination of utility vectors generated by static index policies. Formally,
$$
\mathcal U
=
\left\{
u\in\mathbb R_+^K
\;\middle|\;
\exists \alpha\in[0,1]^{K!}
\text{ such that }
u_i\le \sum_{m=1}^{K!}\alpha_m u_i(\sigma_m^I)
\text{ for all } i,
\ \sum_{m=1}^{K!}\alpha_m=1
\right\}.
$$
\end{proposition}
The proof is in Appendix~\ref{app:proof_prop_achievable}.

\subsection{Utility under Index Policies}\label{subsec:index_util}
Under uniform abandonment, the queueing dynamics aggregate cleanly across higher-priority classes.
Consider a birth--death chain with birth rate $\lambda$, abandonment rate $d$, and state-$n$ service rate $\mu g_\rho(n)$.
Its stationary distribution is
$$
\Pi_0^\rho(\lambda,\mu)
=
\left[1+\sum_{n=1}^\infty \prod_{m=1}^n \frac{\lambda}{dm+\mu g_\rho(m)}\right]^{-1},
\qquad
\Pi_n^\rho(\lambda,\mu)
=
\Pi_0^\rho(\lambda,\mu)\prod_{m=1}^n \frac{\lambda}{dm+\mu g_\rho(m)},
$$
and the steady-state allocation rate, which we call the \emph{cumulative throughput function}, is
$$
\mathcal T^\rho(\lambda,\mu)
=
\mu\sum_{n=1}^\infty \Pi_n^\rho(\lambda,\mu)\, g_\rho(n).
$$
Under any static-priority rule, the aggregate queue length of the top-$k$ priority classes evolves as this exact birth--death chain with $\lambda=\Lambda_k$ (the sum of the top-$k$ arrival rates), yielding an explicit one-dimensional representation of per-class utility.

\begin{lemma}\label{lmm:index_util}
Suppose classes are indexed $k=1,\dots, K$ in decreasing order of priority, and let $\Lambda_k=\sum_{\ell=1}^k \lambda_\ell$ with $\Lambda_0=0$. Then the average utility for class $k$ is
$$
u_k
=
\frac{v_k}{\lambda_k}
\left(
\mathcal T^\rho(\Lambda_k,\mu)-\mathcal T^\rho(\Lambda_{k-1},\mu)
\right).
$$
\end{lemma}
The proof is in Appendix~\ref{app:proof_index_util}.

Using these class-level utilities, we compute the aggregate regional utility $u^r$.
The aggregate regional utility $u^r$ is the weighted average of utilities of all patient types within region $r$:
$$
u^r
=
\frac{1}{\Lambda^r}\sum_{\theta\in\Theta}\lambda_{(r,\theta)} u_{(r,\theta)},
$$
where
$
\Lambda^r=\sum_{\theta}\lambda_{(r,\theta)}.
$
The Pareto frontier of regional utilities $(u^r)_{r\in\mathcal R}$ is formed by the convex hull of regional utility vectors corresponding to different priority orderings.

\subsection{The Autarky Benchmark}\label{subsec:separate_util}
To define the participation constraints, we consider the autarky scenario, where each region $r$ operates its own waiting list using only its local organ supply $\mu^r$.
Consistent with local optimality, we assume that under autarky each region maximizes its own region-average utility.
Let $u_{sep}=(u_{sep}^r)_{r\in\mathcal R}$ denote the vector of these maximum autarky utilities.

As before, reindex types as $\theta=1,2,\ldots,|\Theta|$ so that
$$
v_1\ge v_2\ge \cdots \ge v_{|\Theta|}.
$$
From Section \ref{sec:index_optimality}, the optimal policy for maximizing $u^r$ in a separate regional system is an index policy.
Maximizing $u^r$ is equivalent to maximizing
$$
\sum_{\theta}\lambda_{(r,\theta)} u_{(r,\theta)},
$$
which yields the priority index $b_i=v_\theta$.
Therefore, $u_{sep}^r$ equals the utility under the type-based priority policy within region $r$ with organ-arrival rate $\mu^r$.
By Lemma \ref{lmm:index_util},
$$
u_{sep}^r
=
\frac{1}{\Lambda^r}
\sum_{\theta=1}^{|\Theta|}
v_\theta
\left(
\mathcal T^\rho(\Lambda^r_{\theta},\mu^r)
-
\mathcal T^\rho(\Lambda^r_{\theta-1},\mu^r)
\right),
$$
where
$$
\Lambda^r_{\theta}=\sum_{\theta'=1}^{\theta}\lambda_{\theta'}^r,
\qquad
\Lambda^r_{0}=0.
$$

The same lemma gives the class-level benchmarks behind this average. Writing $u_{sep}^{(r,\theta)}$ for the autarky utility of class $(r,\theta)$,
$$
u_{sep}^{(r,\theta)}
=
\frac{v_\theta}{\lambda_{(r,\theta)}}
\left(
\mathcal T^\rho(\Lambda^r_{\theta},\mu^r)
-
\mathcal T^\rho(\Lambda^r_{\theta-1},\mu^r)
\right),
\qquad
u_{sep}^r=\frac{1}{\Lambda^r}\sum_{\theta}\lambda_{(r,\theta)}u_{sep}^{(r,\theta)}.
$$
The regional benchmark is thus the arrival-weighted average of the class-level ones. We use the regional benchmarks throughout, and return to the class-level ones in Section \ref{seq:experiments}.

\subsection{Guarantee of Pareto Improvement over Autarky}\label{subsec:pareto_guarantee}
Relative to the regional utilities under separation, $u_{sep}$, the pooled system can attain strictly higher social welfare.

\begin{observation}\label{obs:pareto_frontier_improvement}
Assume $|\mathcal R|\ge 2$. For any positive weight vector $w\in\mathbb R_{++}^{|\mathcal R|}$, the maximum achievable weighted social welfare under pooling strictly exceeds the welfare from separate systems:
$$
\max_{\sigma\in\Sigma}\sum_{r\in\mathcal R} w^r u^r
>
\sum_{r\in\mathcal R} w^r u_{sep}^r.
$$

Moreover, Pareto improvement over the autarky outcome is always possible: there exists $u \in \mathcal U$ such that
$$u^r> u_{sep}^r \text{ for all } r\in\mathcal R.$$
\end{observation}
The proof is in Appendix~\ref{app:proof_obs1}. 
It constructs the improving policies explicitly: for each region $s$ there is a pooled policy that leaves every other region exactly at its autarky utility and makes $s$ strictly better off.
Mixing these policies uniformly across regions makes every region strictly better off at once, so a pooled policy that strictly Pareto improves on autarky always exists.
The source of strict Pareto gains is dynamic risk pooling: a pooled system can replicate the autarky policy allocation and, in addition, use organs that would otherwise be wasted under local operation when no compatible local recipient is available upon arrival.%
\footnote{Although transplant waiting lists are typically overloaded in aggregate, the queue for a specific compatible patient type in a specific locality can be effectively empty when an organ arrives, because of biological incompatibility or logistical constraints. Empirically, roughly 20\% of recovered kidneys in the U.S.\ are discarded despite long waiting lists \citep{israni2025optn}, so these local frictions are not minor.}

\section{The Limits of Utilitarian Policies in Asymmetric Regions}\label{sec:utilitarian_improvement}

Given the potential for Pareto gains, a natural choice is the \textit{utilitarian policy}, which maximizes total patient utility
$
\sum_i \lambda_i u_i
$
without regard to region.
The approach is globally efficient but fragile: it often fails to satisfy regional participation constraints when regions are asymmetric in supply and demand.

\subsection{Utilitarian Policies in a Multi-Region Setting}
The utilitarian benchmark maximizes the total utility rate of all patients and weights each arriving patient equally.
This sets welfare weights equal to arrival rates, $w_i=\lambda_i$, so the priority index reduces to
$$
b_i=\frac{w_i v_i}{\lambda_i}=v_i.
$$
Under utilitarianism, the optimal policy ranks classes purely by patient value and, when an organ arrives, allocates it to the highest-value compatible waiting class.

The utilitarian objective leaves the relative priority across regions for a given patient type $\theta$ undefined.
Patients of type $\theta$ in different regions share the same index $v_\theta$, so the planner may pick any tie-breaking rule among compatible classes of that type without changing total social welfare.
The set of achievable utilitarian outcomes therefore forms a face of $\mathcal U$ whose vertices correspond to deterministic tie-breaking rules, i.e., permutations of regions within each type block.

\subsection{The Failure of Utilitarianism in Asymmetric Regions}
Pooling Pareto dominates autarky, yet this does not ensure that a utilitarian policy improves on autarky for every region.

We illustrate the point with two-region examples. Even when the planner selects the tie-breaking rule most favorable to a given region, the participation constraints can fail.

\begin{figure}[tbp]
\centering
\begin{minipage}[b]{0.47\columnwidth}
    \centering
    \includegraphics[width=\columnwidth]{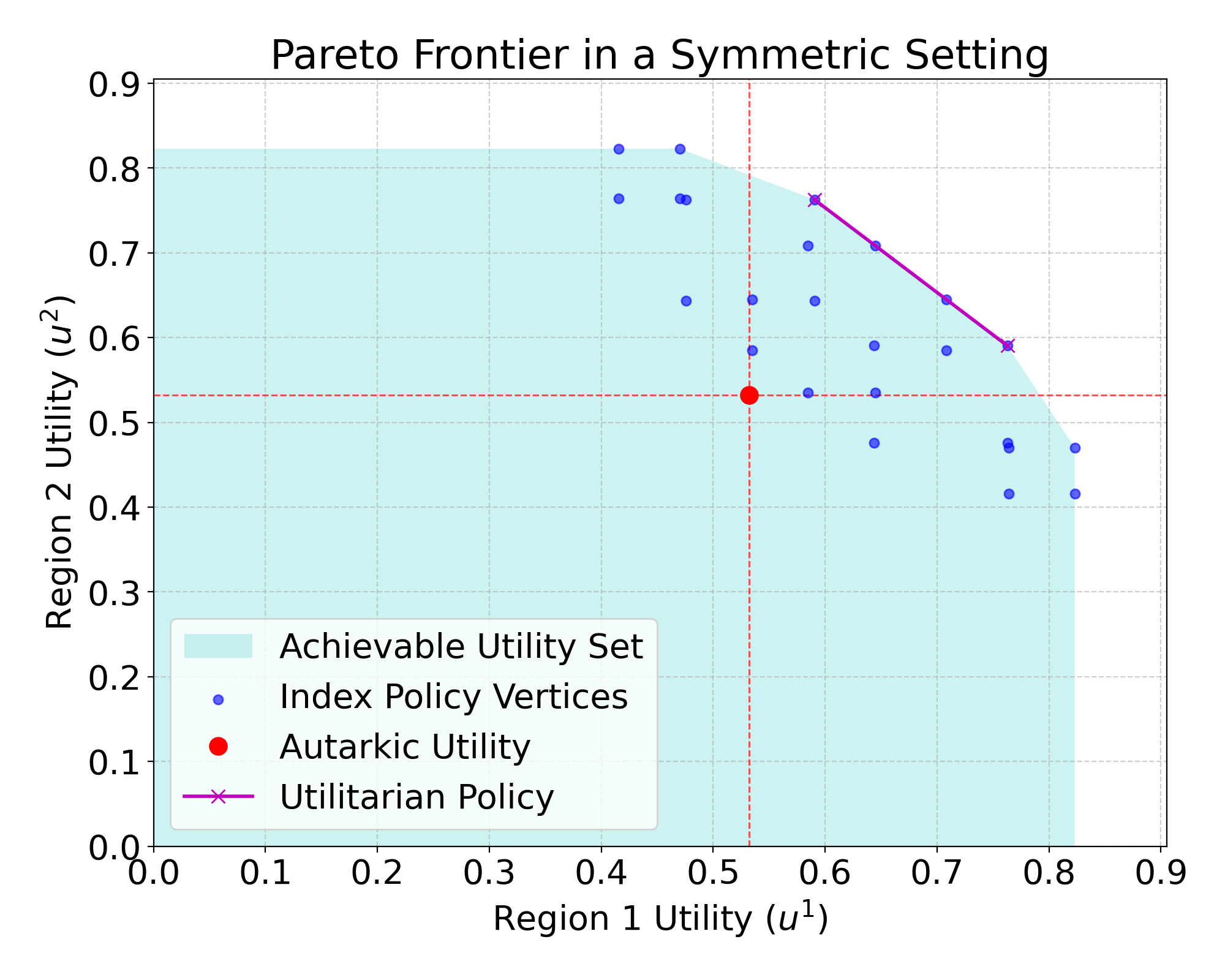}
    \caption{Pareto Frontier in a Symmetric Setting}
    \label{fig:symmetric_frontier}
\end{minipage}
\hspace{0.04\columnwidth}
\begin{minipage}[b]{0.47\columnwidth}
    \centering
    \includegraphics[width=\columnwidth]{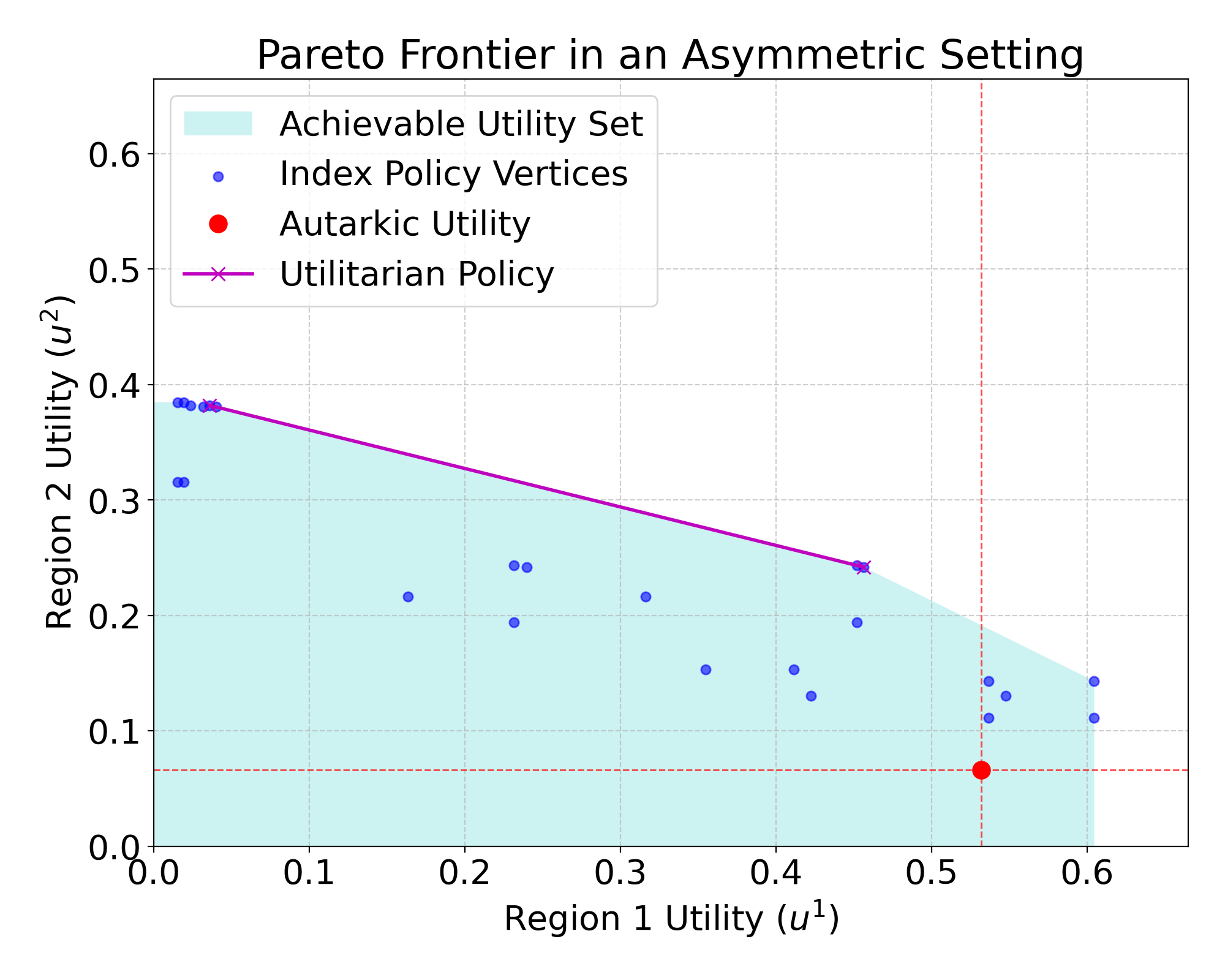}
    \caption{Pareto Frontier in an Asymmetric Setting}
    \label{fig:asymmetric_frontier}
\end{minipage}
\end{figure}

\paragraph{Symmetric Case:}
When regions have similar supply-demand profiles, the flow of organs is balanced.
Pooling then operates mainly as risk sharing, smoothing stochastic fluctuations in arrivals.
\begin{example}\label{ex:symmetric}
Consider a system with $|\mathcal{R}|=2$ regions, $|\Theta|=2$ types, values $v_1=2, v_2=1$, compatibility probability $\rho=0.9$, and abandonment rate $d=1$.
Region $1$: $\mu^1 = 1$, $\lambda_1^1 = 1$, $\lambda_2^1 = 1$; 
Region $2$: $\mu^2 = 1$, $\lambda_1^2 = 1$, $\lambda_2^2 = 1$.
\end{example}
Utilitarianism poses no problem here.
As Figure \ref{fig:symmetric_frontier} shows, the utilitarian point lies well inside the Pareto-improving region (the upper-right quadrant relative to autarky).

\paragraph{Asymmetric Case:}
Consider the following counterexample in which the regions differ structurally.
\begin{example}[Asymmetric Failure]\label{ex:asymmetric}
    Let $|\mathcal{R}|=2, |\Theta|=2, v_1=2, v_2=1$, $\rho=0.9$, and $d=1$.
    \begin{itemize}
        \item \textbf{Region 1 (High Supply/Low Demand):} $\mu^1 = 1.0$, arrivals $\lambda^1 = (1, 1)$.
        \item \textbf{Region 2 (Low Supply/High Demand):} $\mu^2 = 0.2$, arrivals $\lambda^2 = (5, 1)$.
    \end{itemize}
    In autarky, Region 1 patients enjoy a high service rate because of abundant local supply.
    Under a pooled utilitarian policy, the large pool of high-value (Type 1) patients in Region 2 competes for Region 1's organs.
    Type 1 patients in Region 2 outrank Type 2 patients in Region 1, so organs flow from Region 1 to Region 2.
    Region 1's average utility then drops well below its autarky level.
    As Figure \ref{fig:asymmetric_frontier} illustrates, the utilitarian solution lies on the Pareto frontier yet violates Region 1's participation constraint.
\end{example}

\subsection{When Does Utilitarianism Satisfy Participation Constraints?}
Since utilitarian policies generally fail to guarantee Pareto improvements, it is useful to identify conditions under which a utilitarian outcome satisfies every region's participation constraint.
We give a sufficient condition based on \textit{proportional allocation}.
The intuition is straightforward: if the planner can distribute total pooled throughput to each region in proportion to its share of total supply ($\mu^r / \mu$), and if no region is too congested relative to the global average, then a Pareto improvement is guaranteed.

Under any utilitarian policy, the aggregate cumulative throughput for patient types $1,\dots,\theta$ is invariant to the tie-breaking rule across regions.
Indeed, let
$$
\Lambda_{\theta}=\sum_{r\in\mathcal R}\Lambda^r_{\theta}.
$$
Because all types $1,\dots,\theta$ are prioritized above types $\theta+1,\dots,|\Theta|$, the total number of waiting patients of types $1,\dots,\theta$ evolves as a one-dimensional birth--death process with birth rate $\Lambda_{\theta}$ and death rate $dn+\mu g_\rho(n)$.
Therefore, the aggregate utilitarian cumulative throughput is
$$
\mathcal T_{\mathrm{total},\theta}=\mathcal T^\rho(\Lambda_{\theta},\mu),
$$
where $\mathcal T^\rho$ was defined in Section \ref{subsec:index_util}.

Likewise, under autarky, the cumulative throughput of types $1,\dots,\theta$ in region $r$ is
$$
\mathcal T^r_{\mathrm{sep},\theta}
=
\mathcal T^\rho(\Lambda^r_{\theta},\mu^r).
$$

The following condition provides a nontrivial sufficient criterion under which a proportional utilitarian allocation Pareto improves upon autarky.

\begin{proposition}[Sufficient Condition for Utilitarian Pareto Improvement]\label{prop:suff_cond_pareto}
There exists a utilitarian policy that Pareto improves upon the autarky benchmark $u_{sep}$ if, for all types $\theta\in\Theta$ and all regions $r\in\mathcal R$:
\begin{enumerate}
    \item \textbf{Effective Congestion Condition:}
    $$
    \frac{\Lambda^r_{\theta}}{\Lambda_{\theta}}
    \le
    \frac{d+\rho\mu^r}{d+\rho\mu}.
    $$
    \item \textbf{Feasibility Condition:} The proportional target cumulative throughput vector
    $$
    \mathcal T^r_{\mathrm{target},\theta}
    :=
    \frac{\mu^r}{\mu}\mathcal T_{\mathrm{total},\theta}
    $$
    is achievable by some convex combination of utilitarian tie-breaking rules.
\end{enumerate}
\end{proposition}

The proof is in Appendix~\ref{app:proof_suff_cond}.
When both conditions hold, a utilitarian policy that implements the proportional target throughput Pareto dominates autarky.
The intuition is that when no region is too demand-heavy relative to its supply contribution, the gains from pooling stochastic supply, demand, and compatibility risk suffice to compensate all regions.

Example \ref{ex:asymmetric} shows that these conditions need not hold in asymmetric environments.
When they fail, policymakers cannot rely on standard efficiency maximization alone and must search explicitly for a policy that enforces participation constraints.
This motivates the constrained optimization approach developed in the next section.

\section{Pareto-Constrained Optimization: Algorithm and Implementation}\label{sec:algorithm}
As the previous section showed, standard utilitarian policies may violate participation constraints in asymmetric environments.
Policymakers therefore face a nontrivial constrained optimization problem: maximize social welfare subject to the constraint that no region is worse off than under autarky.
The feasible set has a factorial number of priority permutations, so brute-force approaches that enumerate the entire achievable utility set are computationally intractable.

This section solves the problem in two steps, each in polynomial time.
First, Sections \ref{subsec:formulation}--\ref{subsec:cpfs} develop the \textit{Constrained Pareto Frontier Search (CPFS)} algorithm, which exploits the optimality of index policies established in Section \ref{sec:index_optimality} to find an approximately optimal utility vector together with a small set of priority rankings and weights whose mixture attains it.
Second, a mixture of rankings is not itself a rule that a transplant system can run, and Section \ref{subsec:implementation} gives an algorithm that replaces it, to any desired accuracy, by a single priority list in which some patient classes are split by a lottery drawn once at registration.
Section \ref{seq:experiments} illustrates both steps numerically.

\subsection{Problem Formulation}\label{subsec:formulation}

We work with a general formulation that nests the multi-region participation problem.
The objective is to maximize a weighted social welfare function with weights $w \in \mathbb{R}^{K}_{+}$, subject to $L\ge 1$ linear inequality constraints\footnote{The case with no constraints can be trivially solved by Theorem \ref{thm:index_optimality}.} defined by a matrix $M \in \mathbb{R}^{L \times K}_+$ and a vector $c \in \mathbb{R}^L_{+}$.

The optimization problem $(P)$ is defined as:
\begin{equation}
(P) \quad
\begin{aligned}
\max_{\sigma \in \Sigma} \quad & w^\top u(\sigma) \\
\textrm{s.t.} \quad & M u(\sigma) \ge c.
\end{aligned}
\end{equation}
We assume throughout that the constraints admit an interior point, which is slightly more than feasibility.

\begin{assumption}\label{asmp:slater}
There exists an achievable utility vector $\bar u\in\mathcal U$ and a constant $s_0>0$ such that $M\bar u\ge c+s_0\mathbf 1$.
\end{assumption}

Regional participation constraints $u^r \ge u_{sep}^r$ are a special case of this formulation, with the rows of $M$ aggregating class utilities into regional utilities, and they always admit an interior point when multiple regions exist.
Observation \ref{obs:pareto_frontier_improvement} shows that when $|\mathcal R| \ge 2$, there exists $\bar u \in \mathcal U$ which satisfies $\bar u^r>u^r_{sep}$ for every $r$; hence, for regional participation constraints, Assumption \ref{asmp:slater} holds with $s_0=\min_r(\bar u^r-u^r_{sep})$.
The formulation also accommodates other fairness criteria, such as minimum-utility guarantees for specific patient classes.
For two leading scenarios, $M$ and $c$ take the following form:

\begin{enumerate}
    \item \textbf{Regional Participation Constraints (Autarky):}
    To ensure that a region $r$ achieves at least its autarky utility level $u_{sep}^r$, we impose the constraint $u^r(\sigma) \ge u_{sep}^r$.
    Recall that the aggregate utility for region $r$ is a weighted sum of the utilities of its constituent classes:
    \[ u^r(\sigma) = \sum_{\theta \in \Theta} \frac{\lambda_{(r,\theta)}}{\Lambda^r} u_{(r,\theta)}(\sigma), \]
    where $\lambda_{(r,\theta)}$ is the arrival rate of type $\theta$ patients in region $r$, and $\Lambda^r = \sum_{\theta} \lambda_{(r,\theta)}$ is the total arrival rate for the region.
    
    This constraint corresponds to row $\ell$ of matrix $M$:
    \[
    M_{\ell, k} = 
    \begin{cases} 
    \frac{\lambda_{(r,\theta)}}{\Lambda^r} & \text{if } k \text{ corresponds to class } (r,\theta), \\
    0 & \text{otherwise.}
    \end{cases}
    \]
    The corresponding element of the vector $c$ is set to $c_\ell = u_{sep}^r$.

    \item \textbf{Minimum Utility Guarantees (Class-Specific):}
    To protect vulnerable patient groups, a policymaker might require that the utility of a specific class $k^* = (r^*, \theta^*)$ does not fall below a threshold $\underline{u}$.
    This is encoded by a row $\ell$ in $M$ that is a standard basis vector $e_{k^*}^\top$:
    \[
    M_{\ell, k} = 
    \begin{cases} 
    1 & \text{if } k = k^*, \\
    0 & \text{otherwise.}
    \end{cases}
    \]
    The corresponding element of $c$ is $c_\ell = \underline{u}$.
\end{enumerate}

\subsection{The Dual Approach and the Subgradient Oracle via Index Policies}
The primal problem $(P)$ is a linear program over the polytope of achievable utilities $\mathcal{U}$.
Characterizing $\mathcal{U}$ explicitly is computationally demanding because the polytope has a factorial number of vertices.
We sidestep this curse of dimensionality with a Lagrangian dual approach paired with an efficient subgradient oracle.

For the algorithmic analysis, define
$B_M := \max_{1\le \ell\le L}\sum_{k=1}^K M_{\ell, k}$, 
$\bar\lambda=\max_i\lambda_i$, $\underline\lambda=\min_i\lambda_i$, and $\bar v=\max_i v_i$.

The only nontrivial computational task in the oracle is evaluating the utility vector of a given index policy.
The utility formula has no exact closed form, but it remains numerically tractable because each cumulative throughput term is generated by a one-dimensional birth--death chain.

\begin{lemma}[Evaluation of an index policy]\label{lmm:index_eval}
Fix a static priority ordering and a tolerance $\delta>0$.
Truncating each birth--death recursion at a level
$N\ \ge\ \max\left\{\frac{2eK\bar\lambda}{d},\ \log_2\frac{4\bar v\mu}{\underline\lambda\,\delta}\right\}$
yields a vector $\widehat u(\sigma)$ with
$\|\widehat u(\sigma)-u(\sigma)\|_\infty \le \delta$, 
at a cost of
$O\left(K\log K+\frac{K^2\bar\lambda}{d}+K\log\frac{\mu\bar v}{\underline\lambda\,\delta}\right)$
arithmetic operations.
\end{lemma}
The proof is in Appendix~\ref{app:proof_index_eval}.

Let $\gamma \in \mathbb{R}^L_+$ be the vector of Lagrange multipliers associated with the constraints $M u(\sigma) \ge c$.
The Lagrangian function $\mathcal{L}(\sigma,\gamma)$ is
\begin{align*}
    \mathcal{L}(\sigma,\gamma)
    &= w^\top u(\sigma) + \gamma^\top(Mu(\sigma)-c) \\
    &= (w+M^\top\gamma)^\top u(\sigma) - \gamma^\top c.
\end{align*}
The dual function is
\begin{equation}
    g(\gamma)
    =
    \max_{\sigma\in\Sigma}\mathcal{L}(\sigma,\gamma)
    =
    \left(\max_{\sigma\in\Sigma}(w+M^\top\gamma)^\top u(\sigma)\right)-\gamma^\top c.
\end{equation}
The dual problem is to minimize the convex function $g(\gamma)$ over $\gamma\ge 0$.
Assumption \ref{asmp:slater} bounds the dual optimum, which lets us replace the nonnegative orthant by a box.

\begin{lemma}\label{lmm:dual_radius}
Under Assumption \ref{asmp:slater}, strong duality holds for $(P)$, the dual optimum is attained, and every optimal solution $\gamma^*$ satisfies
$$
\|\gamma^*\|_1
\le
\frac{\|w\|_1\,\bar v}{s_0}
=:R.
$$
\end{lemma}
The proof is in Appendix~\ref{app:proof_dual_radius}.
We therefore fix the dual search region
$\mathcal G:=[0,R+1]^L$,
which contains every optimal dual solution, since $\|\gamma^*\|_\infty\le\|\gamma^*\|_1\le R$.

For every nonnegative dual vector $\gamma$, Theorem \ref{thm:index_optimality} implies that the inner maximization is solved by the index policy induced by the modified welfare weights
$\tilde w(\gamma)=w+M^\top\gamma.$
Thus, the oracle can identify the maximizing policy exactly by sorting classes according to the indices $b_i(\gamma)=\tilde w_i(\gamma)v_i/\lambda_i$.
The only numerical approximation enters when we evaluate the resulting policy.
\begin{definition}[Dual oracle]\label{def:oracle}
The \emph{dual oracle} takes a vector $\gamma\in\mathbb R^L_+$ and a tolerance $\delta>0$, and returns
\begin{enumerate}
    \item the index policy $\sigma_\gamma$ induced by the indices $b_i(\gamma)$, under a fixed tie-breaking rule;
    \item a utility estimate $\widehat u(\sigma_\gamma)$ with $\|\widehat u(\sigma_\gamma)-u(\sigma_\gamma)\|_\infty\le\delta$, computed by the truncated recursion of Lemma \ref{lmm:index_eval};
    \item the value estimate $\widehat g(\gamma):=(w+M^\top\gamma)^\top\widehat u(\sigma_\gamma)-\gamma^\top c$;
    \item the subgradient estimate $\widehat\xi(\gamma):=M\widehat u(\sigma_\gamma)-c$.
\end{enumerate}
\end{definition}

Forming $\tilde w(\gamma)$, $\widehat g(\gamma)$ and $\widehat\xi(\gamma)$ takes $O(KL)$ operations and the sort takes $O(K\log K)$, so by Lemma \ref{lmm:index_eval} one oracle call costs
$O\left(KL+K\log K+\frac{K^2\bar\lambda}{d}+K\log\frac{\mu\bar v}{\underline\lambda\,\delta}\right)$
arithmetic operations.

Write $C_w:=\|w\|_1+(R+1)LB_M$ and $C_g:=2B_M\bar v$.
Under Assumption \ref{asmp:slater} we have $c_\ell\le(M\bar u)_\ell\le B_M\bar v$ for every $\ell$, so $\|c\|_\infty\le B_M\bar v$ and $C_g\ge B_M\bar v+\|c\|_\infty$.

\begin{lemma}[The dual function]\label{lmm:dual_props}
For $\gamma\in\mathcal G$ set $\xi(\gamma):=Mu(\sigma_\gamma)-c$. Then $g$ is convex, $\xi(\gamma)$ is a subgradient of $g$ at $\gamma$, and $\|\xi(\gamma)\|_\infty\le C_g$, so that $|g(\gamma)-g(\gamma')|\le C_g\|\gamma-\gamma'\|_1$ for all $\gamma,\gamma'\in\mathcal G$.
\end{lemma}

\begin{lemma}[Oracle accuracy]\label{lmm:oracle_acc}
Let $\gamma\in\mathcal G$ and $\delta>0$. Then
\begin{enumerate}[(i)]
    \item $|\widehat g(\gamma)-g(\gamma)|\le C_w\,\delta$ and $\|\widehat\xi(\gamma)-\xi(\gamma)\|_\infty\le B_M\,\delta$;
    \item $g(\gamma')\ \ge\ \widehat g(\gamma)+\widehat\xi(\gamma)^\top(\gamma'-\gamma)-2C_w\,\delta$ for every $\gamma'\in\mathcal G$.
\end{enumerate}
\end{lemma}
The proofs are in Appendices~\ref{app:proof_dual_props} and~\ref{app:proof_oracle_acc}.

\subsection{The CPFS Algorithm}\label{subsec:cpfs}

Both phases of CPFS run the central-cut ellipsoid method: Phase 1 minimizes the dual function $g$ over $\mathcal G$ with the estimated subgradient $\widehat\xi$ as the cut direction, and Phase 2 maximizes the recovery objective (the function $\widehat\Phi$ in Algorithm~\ref{algo:CPFS}) over mixtures of the policies generated along the way.
Both phases use the standard ellipsoid method (see, e.g., \citealp{grotschel1981ellipsoid,grotschel2012geometric}), with index policy optimality as an efficient subgradient oracle.

\begin{definition}[Central cut]\label{def:cut}
For $\gamma\in\mathbb R^n$ and positive definite $D\in\mathbb R^{n\times n}$, write
$\mathcal B(\gamma,D):=\{\gamma'\in\mathbb R^n\mid(\gamma'-\gamma)^\top D^{-1}(\gamma'-\gamma)\le 1\}$.
Given $z\neq 0$, $\mathrm{cut}_n(\gamma,D,z)$ returns the pair $(\gamma^+,D^+)$ for which $\mathcal B(\gamma^+,D^+)$ is the minimum-volume ellipsoid containing the half-ellipsoid $\mathcal B(\gamma,D)\cap\{\gamma'\mid z^\top\gamma'\le z^\top\gamma\}$:
$$
\gamma^+=\gamma-\frac{1}{n+1}\cdot\frac{Dz}{\sqrt{z^\top Dz}},
\qquad
D^+=
\begin{cases}
\dfrac{n^2}{n^2-1}\left(D-\dfrac{2}{n+1}\cdot\dfrac{Dzz^\top D}{z^\top Dz}\right), & n\ge 2,\\[1.5ex]
D/4, & n=1 .
\end{cases}
$$
\end{definition}

The volume contracts at a fixed rate,
$\mathrm{vol}(\mathcal B(\gamma^+,D^+))/\mathrm{vol}(\mathcal B(\gamma,D))<e^{-1/(2n)}$:
the ratio is $\bigl((n/(n+1))^{n+1}(n/(n-1))^{n-1}\bigr)^{1/2}$ for $n\ge 2$ and exactly $1/2$ for $n=1$.
Both phases use $\mathrm{cut}_n$ in the same way. The current center either violates one of the box constraints that define the feasible set, in which case that constraint supplies a cut discarding no feasible point, or it does not, in which case a sub- or supergradient at the center supplies the cut.
Finally, $\mathcal A(\mathcal S):=\{\alpha\in\mathbb R^{\mathcal S}_+\mid\sum_{\sigma\in\mathcal S}\alpha_\sigma=1\}$ denotes the simplex of mixtures over a finite set $\mathcal S$ of policies.

Algorithm \ref{algo:CPFS} states the CPFS procedure. Phase 1 localizes an optimal dual solution within a sequence of shrinking ellipsoids, and Phase 2 then recovers a corresponding primal solution. Its last line is a concave maximization over a simplex; the proof of Theorem \ref{thm:algorithm} carries it out with the recursion of Definition \ref{def:cut} and counts the cost. 
Geometrically, CPFS approximates $\mathcal U$ by generating only the vertices that are active near the optimum, and Phase 2 recovers the weights over the generated policies alone. The output is thus a utility vector together with a mixture of index policies that attains it; Section \ref{subsec:implementation} converts this mixture into a single policy.

\begin{algorithm}[!tbp]
\small
\caption{Constrained Pareto Frontier Search (CPFS)}\label{algo:CPFS}
\begin{algorithmic}[1]
\State \textbf{Input:} weights $w$, constraints $(M,c)$, Slater margin $s_0$, target accuracy $\epsilon>0$.
\State \textbf{Derived constants:} $R=\|w\|_1\bar v/s_0$, $\mathcal G=[0,R+1]^L$, $\gamma^\circ=\tfrac{R+1}{2}\mathbf 1$, $C_w=\|w\|_1+(R+1)LB_M$, $C_g=2B_M\bar v$, and
\begin{align*}
&\delta_{\mathrm{orc}}=\frac{\epsilon}{48\,C_w},
\qquad
\delta_{\mathrm{rec}}=\frac{\epsilon}{16\,C_w},
\qquad
\kappa=\min\left\{\frac{1}{2}, \frac{\epsilon}{16\,C_g L(R+1)}\right\},
\\
&\bar t=\max\Bigl\{1,\ \bigl\lceil 2L^2\ln(\sqrt L/\kappa)\bigr\rceil+1\Bigr\} .
\end{align*}
\State \textbf{Initialize:} $\gamma_0=\gamma^\circ$ and $D_0=\tfrac{1}{4}(R+1)^2L\,I$, so that $\mathcal G\subseteq\mathcal B_0$, writing $\mathcal B_t:=\mathcal B(\gamma_t,D_t)$; $\mathcal S\leftarrow\emptyset$.
\Statex \textbf{Phase 1: dual minimization.}
\For{$t=0,1,\dots,\bar t-1$}
    \If{$\gamma_t\notin\mathcal G$} \Comment{feasibility cut}
        \State pick $\ell$ with $(\gamma_t)_\ell<0$ or $(\gamma_t)_\ell>R+1$, and set $z_t=-e_\ell$ or $z_t=e_\ell$ accordingly.
    \Else \Comment{optimality cut}
        \State call the dual oracle (Definition \ref{def:oracle}) at $\gamma_t$ with tolerance $\delta_{\mathrm{orc}}$, and add $\sigma_{\gamma_t}$ to $\mathcal S$.
        \State \textbf{if} $\widehat\xi(\gamma_t)=0$ \textbf{then} leave the loop.
        \State $z_t=\widehat\xi(\gamma_t)$.
    \EndIf
    \State $(\gamma_{t+1},D_{t+1})\leftarrow\mathrm{cut}_L(\gamma_t,D_t,z_t)$.
\EndFor
\Statex \textbf{Phase 2: primal recovery.}
\State Re-evaluate every $\sigma\in\mathcal S$ by Lemma \ref{lmm:index_eval} at tolerance $\delta_{\mathrm{rec}}$, obtaining $\widehat u(\sigma)$, and let
\begin{align*}
\widehat\Phi(\alpha)
=
\sum_{\sigma\in\mathcal S}\alpha_\sigma\bigl(w^\top \widehat u(\sigma)\bigr)-(R+1)\sum_{\ell=1}^L
\max\Bigl\{0,\ c_\ell-\sum_{\sigma\in\mathcal S}\alpha_\sigma\bigl(M\widehat u(\sigma)\bigr)_\ell\Bigr\},
\end{align*}
\Statex a concave piecewise-linear function on $\mathcal A(\mathcal S)$.
\State If $|\mathcal S|=1$, use the unique mixture. Otherwise, let $n=|\mathcal S|-1,$

$$
C_\Phi=2C_w(\bar v+\delta_{\mathrm{rec}}),\qquad
\kappa_{\mathrm{rec}}=\min\left\{\frac12,\frac{\epsilon}{32C_\Phi}\right\},\text{ and }
\bar t_{\mathrm{rec}}
=
\left\lceil
2n^2\left(\ln\frac{4}{\kappa_{\mathrm{rec}}}+\ln|\mathcal S|\right)
\right\rceil+1.
$$

\State If $|\mathcal S|>1$, compute $\alpha\in\mathcal A(\mathcal S)$ with
$\widehat\Phi(\alpha)\ge\max_{\mathcal A(\mathcal S)}\widehat\Phi-\epsilon/16$
by the central-cut ellipsoid method of Definition~\ref{def:cut} for $\bar t_{\mathrm{rec}}$ iterations, starting from the uniform mixture and an enclosing radius-$2$ ellipsoid; at each iteration, use a violated simplex constraint for a feasibility cut, and $-\zeta$ for a supergradient $\zeta$ of $\widehat\Phi$ otherwise. Return the best feasible mixture visited.

\State \textbf{Output:} mixture weights $(\alpha_\sigma)_{\sigma\in\mathcal S}$, the policy set $\mathcal{S}$, and the estimated utility vector
$
\widehat u^*=\sum_{\sigma\in\mathcal S}\alpha_\sigma \widehat u(\sigma).
$

\end{algorithmic}
\end{algorithm}

\begin{theorem}\label{thm:algorithm}
Let Assumptions \ref{asmp:same_d} and \ref{asmp:slater} hold.
Run Algorithm \ref{algo:CPFS} with target accuracy $\epsilon>0$.
Then the algorithm returns a set $\mathcal S$ of index policies and mixture weights $\alpha\in\mathcal A(\mathcal S)$ whose realized utility vector $u^*=\sum_{\sigma\in\mathcal S}\alpha_\sigma u(\sigma)$ satisfies
\begin{enumerate}
    \item $w^\top u^* \ge \mathrm{OPT}-\epsilon$, where
    $
    \mathrm{OPT}=\max_{u\in\mathcal U}\{w^\top u \mid Mu\ge c\};
    $
    \item
    $
    \sum_{\ell=1}^L \max\{0,\, c_\ell-(Mu^*)_\ell\}\le \epsilon.
    $
\end{enumerate}
The number of oracle calls is at most
$
\bar t=O\bigl(L^2\bigl(\log L+\log C_g+\log(R+1)+\log(1/\epsilon)\bigr)\bigr),
$
and, under the evaluation routine of Lemma \ref{lmm:index_eval}, the total number of arithmetic operations is polynomial in
$K$, $L$, $\bar\lambda/d$, and the logarithms of $\|w\|_1$, $B_M$, $\mu$, $\bar v$, $1/\epsilon$, $1/s_0$, and $1/\underline\lambda$.
\end{theorem}
The proof is in Appendix~\ref{app:algo}.

\subsection{Implementation of the Optimal Policy}\label{subsec:implementation}

CPFS outputs a set $\mathcal S$ of index policies and weights $\alpha$, whose mixture has the utility vector $u^*=\sum_{\sigma\in\mathcal S}\alpha_\sigma u(\sigma)$.
A mixture of priority rankings is not yet a rule that a transplant system can run.
This subsection shows that the mixture can be replaced by a rule of a familiar kind: a single fixed priority list, in which some classes are split by a lottery that each patient draws once, at registration.
Nothing is randomized when an organ arrives, and the list itself does not change with the state of the queues: the organ goes to the first tier on it that contains a compatible patient.

\begin{definition}[Tiered policy]\label{def:tiered}
A \emph{tier list} is a priority ranking of tiers $1,\dots,n$.
Each tier $k$ belongs to one class, $\mathrm{cl}(k)$, and takes a share $\phi_k\in(0,1]$ of it; the shares of the tiers of each class sum to one.
Under the associated \emph{tiered policy}, a registering patient is assigned to one of the tiers of their class, to tier $k$ with probability $\phi_k$, independently across patients and of everything else, and stays in it; an arriving organ is allocated to the highest-ranked tier that contains a compatible patient.
\end{definition}

The lottery can be replaced by any characteristic fixed at registration that divides the class in the same proportions, independently across patients and of the rest of the model: the tier list fixes the proportions, and the device that produces them is at the policymaker's discretion.

An index policy is the tiered policy in which every class forms one tier, and a tiered policy is an index policy that ranks parts of classes instead of whole classes.
A tier has the first claim on every organ ahead of the tiers ranked below it, which receive only the organs that find no compatible patient above them.
A group of patients therefore receives more organs the higher it is ranked, and a class split between a high tier and a low tier receives more than it would in the low position and less than it would in the high one.
Splitting classes thus interpolates between priority rankings.
The shares that reproduce a given mixture are in general not its weights.

To compute them, write $\mathcal T(\lambda)$ for the cumulative throughput function $\mathcal T^\rho(\lambda,\mu)$ and treat the tiers as classes.
Let $\Lambda_k=\sum_{j\le k}\phi_j\lambda_{\mathrm{cl}(j)}$ be the total arrival rate of tiers $1,\dots,k$, with $\Lambda_0=0$ and $\Lambda_n=\Lambda:=\sum_i\lambda_i$.
Because a Poisson stream split by independent draws yields independent Poisson streams, the tiers behave exactly like classes, and Lemma \ref{lmm:index_util} gives tier $k$ the throughput $\mathcal T(\Lambda_k)-\mathcal T(\Lambda_{k-1})$.
A tier list is thus a division of $[0,\Lambda]$ into consecutive intervals $[\Lambda_{k-1},\Lambda_k]$, each labelled with a class, in which the intervals of class $i$ have total length $\lambda_i$ and an interval $[a,b]$ carries the throughput $\mathcal T(b)-\mathcal T(a)$.
The utility of class $i$ is $v_i/\lambda_i$ times the total throughput of its intervals, so reproducing a utility vector amounts to giving every class the right throughput.

The division that reproduces a given mixture is built as follows.
For $\sigma\in\mathcal S$, let $\Lambda^\sigma_k$ be the total arrival rate of the $k$ highest-priority classes under $\sigma$, so that the class ranked $k$-th by $\sigma$ occupies the interval $[\Lambda^\sigma_{k-1},\Lambda^\sigma_k]$ of $[0,\Lambda]$.
Let $0=\omega_0<\omega_1<\dots<\omega_N=\Lambda$ be the distinct values of $\Lambda^\sigma_k$ over all $\sigma\in\mathcal S$ and $k=0,\dots,K$.
The intervals $C_j=[\omega_{j-1},\omega_j]$, $j=1,\dots,N$, are the \emph{cells}: no ranking in $\mathcal S$ has a cut point inside a cell, so each $\sigma\in\mathcal S$ gives the whole of a cell to a single class.
Let $\beta_{ij}$ be the total weight $\alpha_\sigma$ of the policies $\sigma\in\mathcal S$ that give cell $C_j$ to class $i$.
The mixture gives class $i$ the fraction $\beta_{ij}$ of the throughput of each cell $C_j$, so it suffices to divide each cell among the classes in such a way that class $i$ receives the fraction $\beta_{ij}$ of both the length and the throughput of the cell.

Such a division is found one class at a time.
For a set $\mathcal I\subseteq[0,\Lambda]$ that is a finite union of intervals, write $|\mathcal I|$ for its total length and $\tau(\mathcal I)$ for the sum of $\mathcal T(b)-\mathcal T(a)$ over its intervals $[a,b]$.
For $\beta\in(0,1)$ and $s\in[0,(1-\beta)|\mathcal I|]$, the \emph{window} $W_s(\mathcal I,\beta)$ is the part of $\mathcal I$ that begins after a length $s$ of $\mathcal I$ and has length $\beta|\mathcal I|$.
The window has the fraction $\beta$ of the length of $\mathcal I$ wherever it is placed, and its throughput $\tau(W_s(\mathcal I,\beta))$ is nonincreasing in $s$, because $\mathcal T$ is concave (Appendix \ref{app:proof_tiered}).
At $s=0$ the window is the leftmost part of $\mathcal I$ and carries at least the fraction $\beta$ of $\tau(\mathcal I)$; at $s=(1-\beta)|\mathcal I|$ it is the rightmost part and carries at most that fraction.
Some position in between gives exactly the fraction $\beta$, and bisection on $s$ finds it.
Algorithm \ref{algo:tier} carries this out with a numerical approximation $\widehat{\mathcal T}$ of $\mathcal T$, writing $\widehat\tau$ for $\tau$ computed from $\widehat{\mathcal T}$.

\begin{algorithm}[H]
\small
\caption{Tier-list recovery}\label{algo:tier}
\begin{algorithmic}[1]
\State \textbf{Input:} a finite set $\mathcal S$ of index policies, weights $\alpha\in\mathcal A(\mathcal S)$, target accuracy $\eta\in(0,\bar v]$.
\State \textbf{Derived constants:} $\eta_1=\underline\lambda\,\eta/\bar v$, $\delta=\eta_1/(8K|\mathcal S|^3)$, and $\hat N=\bigl\lceil\max\{2eK\bar\lambda/d,\ \log_2(2\mu/\delta)\}\bigr\rceil$.
\Statex Throughout, $\widehat{\mathcal T}$ is the truncated recursion of Lemma \ref{lmm:index_eval} at level $\hat N$.
\State \textbf{Cells:} sort the points $\Lambda^\sigma_k$ ($\sigma\in\mathcal S$, $k=0,\dots,K$) to obtain the cells $C_1,\dots,C_N$, and compute $\beta_{ij}$ for every class $i$ and cell $C_j$.
\State Let $m$ be the largest number of classes with $\beta_{ij}>0$ in a single cell, and set $\eta_2=\eta_1/(2Nm)$.
\For{$j=1,\dots,N$}
    \State let $i_1,\dots,i_{m_j}$ be the classes with $\beta_{ij}>0$, write $\beta_k=\beta_{i_k,j}$, and set $\mathcal I\leftarrow C_j$.
    \For{$k=1,\dots,m_j-1$}
        \State $\beta\leftarrow\beta_k/(\beta_k+\dots+\beta_{m_j})$, \ $\underline s\leftarrow0$, \ $\bar s\leftarrow(1-\beta)|\mathcal I|$, \ $\bar\tau\leftarrow\widehat\tau(\mathcal I)$.
        \While{$\bar s-\underline s>\eta_2$} \Comment{bisection on the window position}
            \State $s\leftarrow(\underline s+\bar s)/2$; \ \textbf{if} $\widehat\tau(W_s(\mathcal I,\beta))\ge\beta\,\bar\tau$ \textbf{then} $\underline s\leftarrow s$ \textbf{else} $\bar s\leftarrow s$.
        \EndWhile
        \State label the intervals of $W_{\underline s}(\mathcal I,\beta)$ with class $i_k$, and set $\mathcal I\leftarrow \mathcal I\setminus W_{\underline s}(\mathcal I,\beta)$.
    \EndFor
    \State label the intervals of $\mathcal I$ with class $i_{m_j}$.
\EndFor
\State Discard intervals of zero length, merge adjacent intervals with the same label, and let $[a_1,b_1],\dots,[a_n,b_n]$ be the resulting intervals in increasing order.
\State \textbf{Output:} the tier list with $\mathrm{cl}(k)$ the label of $[a_k,b_k]$ and $\phi_k=(b_k-a_k)/\lambda_{\mathrm{cl}(k)}$, for $k=1,\dots,n$.
\end{algorithmic}
\end{algorithm}

\begin{theorem}\label{thm:tiered}
Let Assumption \ref{asmp:same_d} hold, let $\mathcal S$ be a finite set of index policies and $\alpha\in\mathcal A(\mathcal S)$, and write $u^\alpha=\sum_{\sigma\in\mathcal S}\alpha_\sigma u(\sigma)$.
Run Algorithm \ref{algo:tier} with target accuracy $\eta\in(0,\bar v]$.
Then the algorithm returns a tier list with at most $K|\mathcal S|^3$ tiers whose tiered policy has a utility vector $u$ with
$$
\|u-u^\alpha\|_\infty\le\eta .
$$
The number of evaluations of $\widehat{\mathcal T}$ is at most
$$
\bar E=2K|\mathcal S|^3\left(\left\lceil\log_2\frac{2K^2|\mathcal S|^2\,\bar\lambda\,\bar v}{\underline\lambda\,\eta}\right\rceil+1\right),
$$
and the total number of arithmetic operations is
$$
O\left(\bar E\left(\frac{K\bar\lambda}{d}+\log\frac{K|\mathcal S|\,\mu\,\bar v}{\underline\lambda\,\eta}\right)\right),
$$
which is polynomial in $K$, $|\mathcal S|$, $\bar\lambda/d$, and the logarithms of $\mu$, $\bar v$, $\bar\lambda$, $1/\underline\lambda$, and $1/\eta$.
\end{theorem}
The proof is in Appendix~\ref{app:proof_tiered}; it also shows that some tier list of this size has the utility vector $u^\alpha$ exactly.
The set $\mathcal S$ returned by Algorithm \ref{algo:CPFS} contains at most as many policies as the algorithm makes oracle calls, so by Theorem \ref{thm:algorithm} the tier list for the CPFS mixture is computed in polynomial time.
Its utility vector $u$ satisfies $w^\top u\ge\mathrm{OPT}-\epsilon-\|w\|_1\eta$ and $\sum_{\ell}\max\{0,c_\ell-(Mu)_\ell\}\le\epsilon+LB_M\eta$, so the guarantees of Theorem \ref{thm:algorithm} are preserved up to terms that vanish with $\eta$.

\begin{example}\label{ex:tiered_case1}
In Example \ref{ex:asymmetric} with regional participation constraints (Case 1 of Section \ref{seq:experiments}), CPFS returns two rankings.
Writing a class as (region, type), they are
$$
(1,1)\succ(1,2)\succ(2,1)\succ(2,2)
\qquad\text{and}\qquad
(1,1)\succ(2,1)\succ(1,2)\succ(2,2),
$$
with weights $0.511$ and $0.489$; the second is a utilitarian ranking.
They differ only in whether the low-value patients of region 1 rank above or below the high-value patients of region 2.
The mixture is attained, for instance, by the list
$$
(1,1)\ \succ\ (1,2)^{\mathrm{upper}}\ \succ\ (2,1)\ \succ\ (1,2)^{\mathrm{lower}}\ \succ\ (2,2),
$$
in which a registering low-value patient of region 1 enters the upper tier with probability $\phi$, the solution of
$\bigl[\mathcal T(1+\phi)-\mathcal T(1)\bigr]+\bigl[\mathcal T(7)-\mathcal T(6+\phi)\bigr]=\lambda_{(1,2)}u^*_{(1,2)}/v_{(1,2)}$; numerically $\phi\approx0.462$, against the weight $0.511$ of the first ranking.
This list keeps the utilitarian ranking but lifts $46\%$ of the low-value patients of region 1 above the high-value patients of region 2, which is exactly what is needed to return region 1 to its autarky utility.
It is not the only tier list with the utility vector $u^*$: Algorithm \ref{algo:tier} returns a different one, shown in Figure \ref{fig:tier_case1}.
\end{example}

The multipliers of the participation constraints indicate which classes are split.
By strong duality (Lemma \ref{lmm:dual_radius}), an optimal utility vector maximizes the adjusted welfare $(w+M^\top\gamma^*)^\top u$ over $\mathcal U$ at an optimal dual solution $\gamma^*$.
The rankings by $b_i(\gamma^*)$ maximize it and differ only in the order of tied classes, so a mixture of such rankings splits only classes that are tied.
With utilitarian weights and regional constraints, $b_i(\gamma^*)=v_\theta(1+\gamma^*_r/\Lambda^r)$ for $i=(r,\theta)$: the value of a transplant, marked up in the regions whose constraint binds.
In Example \ref{ex:tiered_case1}, $\gamma^*=(2,0)$, and the adjusted indices of the classes $(1,1)$, $(1,2)$, $(2,1)$ and $(2,2)$ are $4$, $2$, $2$ and $1$.
The tier list is the ranking by the adjusted index, and the tie between $(1,2)$ and $(2,1)$ is where it splits the classes.

Two remarks qualify the result.
First, a tiered policy that splits a class conditions on the tier labels of waiting patients, which the state $q$ does not record, so it is not a stationary policy $\sigma(q,A)$ and lies outside the set $\Sigma$ of policies studied so far.
It does not reach outside the achievable set: the tiered policy is an index policy of the same model with the tiers as classes, and its utility vector lies in $\mathcal U$ (Appendix \ref{app:proof_tiered}).
Tiered policies enlarge the set of rules, not the set of achievable utility vectors.
Second, the participation constraints concern long-run average utilities, and the tiered policy satisfies them in that sense, up to the tolerance $\epsilon+LB_M\eta$ above: by ergodicity of the tier-level queue process and the strong law for its counting processes, the realized averages converge to their long-run values almost surely.
Over a finite horizon the constraints can fail on a sample path, because a region gives up an organ at the moment the organ is redirected, whereas the organs that compensate it arrive later and at random.
This is true of any pooled policy that moves organs across regions, tiered or not, and under autarky itself the realized averages fluctuate around the benchmark.

\subsection{Numerical Examples}\label{seq:experiments}
\subsubsection{Theoretical Examples}\label{subseq:experiments_theory}
We test CPFS on three numerical scenarios and benchmark it against a brute-force routine that enumerates all $K!$ index policies and solves the full linear program for the optimal convex combination.
Throughout Section \ref{subseq:experiments_theory}, we set $d=1.0$, $\rho=0.9$, $|\Theta|=2$ with $v_1=2.0, v_2=1.0$, utilitarian welfare weights ($w_i = \lambda_i$), and target accuracy $\epsilon=10^{-6}$.

\paragraph{Case 1: Two-Region Asymmetric Baseline.}
We apply CPFS to the asymmetric setting of Example~\ref{ex:asymmetric} ($|\mathcal R|=2$, $K=4$) with regional participation constraints $u^r(\sigma)\ge u_{sep}^r$.

\paragraph{Case 2: Multi-Region Scalability.}
We test scalability with $|\mathcal R|=5$ regions ($K=10$, $10!\approx 3.6\times 10^6$ vertices), organ rates $\mu^r = 0.5 + 0.1(r-1)$, uniform patient arrivals $\lambda^r=(1,1)$, and regional participation constraints.

\paragraph{Case 3: Class-Level Participation Constraints.}
This case illustrates the algorithm's flexibility with arbitrary linear constraints by imposing class-level participation constraints instead of regional ones.
We keep the primitives of Case 1 and replace the two regional constraints by the $K=4$ class-level constraints $u_{(r,\theta)}(\sigma)\ge u_{sep}^{(r,\theta)}$, with $u_{sep}^{(r,\theta)}$ the class-level autarky benchmark of Section \ref{subsec:separate_util}. Here $M$ is the identity and $c$ collects these benchmarks.

\paragraph{Results.}
\begin{table}[tbp]
    \centering
    \begin{tabular}{llccccc}
        \toprule
        \textbf{Scenario} & \textbf{Method} & \textbf{Time (s)} & \textbf{Vertices} & \textbf{Iterations} & $|\mathcal S|$ & \textbf{Objective}\\
        \midrule
        \multirow{2}{*}{\textbf{Case 1}} & Brute force & 0.0005 & 24 & -- & -- & 2.2125 \\
                                         & \textbf{CPFS} & 0.0035 & -- & 197 & 2 & 2.2125 \\
        \midrule
        \multirow{2}{*}{\textbf{Case 2}} & Brute force & 874.4862 & 3628800 & -- & -- & 6.2106 \\
                                         & \textbf{CPFS} & 0.0204 & -- & 1247 & 134 & 6.2106 \\
        \midrule
        \multirow{2}{*}{\textbf{Case 3}} & Brute force & 0.0006 & 24 & -- & -- & 2.1155 \\
                                         & \textbf{CPFS} & 0.0107 & -- & 820 & 17 & 2.1155 \\
        \bottomrule
    \end{tabular}
    \caption{Performance comparison: CPFS vs. brute force.}\label{tab:cpfs_comparison}
\end{table}

Table \ref{tab:cpfs_comparison} summarizes performance across the three cases.
We report computation time, the number of index-policy vertices ($K!$), the number of CPFS iterations $\bar t$, the number $|\mathcal S|$ of distinct index policies the oracle generated, and the final objective value.\footnote{All experiments were run on a Mac mini with an Apple M2 processor. The implementation performs the last line of Phase 2 by writing $\max_{\mathcal A(\mathcal S)}\widehat\Phi$ as a linear program using slack variables and solving it with the HiGHS solver in floating-point arithmetic, in place of the ellipsoid recursion.}
CPFS matches the brute-force objective in every case to machine precision. In Case 2, it runs about 43,000 times faster, demonstrating scalability to combinatorially large policy spaces. The gap between $\bar t$ and $|\mathcal S|$ is the point of the primal-recovery step: the iteration bound is what the volume argument of Phase 1 needs, while the linear program of Phase 2 runs over the far smaller set of policies actually generated.

\begin{table}[tbp]
    \centering
    \begin{tabular}{cccccc}
        \toprule
        \textbf{Class $(r,\theta)$} & $\lambda_{(r,\theta)}$ & $v_\theta$ & \textbf{Autarky} & \textbf{Case 1} & \textbf{Case 3} \\
        \midrule
        $(1,1)$ & 1 & 2 & 0.7996 & 0.8967 & 0.7996 \\
        $(1,2)$ & 1 & 1 & 0.2645 & 0.1674 & 0.2645 \\
        $(2,1)$ & 5 & 2 & 0.0788 & 0.2282 & 0.2088 \\
        $(2,2)$ & 1 & 1 & 0.0018 & 0.0075 & 0.0075 \\
        \bottomrule
    \end{tabular}
    \caption{Class-level average utilities in the two-region asymmetric example. Case 1 imposes regional participation constraints, Case 3 class-level ones.}\label{tab:class_level_utilities}
\end{table}

Table~\ref{tab:class_level_utilities} reports the class-level averages behind Cases 1 and 3.
Just as pooling on its own does not guarantee that every region improves, a constraint stated at the regional level does not guarantee that every type inside a region improves: under Case 1, region 1 meets its autarky average exactly while its type-2 class ends up below the autarky level.
However, the framework can incorporate class-level constraints. In Case 3, CPFS satisfies them and again matches the brute-force optimum, at a cost of $4.4$ percent of utilitarian welfare compared to Case 1.
In both cases the solution puts positive weight on very few rankings: on both of the two that CPFS generates in Case 1, and on three of the seventeen that it generates in Case 3.

\begin{figure}[tbp]
    \centering
    \includegraphics[width=0.85\linewidth]{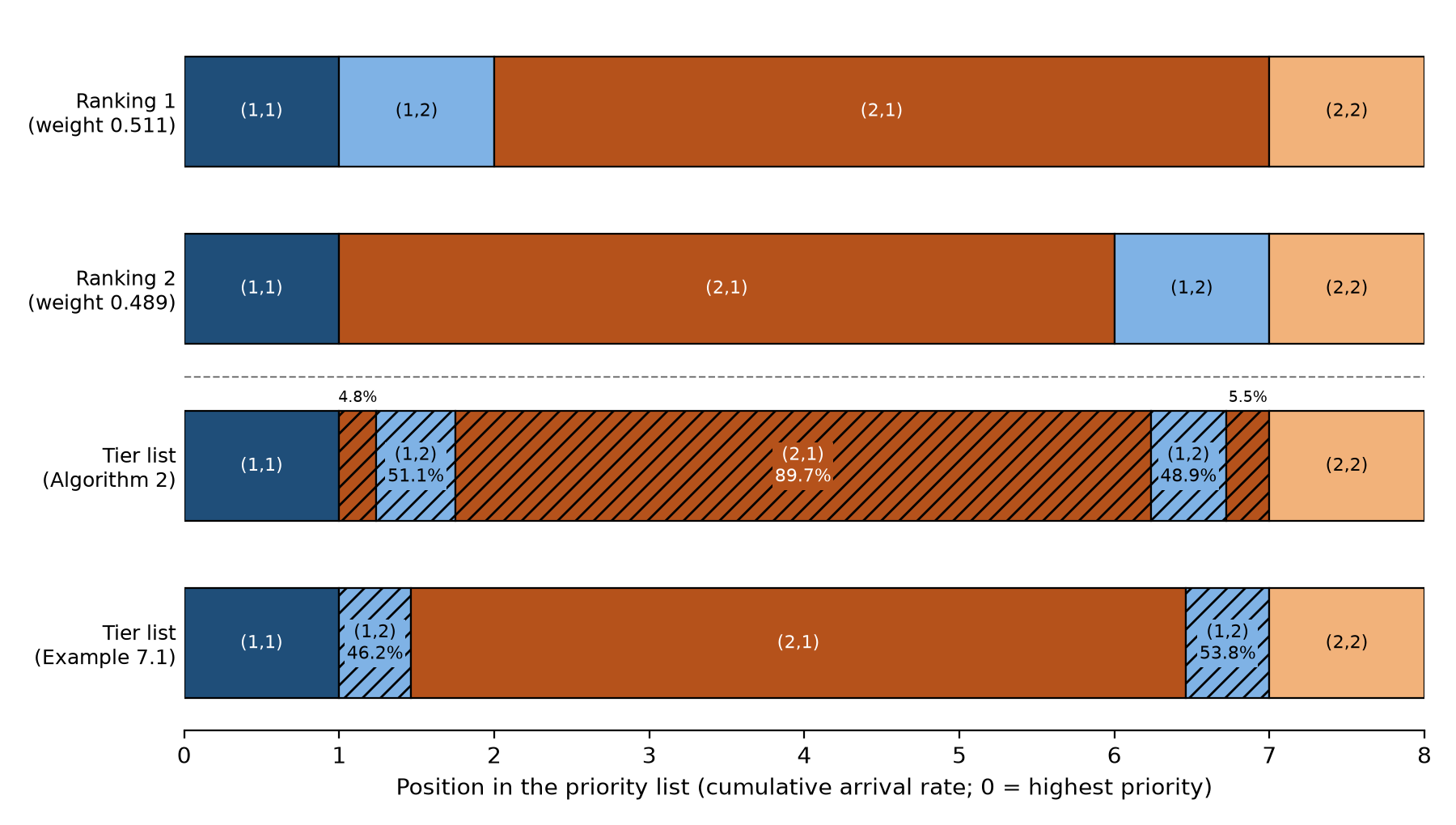}
    \caption{The CPFS solution of Case 1 as a single priority list. Each row lays a priority list on $[0,\Lambda]$, highest priority at the left, with block widths equal to arrival rates. Above the dashed line: the two rankings mixed by CPFS. Below: the tier list returned by Algorithm \ref{algo:tier} and the list of Example \ref{ex:tiered_case1}. Hatched blocks are tiers of a split class, labelled with their registration probabilities.}
    \label{fig:tier_case1}
\end{figure}

Passing the Case 1 mixture to Algorithm \ref{algo:tier} with target accuracy $\eta=10^{-6}$ returns the seven-tier list of Figure \ref{fig:tier_case1}.
Under this list the classes $(1,1)$ and $(2,2)$ are not split and keep the first and the last place.
A registering patient of class $(1,2)$ draws the upper of its two tiers with probability $0.511$ and the lower one otherwise, and a patient of class $(2,1)$ draws one of its three tiers with probabilities $0.048$, $0.897$ and $0.055$.
Each arriving organ is then offered down the list, to class $(1,1)$, to the first tier of $(2,1)$, to the upper tier of $(1,2)$, to the second tier of $(2,1)$ and so on, and goes to the first tier that contains a compatible patient.
The utility vector of this seven-tier list is within $10^{-9}$ of that of the mixture in every class, and the computation takes $104$ evaluations of $\widehat{\mathcal T}$, against the bound $\bar E=2048$ of Theorem \ref{thm:tiered}.
The list of Example \ref{ex:tiered_case1}, in the last row of the figure, has the same utility vector with five tiers: it splits class $(1,2)$ only, in the proportions $0.462$ and $0.538$.

\subsubsection{Calibration to U.S. Kidney-Transplant Data}
\label{subseq:experiments_calibration}

We coarsely calibrate the model to 2023 U.S.\ kidney-transplant data. The aim is not to evaluate the current U.S.\ allocation system quantitatively, but to show, with empirically plausible regional asymmetries, how participation constraints can bind in practice.
We set $|\Theta|=1$ and normalize the transplantation value to one. Regions correspond to the eleven OPTN regions, and the organ is restricted to kidneys from deceased donors.
With a single type, class-level and regional participation constraints coincide, so the distinction drawn in Case 3 does not arise here.
Regional patient-arrival rates $\lambda^r$ are calibrated from annual kidney waitlist additions by OPTN region, organ-arrival rates $\mu^r$ from recovered deceased-donor kidneys by region of procurement \citep{optn_public_data}, and the common abandonment rate $d$ from aggregate waitlist mortality \citep{optn_public_data, schladt2025optn}.
The compatibility parameter $\rho$ is not separately calibrated: the natural identifying moments, such as transplant volumes, discard rates, waiting times, and queue lengths, are policy-dependent and do not pin down $\rho$ in a policy-invariant way, so we treat it as exogenous, set $\rho=0.2$ in the baseline, and check robustness across a range of values in the appendix.
Appendix~\ref{app:calibration_us} reports details.
The exercise therefore introduces empirically grounded heterogeneity through $\lambda^r$ and $\mu^r$ while keeping the abandonment and compatibility parameters common across regions.

We compare three policies: \emph{autarky} ($sep$), in which each region uses only its own kidneys; the \emph{single-list benchmark} ($SL$), which corresponds to merging all regional queues into a single unified waiting list so that every patient receives the same transplant probability regardless of region, giving $u^r_{SL}=T_{pool}/\Lambda$ for all $r$, where $\Lambda=\sum_r\lambda^r$ and $T_{pool}$ is the aggregate pooled throughput (which, in the one-type environment, is the same for all nonwasteful policies); and the \emph{CPFS policy} ($C$), obtained by maximizing utilitarian welfare $\sum_r\lambda^r u^r$ subject to $u^r\ge u^r_{sep}$.
An actual first-come-first-served pooled list uses patients' arrival order and is not stationary in our reduced state space, but the benchmark single-list utility vector still lies in the convex hull of stationary static index policies: a region-blind stationary rule that treats compatible waiting patients symmetrically gives every patient the same transplant probability, and the achievable-set characterization implies that the resulting utility vector can be implemented by a mixture of static index policies.
For each policy $\sigma$, Table~\ref{tab:us_policy_comparison} reports the aggregate throughput $T_\sigma=\sum_r\lambda^r u^r_\sigma$, the throughput gain $\Delta T=(T_\sigma-T_{sep})/T_{sep}$, and the smallest participation slack $\min_r s^r_\sigma$ where $s^r_\sigma=u^r_\sigma-u^r_{sep}$. 
Note that since the value of transplantation is normalized to one, throughputs are equivalent to average utilities.

\begin{table}[tbp]
    \centering
    \small
    \begin{tabular*}{\textwidth}{@{\extracolsep{\fill}} lcccc @{}}
        \toprule
        \textbf{Policy}
        & $T_\sigma$
        & $\Delta T$
        & $\min_r s^r_\sigma$
        & \textbf{\#Viol.} \\
        \midrule
        Autarky $(sep)$
        & 82.27
        & 0.00\%
        & 0.000
        & 0 \\
        Single list $(SL)$
        & 84.39
        & 2.58\%
        & $-0.344$
        & 4 \\
        CPFS $(C)$
        & 84.39
        & 2.58\%
        & 0.000
        & 0 \\
        \bottomrule
    \end{tabular*}
    \caption{Aggregate performance in the U.S.\ kidney-transplant calibration ($\rho=0.2$). $T_\sigma$: aggregate throughput. $\Delta T$: gain over autarky. \#Viol.: regions violating participation constraints.}
    \label{tab:us_policy_comparison}
\end{table}

The single-list benchmark attains the pooled throughput and equalizes transplant probabilities across regions, yet it violates four regions' participation constraints: regions whose autarky transplant probability exceeds the pooled average lose from equal treatment.
CPFS preserves the same aggregate throughput while satisfying every constraint; Figure~\ref{fig:us_participation_slack} shows the per-region relative participation slack under the single-list benchmark.

\begin{figure}[tbp]
    \centering
    \includegraphics[width=0.7\linewidth]{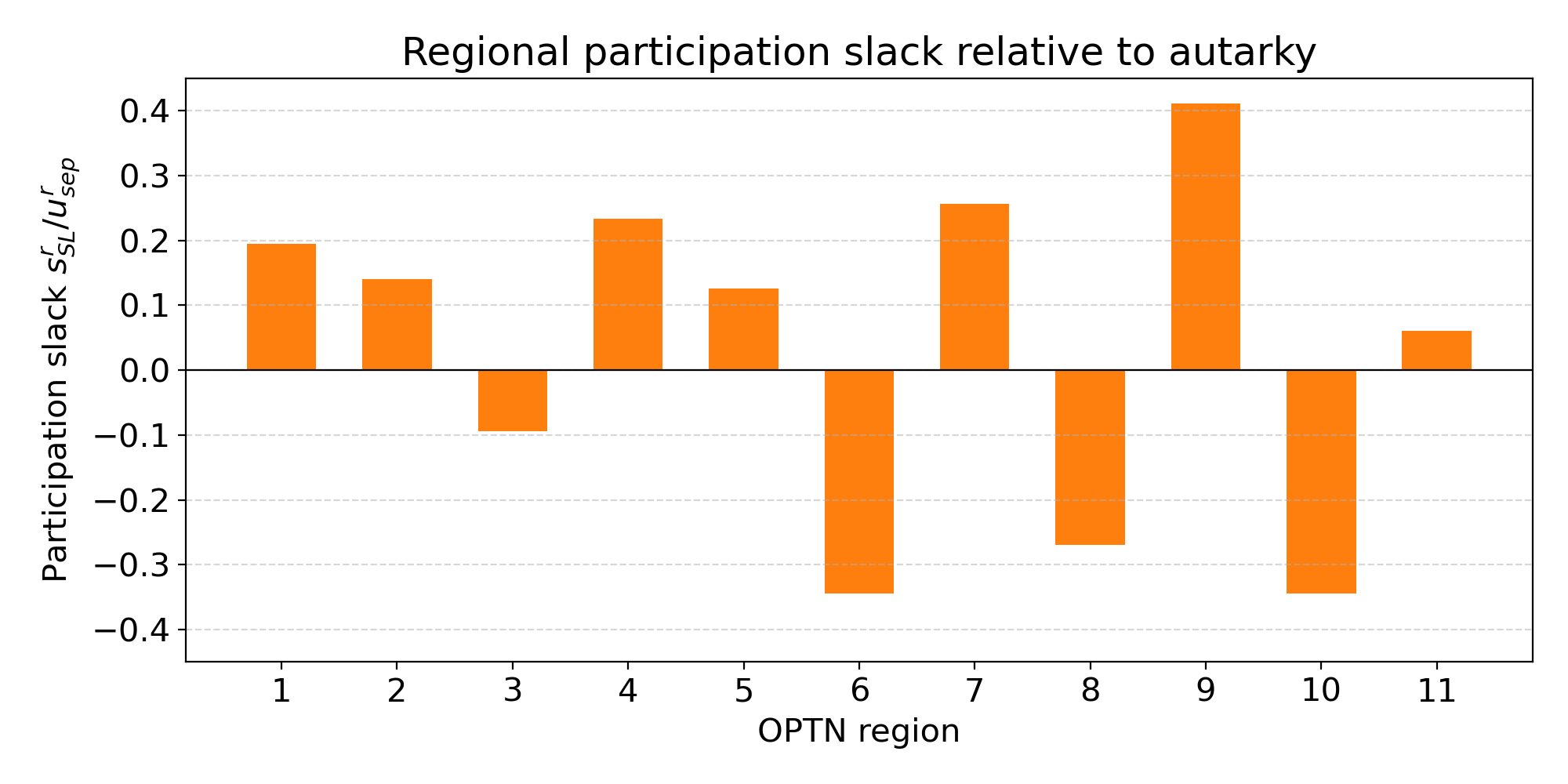}
    \caption{Relative participation slack by region under the single-list benchmark: $s^r_{SL}/u^r_{sep}=(u^r_{SL}-u^r_{sep})/u^r_{sep}$. The horizontal zero line corresponds to the autarky participation constraint.}
    \label{fig:us_participation_slack}
\end{figure}

\section{Conclusion}\label{sec:conclusion}
This paper studies Pareto improvements from moving from fragmented to pooled deceased-donor waiting lists. Motivated by the tension between efficiency gains from broader sharing and regional distributional concerns, we model allocation as a multi-class queue with abandonment and multiple regions.

The analysis yields four main results. First, index policies are optimal for maximizing weighted social welfare in this dynamic setting, which gives a tractable structure for policy design.
Second, although dynamic risk pooling creates a surplus that enlarges the achievable utility set, standard utilitarian policies often fail to satisfy participation constraints in asymmetric environments; we also identify a sufficient symmetry condition under which a utilitarian policy does yield a Pareto improvement.
Third, the Constrained Pareto Frontier Search (CPFS) algorithm efficiently computes a utility vector that approximately maximizes social welfare over the achievable utility set subject to regional participation constraints, expressed as a convex combination of the utility vectors of a few priority rankings.
Fourth, although such a combination is not itself an allocation rule, a single priority list that splits some classes into tiers by a one-time lottery at registration reproduces it to any desired accuracy, and is also computed in polynomial time.
Together, these results allow policymakers to keep redistribution within the gains from thicker markets, balancing global efficiency with regional fairness.

\bibliographystyle{ACM-Reference-Format}
\bibliography{references}

\appendix

\section{Proof of Theorem \ref{thm:index_optimality}}\label{app:proof_indexopt}
\subsection{Proof Outline}
The long-run average cost
$$
C(\sigma)=\sum_i \frac{w_i d_i v_i}{\lambda_i}L_i^\sigma
=\sum_i b_i d_i L_i^\sigma
=E_{\pi^\sigma}\left[\sum_i b_i d_i q_i\right]
$$
equals the long-run average cost starting from any state $q$:
$$
C(\sigma)=\limsup_{T\to\infty}\frac{1}{T}
E^\sigma_{q(0)=q}\int_0^T \sum_i b_i d_i q_i(t)\,dt.
$$
We define the corresponding $\alpha$-discounted cost as
$$
C_\alpha(\sigma,q)=
E^\sigma_{q(0)=q}\int_0^\infty e^{-\alpha t}\sum_i b_i d_i q_i(t)\,dt.
$$
We denote the optimal costs by $C^*$ and $C_\alpha^*(q)$, respectively. We establish average-cost optimality via the vanishing-discount limit of the $\alpha$-discounted value function.

In this proof, without loss of generality, we restrict attention to \emph{nonwasteful} policies: whenever an organ arrives, and the realized compatible set is nonempty, the organ is allocated to some compatible class; it is discarded only when the compatible set is empty. This restriction is without loss because allocating an otherwise perishable organ to a compatible patient weakly reduces the future queue-length process relative to discarding it.

For each state $q\in S$, let $p_q(A)$ denote the probability that the realized compatible set is exactly $A\subseteq J(q)$, as defined in the model section. For a real-valued function $c$ on $S$, define
$$
\psi_c(q,A)=
\begin{cases}
c(q) & \text{if } A=\emptyset,\\
\min_{i\in A} c(q-e_i) & \text{if } A\neq\emptyset,
\end{cases}
$$
and the compatibility-adjusted service operator
$$
\Gamma c(q)=\sum_{A\subseteq J(q)} p_q(A)\,\psi_c(q,A).
$$
Thus, $\Gamma c(q)$ is the expected minimal post-arrival continuation value at an organ-arrival epoch in state $q$.

Define the event rate
$$
\widehat Q(q)=\sum_i \lambda_i+\sum_i d_i q_i+\mu,
$$
which is independent of the policy.
For a real-valued function $c$ on $S$, define the operator $T$ for the average-cost problem by
$$
Tc(q)=
\sum_i b_i d_i q_i
+\sum_i \lambda_i c(q+e_i)
+\sum_i d_i q_i c(q-e_i)
+\mu \Gamma c(q).
$$
The average-cost optimality equation is
$$
g+\widehat Q(q)h(q)=Th(q),
$$
where $h$ is the relative cost function and $g$ is the optimal average cost.

For the discounted problem, the Bellman operator is
$$
T_\alpha c(q)=\frac{1}{\alpha+\widehat Q(q)}\,Tc(q),
$$
with the discounted optimality equation
$$
c_\alpha(q)=T_\alpha c_\alpha(q).
$$

We first collect standard properties of average- and discounted-cost CTMDPs from \citet{guo2006survey}.

\begin{lemma}[\citet{guo2006survey}, Theorems 3.1, 3.2]\label{lmm:opt_eq_discounted}
    Let $C^0_\alpha(q) = 0$ and $C^{k+1}_\alpha(q) = T_\alpha C^k_\alpha(q)$ for $k\ge 0$.
    Then
    \begin{enumerate}
        \item There exist deterministic policies $\sigma_k$ attaining the minimum in $C^{k+1}_\alpha = T_\alpha C^k_\alpha$ for $k\ge 0$.
        \item $\lim_{k\to \infty}C^k_\alpha = C^*_\alpha$, and the function $C^*_\alpha$ is a solution of the discounted optimality equations.
        \item There exist deterministic stationary policies attaining the minimum in the discounted optimality equations.
    \end{enumerate}
\end{lemma}
Thus, the discounted value function can be found by value iteration, and an optimal policy can be derived from it.

\begin{lemma}[\citet{guo2006survey}, in proof of Theorem 4.1]\label{lmm:discount_to_avg_convergence}
    There exist a sequence $(\alpha_m)$ of discount factors, discounted-cost optimal policies $(\sigma^*_{\alpha_m})$, a constant $g^*$, and a function $h^*$ on $S$ such that, for all $q$,
    \begin{align}
        g^*&=\lim_{m\to\infty} \alpha_m C^*_{\alpha_m}(q_0),\\
        h^*(q) &= \lim_{m\to\infty}(C^*_{\alpha_m}(q) - C^*_{\alpha_m}(q_0)),\\
        \sigma^*(q) &= \lim_{m\to\infty}\sigma^*_{\alpha_{m}}(q),
    \end{align}
    where $\sigma^*_{\alpha}$ is a discounted-cost optimal policy for discount factor $\alpha$.
    Moreover, $(g^*, h^*)$ satisfy the average-cost optimality equations and $C^*=g^*$, and $\sigma^*$ is average-cost optimal.
\end{lemma}
This implies that an average-cost optimal policy can be obtained as the limit of discounted-cost optimal policies.

\subsubsection{Optimality for a Truncated Model}
To use uniformization \citep{lippman1975applying}, we first truncate the state space to
$$
S_N=\{q\in S\mid N(q)\le N\}.
$$
If $N(q)=N$ and a class $i$ patient arrives, we add the arrival and eject one patient: the class with the highest priority (and thus $b,d$) among those in $q+e_i$.
The state becomes $q+e_i-e_{\min J(q+e_i)}$\footnote{This ejection rule is a technical device to facilitate the proof. Although ejecting a high-priority patient seems counterintuitive, the truncation serves only as an intermediate step; optimality is established in the untruncated limit, where boundary mechanics do not affect the result.}.
Abandonments, organ arrivals, compatibility realizations, allocations, and costs remain as in the original system.

This truncation yields a bounded total event rate, enabling uniformization.
Let
$$
Q_N=\sum_i \lambda_i + (N+1)\max_i d_i + \mu.
$$
For convenience, define
$$
\tau_i(q)=
\begin{cases}
q+e_i & \text{if } N(q)<N,\\
q+e_i-e_{\min J(q+e_i)} & \text{if } N(q)=N.
\end{cases}
$$

Events occur according to a Poisson process with rate $Q_N$.
Transition probabilities $\nu^\sigma_N(q, q')$ are proportional to the original rates. For all $i$:
\begin{align*}
    \nu^\sigma_N(q,q-e_i) &= \frac{1}{Q_N}(d_i q_i + \mu \bar{\sigma}_i(q)),\\
    \nu^\sigma_N(q,\tau_i(q)) &= \frac{1}{Q_N}\lambda_i \quad \text{ if } \tau_i(q)\neq q,\\
    \nu^\sigma_N(q,q)&=1-\frac{1}{Q_N}\left(\sum_id_iq_i+\sum_i\lambda_i\mathbf{1}_{\tau_i(q)\neq q} + \mu\sum_{i}\bar{\sigma}_i(q)\right)
\end{align*}
All other transition probabilities are 0.

Let $C_{\alpha, N}^*$ be the optimal $\alpha$-discounted cost for the truncated model.
The optimality equation is $C_{\alpha,N}^*=T_{\alpha,N}C_{\alpha,N}^*$, where
\begin{align*}
(\alpha+Q_N)T_{\alpha,N}c(q)
={}&\sum_i b_i d_i q_i+\sum_i \lambda_i c(\tau_i(q))+\sum_i d_i q_i c(q-e_i)\\
&+\mu \Gamma c(q)+\Bigl(Q_N-\sum_i \lambda_i-\sum_i d_i q_i-\mu\Bigr)c(q).
\end{align*}

We next define two properties for a function $c$.

\begin{definition}
    $c$ is said to satisfy (P1) if
    $$
    c(q-e_i)\le c(q-e_j)
    $$
    for all $i<j$ and all $q\in S_{N+1}$ with $i,j\in J(q)$.
\end{definition}

\begin{definition}
    $c$ is said to satisfy (P2) if
    $$
    c(q+e_i)-c(q)\le b_i
    $$
    for all $i$ and all $q$ with $q+e_i\in S_N$.
\end{definition}

We first record the two monotonicity properties of $\Gamma$.

\begin{lemma}\label{lmm:Gamma_monotone}
If $c$ satisfies (P1) and (P2), then:
\begin{enumerate}
    \item for all $i<j$ and all $q\in S_{N+1}$ with $i,j\in J(q)$,
    $$
    \Gamma c(q-e_i)\le \Gamma c(q-e_j);
    $$
    \item for all $i$ and all $q$ with $q+e_i\in S_N$,
    $$
    \Gamma c(q+e_i)-\Gamma c(q)\le b_i.
    $$
\end{enumerate}
\end{lemma}

Under Assumption \ref{asmp:d_order}, these properties imply that the Bellman operator preserves (P1) and (P2).

\begin{lemma}\label{lmm:conserve}
    If $c$ satisfies (P1) and (P2), then $T_{\alpha,N}c$ also satisfies (P1) and (P2).
\end{lemma}

Starting from $C_{\alpha,N}^0=0$, which satisfies (P1) and (P2), value iteration converges to $C_{\alpha,N}^*$, which therefore inherits both properties. Convergence is guaranteed for this finite-state MDP \citep{puterman2014markov}. Property (P1) implies $\alpha$-discounted optimality of the index policy in the truncated model.

\subsubsection{Convergence of the Truncated Model}
We now show that the truncated model converges to the original model.

First, we prove that $C_{\alpha,N}^*(q)$ converges.

\begin{lemma}\label{lmm:nondecr_C}
    $C_{\alpha,N}^*(q)$ is pointwise nondecreasing in $N$.
\end{lemma}

The instantaneous cost is bounded by $N(q)\max_i b_i d_i$.
For any fixed initial state, the discounted expected total queue length is finite even in the untruncated model and in the complete absence of service because abandonment remains active.
Thus, for fixed $\alpha$ and $q$, the discounted cost $C_{\alpha,N}^*(q)$ has a common finite upper bound for all $N$.
Hence, by Lemma \ref{lmm:nondecr_C}, $C_{\alpha,N}^*(q)$ converges pointwise to some limit $C_{\alpha,\infty}^*(q)$ as $N\to\infty$.

This limit satisfies the discounted optimality equation and Property (P1).

\begin{lemma}\label{lmm:convergence}
    The limit $C_{\alpha,\infty}^*$ satisfies the discounted optimality equation and Property (P1) on the untruncated state space.
\end{lemma}

From Lemmas \ref{lmm:convergence} and \ref{lmm:opt_eq_discounted}, discounted-cost optimality holds in the untruncated model.

\begin{lemma}\label{lmm:discount_optimality}
    The index policy is $\alpha$-discounted-cost optimal in the untruncated model for every $\alpha>0$.
\end{lemma}

By Lemma \ref{lmm:discount_to_avg_convergence}, average-cost optimality follows, completing the proof.
\qed

\subsection{Proof of Lemmas \ref{lmm:opt_eq_discounted} and \ref{lmm:discount_to_avg_convergence}}
\begin{proof}
We verify the conditions from \citet{guo2006survey} for the lemmas (Assumptions A, A*, B, C, D).

\begin{condition}{(Assumption A in \citet{guo2006survey})}\label{cond:A}
    There exist a sequence $\{S_m, m\ge 1\}$ of subsets of $S$, a nondecreasing function $w\ge 1$ on $S$, constants $b\ge 0$ and $c\neq 0$ such that
    \begin{enumerate}
        \item $S_m\uparrow S$ and $\sup\{-\nu^\sigma(q,q)| q \in S_m\} <\infty$ for each $m\ge 1$ and any policy $\sigma$.
        \item $\inf\{w(q)|q\notin S_m\} \to \infty$ as $m\to\infty$.
        \item $\sum_{q'\in S} \nu^\sigma(q,q')w(q') \le cw(q) + b$ for all $q\in S$ and any policy $\sigma$.
    \end{enumerate}
\end{condition}
Building on the Lyapunov function $V(q)=N(q)$ of the proof of Lemma \ref{lmm:ergodic}, let $w(q)=N(q)+1$, $S_m=\{q\in S\mid N(q)\le m\}$, $\lambda = \sum_i \lambda_i$, and $\underline{d} = \min_i d_i > 0$. Conditions 1 and 2 are immediate. The left-hand side of Condition \ref{cond:A}.3 is the drift $\mathcal{D}N(q)$, bounded by $\lambda - \underline{d}N(q)=\lambda+\underline{d}-\underline{d}\,w(q)$. Thus, setting $b=\lambda+\underline{d}$ and $c=-\underline{d}$, Condition \ref{cond:A} holds.

\begin{condition}{(Assumption B in \citet{guo2006survey})}\label{cond:B}
    \begin{enumerate}
        \item For every $q\in S$, any policy $\sigma$, and some constant $M> 0$,
        $|\sum_i b_id_iq_i| \le Mw(q)$, where $w$ comes from Condition \ref{cond:A}.
        \item The discount factor $\alpha$ satisfies $\alpha > c$, where the constant $c$ is as in Condition \ref{cond:A}.
    \end{enumerate}
\end{condition}
Setting $M=\max_i b_id_i$ satisfies Condition \ref{cond:B}.1. Since $c<0$, any $\alpha>0$ satisfies Condition \ref{cond:B}.2.

\begin{condition}{(Assumption C in \citet{guo2006survey})}\label{cond:C}
    \begin{enumerate}
        \item The action set is compact for each state.
        \item The temporal cost, transition rates, and drifts are continuous in actions.
        \item There exist a nonnegative function $w'$ on $S$, and constants $c'>0$, $b'\ge 0$, and $M'>0$ such that $-\nu^\sigma(q,q)w(q)\le M' w'(q)$ and $\sum_{q'\in S}\nu^\sigma(q,q')w'(q')\le c' w'(q)+b'$ for all $q\in S$ and any policy $\sigma$.
    \end{enumerate}
\end{condition}
At state $q$, an action specifies, for each realized compatible set $A\subseteq J(q)$, a feasible patient class in $A$.
Hence the action space is finite, which verifies Conditions \ref{cond:C}.1 and \ref{cond:C}.2.
For \ref{cond:C}.3, let $M'=1$ and $$w'(q)=(\lambda+N(q)\max_i d_i+ \mu)(N(q)+1).$$ The first inequality holds. For the second inequality, let $N(q)=n$ and $\max_i d_i = \bar{d}$.
From the definition of $w'$,
\begin{align*}
    \sum_{q'\in S} \nu^\sigma(q,q')w'(q')
    &=\sum_i \lambda_i (w'(q+e_i)-w'(q)) + \sum_id_iq_i (w'(q-e_i)-w'(q)) \\
    &\quad+ \mu\sum_i\bar{\sigma}_i(q)(w'(q-e_i)-w'(q))\\
    &= \sum_i \lambda_i ((\lambda+\bar{d}(n+1)+\mu)(n+2)-(\lambda+\bar{d}n+\mu)(n+1)) \\
    &\quad + \sum_id_iq_i((\lambda+\bar{d}(n-1)+\mu)n-(\lambda+\bar{d}n+\mu)(n+1)) \\
    &\quad + \mu\sum_i\bar{\sigma}_i(q)((\lambda+\bar{d}(n-1)+\mu)n-(\lambda+\bar{d}n+\mu)(n+1)) \\
    &\le \lambda (\lambda+\bar{d}(n+1)+\mu)(n+2).
\end{align*}
Since $n+2\le 2(n+1)$ and $\lambda+\bar{d}(n+1)+\mu\le\bigl(1+\bar{d}/(\lambda+\mu)\bigr)(\lambda+\bar{d}n+\mu)$,
$$\lambda (\lambda+\bar{d}(n+1)+\mu)(n+2)\le 2\lambda\Bigl(1+\frac{\bar{d}}{\lambda+\mu}\Bigr)w'(q).$$
Thus, $c'=2\lambda(1+\bar{d}/(\lambda+\mu))$ and $b'=0$ suffice.

\begin{condition}{(Assumption A* in \citet{guo2006survey})}\label{cond:Astar}
    Conditions \ref{cond:A}.1 and \ref{cond:A}.2 hold. Condition \ref{cond:A}.3 is replaced with: for the Lyapunov function $w$, the Foster--Lyapunov condition is satisfied.
\end{condition}
\begin{condition}{(Assumption D(1) in \citet{guo2006survey})}\label{cond:D1}
    For each policy, the Markov process is irreducible.
\end{condition}
These are shown in the proof of Lemma \ref{lmm:ergodic}, since $w=V+1$ has the same drift as $V$ there.

\begin{condition}{(Sufficient condition for D(2) in Proposition 4.2 of \citet{guo2006survey})}\label{cond:D2}
    The function $w^2$ satisfies the Foster--Lyapunov condition.
\end{condition}
The drift of $w^2(q)=(N(q)+1)^2$ is
\begin{align*}
    \mathcal{D}w^2(q)&=\sum_i\lambda_i ((N(q)+2)^2-(N(q)+1)^2) + \sum_id_iq_i(N^2(q)-(N(q)+1)^2) \\
    &\quad + \mu\sum_i\bar{\sigma}_i(q)(N^2(q)-(N(q)+1)^2)\\
    &\le \lambda (2N(q)+3) - \underline{d}N(q)(2N(q)+1).
\end{align*}
This is quadratic in $N(q)$ with a negative leading coefficient. For sufficiently large $N(q)$, the right-hand side is negative and decreasing in $N(q)$, so Condition \ref{cond:D2} holds.
\end{proof}

\subsection{Proof of Lemma \ref{lmm:Gamma_monotone}}
\begin{proof}
For the first part, fix $i<j$ and $q$ with $i,j\in J(q)$, and write
$
x=q-e_i, y=q-e_j.
$

Couple the compatibility realizations in states $x$ and $y$ as follows.
All patients common to both states receive the same Bernoulli compatibility realizations.
The unique additional class-$j$ patient in $x$ and the unique additional class-$i$ patient in $y$ receive the same Bernoulli realization.
This coupling is valid because compatibility is i.i.d. across patients and homogeneous across classes.

For any joint realization of these coupled draws, let $A_x$ and $A_y$ denote the corresponding compatible sets.
If $A_y=\emptyset$, then necessarily $A_x=\emptyset$ as well, and therefore
$$
\psi_c(x,A_x)=c(x)\le c(y)=\psi_c(y,A_y)
$$
by (P1).

Now suppose $A_y\neq\emptyset$.
Let $\ell\in A_y$ attain the minimum in $\psi_c(y,A_y)$.

If the chosen patient in class $\ell$ is common to both $x$ and $y$, then $\ell\in A_x$ as well, and hence
$$
\psi_c(x,A_x)\le c(x-e_\ell)=c(q-e_i-e_\ell).
$$
Since $i<j$, property (P1) gives
$$
c(q-e_i-e_\ell)\le c(q-e_j-e_\ell)=c(y-e_\ell)=\psi_c(y,A_y).
$$
Thus $\psi_c(x,A_x)\le \psi_c(y,A_y)$.

If $\ell=i$ and the chosen compatible patient in class $i$ is the unique extra class-$i$ patient in $y$,
then by the coupling the unique extra class-$j$ patient in $x$ is also compatible.
Therefore
$$
\psi_c(x,A_x)\le c(x-e_j)=c(q-e_i-e_j)=c(y-e_i)=\psi_c(y,A_y).
$$

Hence, realization by realization,
$$
\psi_c(x,A_x)\le \psi_c(y,A_y).
$$
Taking expectations yields
$$
\Gamma c(q-e_i)\le \Gamma c(q-e_j).
$$

For the second part, fix $i$ and $q$ with $q+e_i\in S_N$.
Write
$
x=q,\quad y=q+e_i.
$
Couple compatibility realizations by assigning identical realizations to all patients common to both states, and an independent Bernoulli realization to the unique extra class-$i$ patient in $y$.

For any realization, let $A_x$ and $A_y$ be the compatible sets.
If $A_x\neq\emptyset$, let $\ell\in A_x$ attain the minimum in $\psi_c(x,A_x)$.
Then $\ell\in A_y$ as well, so
$$
\psi_c(y,A_y)\le c(y-e_\ell)=c(q-e_\ell+e_i).
$$
Because $(q-e_\ell)+e_i\in S_N$, property (P2) implies
$$
c(q-e_\ell+e_i)-c(q-e_\ell)\le b_i.
$$
Therefore
$$
\psi_c(y,A_y)\le c(q-e_\ell)+b_i=\psi_c(x,A_x)+b_i.
$$

If $A_x=\emptyset$, then $\psi_c(x,A_x)=c(q)$.
If the extra class-$i$ patient in $y$ is incompatible, then $A_y=\emptyset$ as well and
$$
\psi_c(y,A_y)=c(q+e_i)\le c(q)+b_i=\psi_c(x,A_x)+b_i
$$
by (P2).
If the extra class-$i$ patient is compatible, then $y$ can serve that patient, so
$$
\psi_c(y,A_y)\le c(q)=\psi_c(x,A_x).
$$

Thus, realization by realization,
$$
\psi_c(y,A_y)\le \psi_c(x,A_x)+b_i.
$$
Taking expectations yields
$$
\Gamma c(q+e_i)-\Gamma c(q)\le b_i.
$$
\end{proof}

\subsection{Proof of Lemma \ref{lmm:conserve}}
\begin{proof}
First, we show (P1). Fix $i<j$ and $q\in S_{N+1}$ with $i,j\in J(q)$.

If $N(q)\le N$, then by (P1),
\begin{align*}
    &\quad (\alpha+Q_N)(T_{\alpha,N}c(q-e_i) - T_{\alpha,N}c(q-e_j))\\
    &= -b_id_i+b_jd_j + \sum_k\lambda_k(c(q+e_k-e_i)-c(q+e_k-e_j)) \\
    &\quad + \sum_{k\neq i,j}d_k q_k(c(q-e_k-e_i)-c(q-e_k-e_j))\\
    &\quad + d_i(q_i-1)(c(q-e_i-e_i)-c(q-e_i-e_j)) \\
    &\quad + d_j(q_j-1)(c(q-e_j-e_i)-c(q-e_j-e_j))\\
    &\quad + (Q_N - \sum_k\lambda_k - \sum_{k} d_kq_k-\mu)(c(q-e_i)-c(q-e_j)) \\
    &\quad +\mu (\Gamma c(q-e_i) - \Gamma c(q-e_j)) -d_ic(q-e_i-e_j)+d_jc(q-e_i-e_j)+d_ic(q-e_i)-d_jc(q-e_j)\\
    &\le -b_id_i+b_jd_j+d_i(c(q-e_i)-c(q-e_i-e_j))-d_j(c(q-e_j)-c(q-e_i-e_j))\\
    &\quad +\mu (\Gamma c(q-e_i) - \Gamma c(q-e_j))\\
    &\le -b_i d_i+b_j d_j
    +(d_i-d_j)\bigl(c(q-e_i)-c(q-e_i-e_j)\bigr)
    +\mu\bigl(\Gamma c(q-e_i)-\Gamma c(q-e_j)\bigr).
\end{align*}
By Lemma \ref{lmm:Gamma_monotone}, the last term is non-positive.
By Assumption \ref{asmp:d_order} and (P2),
\begin{align*}
    -b_i d_i+b_j d_j
    +(d_i-d_j)\bigl(c(q-e_i)-c(q-e_i-e_j)\bigr)
    &\le -b_i d_i+b_j d_j+(d_i-d_j)b_j\\
    &= d_i(-b_i+b_j)\le 0.
\end{align*}
Hence
$$
T_{\alpha,N}c(q-e_i)\le T_{\alpha,N}c(q-e_j).
$$

If $N(q)=N+1$, then all terms except patient arrivals are the same. The patient arrival term for class $k$ is:
\[\lambda_k (c(q-e_i+e_k-e_{\min J(q-e_i+e_k)})-c(q-e_j+e_k-e_{\min J(q-e_j+e_k)}))\]
Let $m=\min J(q+e_k)$. Since $i,j\in J(q)\subset J(q+e_k)$, $m\le i< j$.
If $m<i$ or $q_i\ge 2$, then $m = \min J(q-e_i+e_k) = \min J(q-e_j+e_k)$. In this case, by (P1),
\[\lambda_k (c(q+e_k-e_m-e_i)-c(q+e_k-e_m-e_j))\le 0.\]
If $m=i$ and $q_i=1$, then $\min J(q-e_j+e_k) = i$ and the term is:
\[\lambda_k (c(q-e_i+e_k-e_{\min J(q-e_i+e_k)})-c(q-e_j+e_k-e_i)).\]
Since $j\in J(q-e_i+e_k)$, $\min J(q-e_i+e_k) \le j$.
If $\min J(q-e_i+e_k) < j$, the term is non-positive by (P1).
If $\min J(q-e_i+e_k) = j$, the term is $\lambda_k (c(q-e_i+e_k-e_j)-c(q-e_j+e_k-e_i))=0$.
Thus, the arrival term is non-positive for all $k$. The rest follows as in the case $N(q)\le N$.

\vskip\baselineskip
Next, we show (P2). Fix $i$ and $q$ with $q+e_i\in S_N$.

If $N(q+e_i)<N$, then by (P2),
\begin{align*}
    &\quad (\alpha+Q_N)(T_{\alpha,N}c(q+e_i) - T_{\alpha,N}c(q))\\
    &= b_id_i + \sum_k\lambda_k(c(q+e_i+e_k)-c(q+e_k)) \\
    &\quad + \sum_{k}d_k q_k(c(q+e_i-e_k)-c(q-e_k))\\
    &\quad + (Q_N - \sum_k\lambda_k - \sum_{k\neq i} d_k q_k -d_i(q_i+1) -\mu)(c(q+e_i)-c(q)) \\
    &\quad +\mu (\Gamma c(q+e_i) - \Gamma c(q))\\
    &\le b_id_i + \sum_k\lambda_kb_i + \sum_{k}d_k q_k b_i\\
    &\quad + (Q_N - \sum_k\lambda_k - \sum_{k\neq i} d_k q_k -d_i(q_i+1) -\mu)b_i +\mu (\Gamma c(q+e_i) - \Gamma c(q))\\
    &= (Q_N - \mu)b_i +\mu (\Gamma c(q+e_i) - \Gamma c(q)).
\end{align*}
By Lemma \ref{lmm:Gamma_monotone},
$$
(\alpha+Q_N)\bigl(T_{\alpha,N}c(q+e_i)-T_{\alpha,N}c(q)\bigr)\le Q_N b_i.
$$
Hence
$$
T_{\alpha,N}c(q+e_i)-T_{\alpha,N}c(q)\le \frac{Q_N}{\alpha+Q_N}b_i\le b_i.
$$

If $N(q+e_i)=N$, the patient-arrival term for class $k$ is:
\[\lambda_k (c(q+e_i+e_k-e_{\min J(q+e_i+e_k)})-c(q+e_k)).\]
Let $m=\min J(q+e_i+e_k)$. Since $i\in J(q+e_i+e_k)$, $m\le i$.
If $m<i$, by (P1), $\lambda_k (c(q+e_i+e_k-e_m)-c(q+e_k))\le 0$.
If $m=i$, $\lambda_k (c(q+e_k)-c(q+e_k))=0$.
The arrival term is non-positive, and thus no greater than $\sum_k\lambda_kb_i$. The rest follows as in the case $N(q+e_i)<N$.

Therefore $T_{\alpha,N}c$ satisfies both (P1) and (P2).
\end{proof}

\subsection{Proof of Lemma \ref{lmm:nondecr_C}}
\begin{proof}
Fix $q$.
Compare the systems with capacities $N$ and $N+1$ started from the same initial state $q$.
Couple all patient-arrival clocks and abandonment clocks across the two systems.
At each organ-arrival epoch, couple compatibility realizations patient-by-patient on all patients common to both systems.
If one system contains one additional patient not present in the other, assign an independent compatibility realization to that extra patient.

Let System $1$ (capacity $N+1$) follow its optimal policy $\sigma_{N+1}$.
Let System $2$ (capacity $N$) mimic System $1$ as follows: whenever System $1$ allocates an organ to a patient that is also present in System $2$, System $2$ allocates to that same patient; if System $1$ allocates to a patient that is absent from System $2$, then System $2$ discards the organ.

Initially, the two systems are identical.
The first discrepancy can only arise when a patient arrives while System $2$ is at capacity $N$.
At that moment, System $2$ ejects one of the highest-priority patients, while System $1$ keeps that patient.
Hence, System $1$ differs from System $2$ by exactly one additional patient.

From that point onward, after every event, the state of System $1$ is either identical to the state of System $2$ or differs from it by exactly one additional patient:
\begin{enumerate}
    \item Under a patient arrival, both systems receive the same arrival.
    \begin{enumerate}
        \item If System $2$ is not full, System 1 is also not full, and the difference remains at most one patient.
        \item If System $2$ is full, it ejects one patient. Let it be of class $j$. If System $1$ is not full, System $1$ now differs from System $2$ by an additional class-$j$ patient. If System $1$ has the unique additional class-$i$ patient and $i>j$, System $1$ also ejects $j$ (as $j$ also exists in System $1$), and states remain identical except for the extra class $i$ patient. If $i\le j$, System $1$ ejects $i$, and states are now identical except System $1$ has an additional $j$.
    \end{enumerate}
    \item Under an abandonment of a common patient, both systems lose that patient simultaneously. Under an abandonment of the extra patient, the systems become identical again.
    \item Under an organ arrival, if System $1$ allocates to a common patient, then System $2$ can mimic that allocation and the difference remains unchanged; if System $1$ allocates to the extra patient, then System $2$ discards and the systems become identical; if no compatible patient exists in System $1$, then neither system allocates.
\end{enumerate}

Therefore, along every coupled sample path, System $1$ is always in the same state as System $2$ or in a state with one extra patient of some class.
Since the instantaneous cost $\sum_i b_i d_i q_i$ is nondecreasing in queue lengths, System $1$ incurs weakly greater cost rate at all times.
Hence the total discounted cost of System $2$ under the mimicking policy is weakly smaller than the total discounted cost of System $1$ under $\sigma_{N+1}$.
Since the optimal policy for capacity $N$ does at least as well as the mimicking policy,
$$
C_{\alpha,N}^*(q)\le C_{\alpha,N+1}^*(q).
$$
\end{proof}

\subsection{Proof of Lemma \ref{lmm:convergence}}

\begin{proof}
Fix $q$.
For all sufficiently large $N$ such that $N(q)<N$, the optimality equation for $C_{\alpha,N}^*$ can be written as
$$
(\alpha+\widehat Q(q))C_{\alpha,N}^*(q)
=
\sum_i b_i d_i q_i
+\sum_i \lambda_i C_{\alpha,N}^*(q+e_i)
+\sum_i d_i q_i C_{\alpha,N}^*(q-e_i)
+\mu \Gamma C_{\alpha,N}^*(q).
$$
Taking the limit as $N\to\infty$, the left-hand side converges to
$$
(\alpha+\widehat Q(q))C_{\alpha,\infty}^*(q),
$$
and each term on the right-hand side converges pointwise.
Hence
$$
(\alpha+\widehat Q(q))C_{\alpha,\infty}^*(q)=T C_{\alpha,\infty}^*(q),
$$
so $C_{\alpha,\infty}^*$ satisfies the discounted optimality equation.

Since $C_{\alpha,N}^*$ satisfies (P1) for every $N$, we have
$$
C_{\alpha,N}^*(q-e_i)\le C_{\alpha,N}^*(q-e_j)
$$
for all $i<j$ and all $q$ with $i,j\in J(q)$.
Taking limits gives (P1) for $C_{\alpha,\infty}^*$.
\end{proof}

\section{Other Omitted Proofs}\label{app:proofs}
\subsection{Proof of Lemma \ref{lmm:ergodic}}\label{app:proof_ergodic}
\begin{proof}
Fix any stationary policy $\sigma$.
We show that the resulting queue-length process is ergodic.
A countable-state Markov process is ergodic if it is irreducible and positive recurrent \citep{bremaud2013markov}.

Irreducibility holds because any state can reach the empty state $q_0$ with positive probability through patient abandonments, and from $q_0$ any finite state can be reached with positive probability through patient arrivals.

A process is positive recurrent if for every state $s\in S$, starting from $s$, the expected time until the process returns to state $s$ is finite.
To prove positive recurrence, we apply the Foster--Lyapunov criterion~\citep{foster1953stochastic} with the Lyapunov function
$$
V(q)=N(q)=\sum_i q_i.
$$
With a Lyapunov function $V(q)$, a sufficient condition for the Foster--Lyapunov criterion is that there exist $M,b\ge 0$ and $c>0$ such that
the set $\{q\in S|V(q)\le M\}$ is finite,
$V(q)\le M\implies \mathcal{D}V(q)< b,$
and
$V(q)>M\implies \mathcal{D}V(q)\le -c,$
where the drift $\mathcal{D}V(q)$ is defined as $\mathcal{D}V(q) = \sum_{q' \neq q} \nu^{\sigma}(q, q') (V(q') - V(q))$ \citep{bremaud2013markov}.

Let
$$
\eta^\sigma(q)=\sum_i \bar{\sigma}_i(q)\in[0,1]
$$
denote the ex-ante probability that an arriving organ is allocated in state $q$ under policy $\sigma$.
The drift of $V$ is
\begin{align}
    \mathcal{D}V(q)
    &= \sum_{q'\neq q}\nu^\sigma(q,q')\bigl(N(q')-N(q)\bigr)\\
    &= \sum_i \lambda_i(+1)+\sum_i d_i q_i(-1)+\mu\sum_i \bar{\sigma}_i(q)(-1)\\
    &= \left(\sum_i \lambda_i\right)-\left(\sum_i d_i q_i\right)-\mu \eta^\sigma(q).
\end{align}
Let
$$
\lambda=\sum_i \lambda_i,
\qquad
\underline d=\min_i d_i>0.
$$
Since $\eta^\sigma(q)\ge 0$, we obtain the upper bound
$$
\mathcal{D}V(q)\le \lambda-\underline d\,N(q).
$$
Therefore, whenever
$$
N(q)>\frac{\lambda+1}{\underline d},
$$
we have
$$
\mathcal{D}V(q)<-1.
$$
On the finite set
$$
\left\{q\in S: N(q)\le \frac{\lambda+1}{\underline d}\right\},
$$
the drift is bounded above.
Hence the Foster--Lyapunov criterion implies positive recurrence.

Thus the process is irreducible and positive recurrent, and therefore ergodic.
\end{proof}

\subsection{Proof of Proposition \ref{prop:achievable_set}}\label{app:proof_prop_achievable}
\begin{proof}
We first establish convexity.
\begin{lemma}\label{lmm:convexity}
    $\mathcal U$ is a convex set.
\end{lemma}
\begin{proof}
Fix a pair of policies $\sigma_1,\sigma_2\in\Sigma$ and $p\in[0,1]$.
We construct a stationary policy $\sigma^p\in\Sigma$ that achieves
$$
u(\sigma^p)=p\,u(\sigma_1)+(1-p)\,u(\sigma_2).
$$

For any policy $\sigma$, define the compatibility-resolved occupancy measure
$$
\rho^\sigma(q,A,i)
=
\pi^\sigma(q)\,p_q(A)\,\sigma_i(q,A),
$$
where $q\in S$, $A\subseteq J(q)$, and $i\in A$.
Then, by the definition of throughput,
$$
u_i(\sigma)
=
\frac{v_i\mu}{\lambda_i}
\sum_{q\in S}\sum_{A\subseteq J(q):\, i\in A}\rho^\sigma(q,A,i).
$$

Let $\pi_1,\pi_2$ be the stationary distributions under $\sigma_1,\sigma_2$.
Set
$$
\pi_p=p\pi_1+(1-p)\pi_2,
$$
and
$$
\rho_p(q,A,i)
=
p\,\rho^{\sigma_1}(q,A,i)+(1-p)\,\rho^{\sigma_2}(q,A,i).
$$
For every $q$ and every $A\subseteq J(q)$ with $p_q(A)>0$, define
$$
\sigma_i^p(q,A)
=
\frac{\rho_p(q,A,i)}{\pi_p(q)\,p_q(A)}.
$$
If $p_q(A)=0$, define $\sigma^p(q,A)$ arbitrarily subject to feasibility, since such events never occur.
Because patient arrivals and abandonments make every state reachable with positive stationary probability under every stationary policy, we have $\pi_p(q)>0$ for all $q$.

The policy $\sigma^p$ is feasible because $\sigma_i^p(q,A)=0$ for $i\notin A$, and
$$
\sum_{i\in A}\sigma_i^p(q,A)
=
\frac{\sum_{i\in A}\rho_p(q,A,i)}{\pi_p(q)\,p_q(A)}
\le 1,
$$
since each underlying policy is feasible.

Define the ex-ante allocation probabilities
$$
\bar\sigma_i^\ell(q)
=
\sum_{A\subseteq J(q):\, i\in A} p_q(A)\sigma_i^\ell(q,A),
\qquad
\ell\in\{1,2,p\}.
$$
Then
$$
\pi_p(q)\bar\sigma_i^p(q)
=
\sum_{A\subseteq J(q):\, i\in A}\rho_p(q,A,i)
=
p\,\pi_1(q)\bar\sigma_i^1(q)+(1-p)\,\pi_2(q)\bar\sigma_i^2(q).
$$

We now show that $\pi_p$ is the stationary distribution for $\sigma^p$.
Split transitions into a policy-independent part $\nu_0$ and a policy-dependent part $\nu_1^\sigma$.
The policy-independent part consists of patient arrivals and abandonments:
$$
\nu_0(q,q')
=
\begin{cases}
\lambda_i & \text{if } q'=q+e_i,\\
d q_i & \text{if } q'=q-e_i \text{ for } i\in J(q),\\
-\left(\sum_i \lambda_i+\sum_i d q_i\right) & \text{if } q'=q,\\
0 & \text{otherwise.}
\end{cases}
$$
The policy-dependent part consists of organ allocations:
$$
\nu_1^\sigma(q,q')
=
\begin{cases}
\mu \bar\sigma_i(q) & \text{if } q'=q-e_i \text{ for } i\in J(q),\\
-\mu\sum_i \bar\sigma_i(q) & \text{if } q'=q,\\
0 & \text{otherwise.}
\end{cases}
$$
Hence
$$
\nu^\sigma(q,q')=\nu_0(q,q')+\nu_1^\sigma(q,q').
$$

The total stationary flow into state $q$ under $(\pi_p,\sigma^p)$ is
$$
\sum_{q'\in S}\pi_p(q')\nu^{\sigma^p}(q',q)
=
\sum_{q'\in S}\pi_p(q')\nu_0(q',q)
+
\sum_{q'\in S}\pi_p(q')\nu_1^{\sigma^p}(q',q).
$$
Using the linearity of $\pi_p$ and the identity for $\pi_p(q)\bar\sigma_i^p(q)$,
$$
\sum_{q'\in S}\pi_p(q')\nu^{\sigma^p}(q',q)
=
p\sum_{q'\in S}\pi_1(q')\nu^{\sigma_1}(q',q)
+
(1-p)\sum_{q'\in S}\pi_2(q')\nu^{\sigma_2}(q',q)
=
0.
$$
Thus $\pi_p$ satisfies the global balance equations for $\sigma^p$ and is therefore its stationary distribution.

Finally,
$$
u_i(\sigma^p)
=
\frac{v_i\mu}{\lambda_i}
\sum_{q,A:\, i\in A}\rho_p(q,A,i)
=
p\,u_i(\sigma_1)+(1-p)\,u_i(\sigma_2)
$$
for each class $i$.
Hence $\mathcal U$ is convex.
\end{proof}

Convexity is achieved not by mixing two policies ex ante (which would not generate a single steady-state distribution), but by the standard mixing of occupancy measures.

Let $S_I$ denote the set on the right-hand side of the proposition.
Any utility vector that is componentwise weakly below an achievable vector can be attained by stochastically discarding organs that would otherwise be allocated to the corresponding classes.
This and convexity of $\mathcal U$ imply that $S_I\subseteq \mathcal U$.

Suppose there exists $u^*\in \mathcal U\setminus S_I$.
By the separating hyperplane theorem and the shape of $S_I$, there exists $w\in\mathbb R_+^K$ such that
$$
w\cdot u^*>w\cdot u
\qquad
\text{for all } u\in S_I.
$$
Since $S_I$ contains every utility vector generated by a static index policy, this contradicts Theorem \ref{thm:index_optimality}.
Hence $\mathcal U\subseteq S_I$.
\end{proof}

\subsection{Proof of Lemma \ref{lmm:index_util}}\label{app:proof_index_util}
\begin{proof}
Fix a static priority ordering $1,\dots,K$.
For each $k\in\{1,\dots,K\}$, let
$$
N_k=\sum_{i=1}^k q_i
$$
denote the total number of waiting patients in the top $k$ priority classes.

Under the index policy, if an organ arrives when $N_k=n$, then the organ is allocated to one of the top $k$ classes if and only if at least one of those $n$ patients is compatible with the organ.
Because compatibility is independent across patients and homogeneous across classes, this event occurs with probability
$$
g_\rho(n)=1-(1-\rho)^n.
$$
Therefore, $N_k$ is a one-dimensional birth--death process with birth rate
$$
\nu_k(n,n+1)=\Lambda_k
$$
and death rate
$$
\nu_k(n,n-1)=dn+\mu g_\rho(n)
\qquad \text{for } n\ge 1.
$$

Let $\pi_k(n)$ denote the stationary distribution of $N_k$.
The detailed balance equations are
$$
\Lambda_k \pi_k(n)
=
\bigl(d(n+1)+\mu g_\rho(n+1)\bigr)\pi_k(n+1)
\qquad \text{for all } n\ge 0.
$$
Hence, for each $n\ge 1$,
$$
\pi_k(n)
=
\pi_k(0)\prod_{m=1}^n \frac{\Lambda_k}{dm+\mu g_\rho(m)}.
$$
Normalization yields
$$
\pi_k(0)
=
\left[
1+\sum_{n=1}^\infty \prod_{m=1}^n \frac{\Lambda_k}{dm+\mu g_\rho(m)}
\right]^{-1}.
$$
This is exactly $\Pi_0^\rho(\Lambda_k,\mu)$, and more generally
$$
\pi_k(n)=\Pi_n^\rho(\Lambda_k,\mu)
\qquad \text{for all } n\ge 0.
$$

The cumulative throughput to the top $k$ classes equals the organ-arrival rate times the probability that at least one of those classes is compatible:
$$
\mu\sum_{n=1}^\infty \pi_k(n) g_\rho(n)
=
\mu\sum_{n=1}^\infty \Pi_n^\rho(\Lambda_k,\mu) g_\rho(n)
=
\mathcal T^\rho(\Lambda_k,\mu).
$$
With the convention $\Lambda_0=0$, the throughput of class $k$ is
$$
\mathcal T^\rho(\Lambda_k,\mu)-\mathcal T^\rho(\Lambda_{k-1},\mu).
$$
Therefore,
$$
u_k
=
\frac{v_k}{\lambda_k}
\left(
\mathcal T^\rho(\Lambda_k,\mu)-\mathcal T^\rho(\Lambda_{k-1},\mu)
\right).
$$
This proves the lemma.
\end{proof}

\subsection{Proof of Observation \ref{obs:pareto_frontier_improvement}}\label{app:proof_obs1}
Throughout, types are indexed so that $v_1\ge v_2\ge\cdots\ge v_{|\Theta|}$, and we set $v_{|\Theta|+1}=0$. Assumption \ref{asmp:same_d} is in force, so every class abandons at the common rate $d$.

\begin{proof}
Let $\sigma_{sep}$ be the stationary policy that, at each organ arrival, draws a nominal origin region $r$ with probability $\mu^r/\mu$ independently of everything else, and allocates the organ to the highest-value compatible class of region $r$, discarding it if region $r$ has no compatible waiting patient.
Because organ arrivals from the different regions are independent Poisson processes, the nominal origins reproduce the law of the true origins, so under $\sigma_{sep}$ each region faces exactly its autarky system and applies its autarky-optimal type-based priority rule. Hence $u^r(\sigma_{sep})=u_{sep}^r$ for every $r\in\mathcal R$.

Fix a region $s\in\mathcal R$ and let $\sigma_s$ be the policy obtained from $\sigma_{sep}$ by redirecting the organs that $\sigma_{sep}$ discards: when the nominal origin $r$ has no compatible waiting patient, the organ is offered to the highest-value compatible class of region $s$, and is discarded only if region $s$ has none either.
Both policies depend on the state only through the queue vector and the realized compatible set, so both belong to $\Sigma$ and the utility identity of Section \ref{subsec:flow_conv} applies to them.

Run the two systems on a single probability space, with the same patient-arrival processes, organ-arrival process, nominal-origin draws, and abandonment clocks.

Consider first a region $r\neq s$. At every organ arrival with nominal origin $r$, both policies offer the organ to $r$'s highest-value compatible class; at every other arrival, neither policy allocates anything to $r$. Region $r$'s queues therefore follow the same path under $\sigma_s$ as under $\sigma_{sep}$, and
$$
u^r(\sigma_s)=u^r(\sigma_{sep})=u_{sep}^r,
\qquad r\neq s.
$$

For region $s$, let $\mathcal O$ denote the set of organ-arrival epochs whose nominal origin is $s$, and let $\mathcal E$ denote the set of epochs whose nominal origin is some $r\neq s$ that has no compatible waiting patient.
Region $s$ is offered organs at the epochs of $\mathcal O$ under $\sigma_{sep}$, and at the epochs of $\mathcal O\cup\mathcal E$ under $\sigma_s$.

For $k=1,\dots,|\Theta|$ let
$$
N_k(t)=\sum_{\theta\le k} q_{(s,\theta)}(t)
$$
be the number of region-$s$ patients waiting in the $k$ highest-value types.
Both policies give an organ offered to region $s$ to its highest-value compatible class, so the top-$k$ block loses a patient at an offer epoch exactly when at least one of its $N_k$ patients is compatible, an event of probability $g_\rho(N_k)$ that is independent across offers.
Given offer epochs, the process $N_k$ therefore evolves as follows: it increases by one at rate $\Lambda^s_{k}$, decreases by one at rate $dN_k$ through abandonment, and decreases by one at each offer epoch with probability $g_\rho(N_k)$.

Write $N_k^{s}$ and $N_k^{sep}$ for this process under $\sigma_s$ and under $\sigma_{sep}$, started from the same state, and couple them as follows.
Births occur simultaneously. Abandonments are coupled so that the larger count abandons whenever the smaller one does. At an epoch of $\mathcal O$ the same uniform draw is used in both systems, so that a service in the system with the smaller count implies a service in the other, which is consistent because $g_\rho$ is nondecreasing. Epochs of $\mathcal E$ act on $N_k^{s}$ only.
Both counts move one step at a time. At a state with $N_k^{s}\le N_k^{sep}$, none of the transition sources of $N_k$ can reverse the inequality: a birth raises both, a coupled abandonment or a coupled service lowers the larger by at least as much as the smaller, and an epoch of $\mathcal E$ can only lower $N_k^{s}$. Hence
$$
N_k^{s}(t)\le N_k^{sep}(t)
\qquad\text{for all } t\ge 0 \text{ and all } k,
$$
and letting $t\to\infty$ and using ergodicity (Lemma \ref{lmm:ergodic}),
$$
E_{\pi^{\sigma_s}}[N_k]\le E_{\pi^{\sigma_{sep}}}[N_k]
\qquad\text{for all } k.
$$

By the utility identity of Section \ref{subsec:flow_conv} and the definition of the regional utility, for any policy $\sigma$,
$$
\Lambda^s u^s(\sigma)
=
\sum_{\theta}\lambda_{(s,\theta)}v_\theta
-
d\sum_{\theta} v_\theta L_{(s,\theta)}^\sigma
=
\sum_{\theta}\lambda_{(s,\theta)}v_\theta
-
d\sum_{k=1}^{|\Theta|}(v_k-v_{k+1})\,E_{\pi^\sigma}[N_k],
$$
where the second equality is summation by parts.
The coefficients $v_k-v_{k+1}$ are nonnegative, so region $s$'s utility is a constant minus a nonnegatively weighted combination of the quantities $E_{\pi^\sigma}[N_k]$. Together with the previous display this gives $u^s(\sigma_s)\ge u^s(\sigma_{sep})=u_{sep}^s$.

The inequality is strict. Take $k=|\Theta|$, whose coefficient $v_{|\Theta|}$ is positive.
With positive stationary probability, an organ arrives with a nominal origin $r\neq s$ that has no compatible waiting patient while region $s$ does have one; the two events concern disjoint collections of independent arrival, abandonment, and compatibility variables, and each has positive probability because all arrival rates are positive and $\rho>0$.
At such an epoch $N_{|\Theta|}^{s}$ falls while $N_{|\Theta|}^{sep}$ does not, and the coupled process reaches this configuration from any state with positive probability, so the event $\{N_{|\Theta|}^{s}<N_{|\Theta|}^{sep}\}$ has positive stationary probability.
Hence $E_{\pi^{\sigma_s}}[N_{|\Theta|}]<E_{\pi^{\sigma_{sep}}}[N_{|\Theta|}]$ and $u^s(\sigma_s)>u_{sep}^s$.

Therefore, for any $w\in\mathbb R_{++}^{|\mathcal R|}$,
$$
\max_{\sigma\in\Sigma}\sum_{r\in\mathcal R} w^r u^r(\sigma)
\ge
\sum_{r\in\mathcal R} w^r u^r(\sigma_s)
>
\sum_{r\in\mathcal R} w^r u_{sep}^r,
$$
because $u^r(\sigma_s)=u^r_{sep}$ for $r\neq s$ and $u^s(\sigma_s)>u^s_{sep}$.

The construction also yields the strict Pareto improvement.
By Lemma \ref{lmm:convexity} the achievable set is convex, so the utility vector $\bar u=|\mathcal R|^{-1}\sum_{s\in\mathcal R}u(\sigma_s)$ is achievable, and for every region $r$,
$$
\bar u^r
=
\frac{1}{|\mathcal R|}\Bigl(u^r(\sigma_r)+\sum_{s\neq r}u^r(\sigma_s)\Bigr)
=
\frac{1}{|\mathcal R|}u^r(\sigma_r)+\frac{|\mathcal R|-1}{|\mathcal R|}u_{sep}^r
>
u_{sep}^r .
$$
\end{proof}

\subsection{Proof of Proposition \ref{prop:suff_cond_pareto}}\label{app:proof_suff_cond}
\begin{proof}
We prove the proposition by constructing the proportional utilitarian target and showing that it Pareto dominates autarky.

We first show the following lemma.
\begin{lemma}\label{lmm:order_Urho}
If $\mu\ge \mu'$ and
$$
\frac{\lambda'}{\lambda}\le \frac{d+\rho\mu'}{d+\rho\mu},
$$
then
$$
\frac{\mathcal T^\rho(\lambda,\mu)}{\mu}\ge
\frac{\mathcal T^\rho(\lambda',\mu')}{\mu'}.
$$
\end{lemma}

\begin{proof}
For $n\ge 1$, define
$$
\beta_n=\frac{\lambda}{dn+\mu g_\rho(n)},
\qquad
\beta_n'=\frac{\lambda'}{dn+\mu' g_\rho(n)}.
$$
The sequence
$$
a_n=\frac{g_\rho(n)}{n}
$$
is non-increasing in $n\ge 1$, with $a_1=\rho$.

Because $\mu\ge \mu'$, the function
$$
a\mapsto \frac{d+a\mu'}{d+a\mu}
$$
is non-increasing in $a$.

Hence, for every $n\ge 1$,
$$
\frac{dn+g_\rho(n)\mu'}{dn+g_\rho(n)\mu}
=
\frac{d+a_n\mu'}{d+a_n\mu}
\ge
\frac{d+\rho\mu'}{d+\rho\mu}.
$$
By the assumption, it follows that
$$
\frac{\lambda'}{\lambda}
\le
\frac{dn+\mu' g_\rho(n)}{dn+\mu g_\rho(n)}
\qquad \text{for all } n\ge 1,
$$
or equivalently,
$$
\beta_n\ge \beta_n'
\qquad \text{for all } n\ge 1.
$$
Hence,
$$
\frac{\Pi_n^\rho(\lambda,\mu)}{\Pi_n^\rho(\lambda',\mu')}
=
\frac{\Pi_0^\rho(\lambda,\mu)}{\Pi_0^\rho(\lambda',\mu')}
\prod_{m=1}^n \frac{\beta_m}{\beta_m'}
$$
is nondecreasing in $n$.

Since $g_\rho(n)$ is nondecreasing in $n$, this implies
$$
\sum_{n=1}^\infty \Pi_n^\rho(\lambda,\mu)g_\rho(n)
\ge
\sum_{n=1}^\infty \Pi_n^\rho(\lambda',\mu')g_\rho(n).
$$
That is,
$$
\frac{\mathcal T^\rho(\lambda,\mu)}{\mu}\ge
\frac{\mathcal T^\rho(\lambda',\mu')}{\mu'}.
$$
\end{proof}

By Condition 1 and Lemma \ref{lmm:order_Urho}, we have
$$
\frac{\mathcal T^\rho(\Lambda_{\theta},\mu)}{\mu}\ge \frac{\mathcal T^\rho(\Lambda^r_{\theta},\mu^r)}{\mu^r}
$$
for all $r$ and $\theta$.
Therefore, 
$$
\mathcal T^r_{\mathrm{target},\theta}
=
\frac{\mu^r}{\mu}\mathcal T^\rho(\Lambda_{\theta},\mu)
\ge
\mathcal T^r_{\mathrm{sep},\theta}.
$$

By Condition 2, there exists a utilitarian policy $\sigma^*$ that achieves this target cumulative throughput $\mathcal T^r_{\mathrm{target},\theta}$.

For each region $r$, the regional utility is
$$
u^r
=
\frac{1}{\Lambda^r}\sum_{\theta=1}^{|\Theta|} v_\theta \bigl(\mathcal T^r_{\mathrm{target},\theta}-\mathcal T^r_{\mathrm{target},\theta-1}\bigr)
=
\frac{1}{\Lambda^r}\sum_{\theta=1}^{|\Theta|} (v_\theta-v_{\theta+1})\mathcal T^r_{\mathrm{target},\theta},
$$
where $v_{|\Theta|+1}=0$.
Since $v_\theta-v_{\theta+1}\ge 0$ and
$$
\mathcal T^r_{\mathrm{target},\theta}\ge \mathcal T^r_{\mathrm{sep},\theta}
\qquad
\text{for all } \theta,
$$
it follows that
$$
u^r(\sigma^*)\ge u_{sep}^r
\qquad
\text{for all } r\in\mathcal R.
$$
Hence $\sigma^*$ Pareto improves upon autarky.
\end{proof}

\subsection{Proof of Lemma \ref{lmm:index_eval}}\label{app:proof_index_eval}
\begin{proof}
Fix a static priority ordering and let
$$
\Lambda_k=\sum_{i=1}^k \lambda_i,
\qquad k=1,\dots,K.
$$
For each $k$, define the recursion
$$
r_0^{(k)}=1,
\qquad
r_n^{(k)} = r_{n-1}^{(k)}\frac{\Lambda_k}{dn+\mu g_\rho(n)}
\quad \text{for } n\ge 1.
$$
Then
$$
\Pi_n^\rho(\Lambda_k,\mu)
=
\frac{r_n^{(k)}}{\sum_{m=0}^\infty r_m^{(k)}},
\qquad
\mathcal T^\rho(\Lambda_k,\mu)
=
\mu\frac{\sum_{n=1}^\infty r_n^{(k)} g_\rho(n)}{\sum_{n=0}^\infty r_n^{(k)}}.
$$

For a truncation level $N$, define
$$
Z_N^{(k)}:=\sum_{n=0}^N r_n^{(k)},
\qquad
S_N^{(k)}:=\sum_{n=1}^N r_n^{(k)} g_\rho(n),
$$
and the truncated throughput
$$
\widehat{\mathcal T}_N^\rho(\Lambda_k,\mu)
:=
\mu \frac{S_N^{(k)}}{Z_N^{(k)}}.
$$
Let
$$
R_N^{(k)}:=\sum_{n=N+1}^\infty r_n^{(k)}.
$$
Since $0\le g_\rho(n)\le 1$, we have
$$
0\le \sum_{n=N+1}^\infty r_n^{(k)} g_\rho(n)\le R_N^{(k)}.
$$
Also, $Z_N^{(k)}\ge 1$ and $Z_\infty^{(k)}=Z_N^{(k)}+R_N^{(k)}$.
Hence
\begin{align*}
\left|
\mathcal T^\rho(\Lambda_k,\mu)-\widehat{\mathcal T}_N^\rho(\Lambda_k,\mu)
\right|
&=
\mu\left|
\frac{S_\infty^{(k)}}{Z_\infty^{(k)}}-\frac{S_N^{(k)}}{Z_N^{(k)}}
\right| \\
&\le 2\mu R_N^{(k)}.
\end{align*}

Now,
$$
r_n^{(k)}
=
\prod_{m=1}^n \frac{\Lambda_k}{dm+\mu g_\rho(m)}
\le
\prod_{m=1}^n \frac{\Lambda_k}{dm}
=
\frac{(\Lambda_k/d)^n}{n!}.
$$
Therefore, with $\bar\lambda=\max_i\lambda_i$ as in Lemma \ref{lmm:index_eval},
$$
R_N^{(k)}
\le
\sum_{n=N+1}^\infty \frac{(\Lambda_k/d)^n}{n!}
\le
\sum_{n=N+1}^\infty \frac{(K\bar\lambda/d)^n}{n!}.
$$
Thus, if $N$ is chosen so that
$$
\sum_{n=N+1}^\infty \frac{(K\bar\lambda/d)^n}{n!}
\le
\frac{\delta\,\underline\lambda}{4\bar v\,\mu},
$$
where
$$
\underline\lambda:=\min_i\lambda_i,
\qquad
\bar v:=\max_i v_i,
$$
then
$$
\left|
\mathcal T^\rho(\Lambda_k,\mu)-\widehat{\mathcal T}_N^\rho(\Lambda_k,\mu)
\right|
\le
\frac{\delta\,\underline\lambda}{2\bar v}
\qquad
\text{for every } k.
$$

Now define the approximate class utilities by
$$
\widehat u_k
=
\frac{v_k}{\lambda_k}
\left(
\widehat{\mathcal T}_N^\rho(\Lambda_k,\mu)
-
\widehat{\mathcal T}_N^\rho(\Lambda_{k-1},\mu)
\right).
$$
By Lemma \ref{lmm:index_util},
\begin{align*}
|u_k-\widehat u_k|
&\le
\frac{v_k}{\lambda_k}
\left(
\left|
\mathcal T^\rho(\Lambda_k,\mu)-\widehat{\mathcal T}_N^\rho(\Lambda_k,\mu)
\right|
+
\left|
\mathcal T^\rho(\Lambda_{k-1},\mu)-\widehat{\mathcal T}_N^\rho(\Lambda_{k-1},\mu)
\right|
\right) \\
&\le
\frac{\bar v}{\underline\lambda}\cdot
2\cdot
\frac{\delta\,\underline\lambda}{2\bar v}
= \delta.
\end{align*}
Hence $\|\widehat u(\sigma)-u(\sigma)\|_\infty\le \delta$.

It remains to choose the truncation level. Write
$$
a=\frac{K\bar\lambda}{d},
\qquad
\eta=\frac{\delta\,\underline\lambda}{4\bar v\,\mu},
$$
so that the requirement displayed above reads $\sum_{n>N}a^n/n!\le\eta$.
Since $n!\ge (n/e)^n$, we have $a^n/n!\le (ea/n)^n$, and for $n\ge N+1$ with $N\ge 2ea$ this is at most $2^{-n}$; hence
$$
\sum_{n=N+1}^\infty\frac{a^n}{n!}
\le
\sum_{n=N+1}^\infty 2^{-n}
=
2^{-N}.
$$
Any
$$
N\ \ge\ \max\left\{2ea,\ \log_2\frac{1}{\eta}\right\}
=
\max\left\{\frac{2eK\bar\lambda}{d},\ \log_2\frac{4\bar v\mu}{\underline\lambda\,\delta}\right\}
$$
therefore attains accuracy $\delta$, so one may take
$$
N=O\left(\frac{K\bar\lambda}{d}+\log\frac{\mu\bar v}{\underline\lambda\,\delta}\right).
$$

Finally, computing all $r_n^{(k)}$ up to level $N$ requires $O(KN)$ arithmetic operations once the priority order is fixed, and sorting requires $O(K\log K)$ operations.
The total cost is $O(K\log K+KN)$, which is the bound stated in the lemma.
\end{proof}

\subsection{Proof of Lemma \ref{lmm:dual_radius}}\label{app:proof_dual_radius}
\begin{proof}
Let $\bar u$ and $s_0$ be as in Assumption \ref{asmp:slater}. For any $\gamma\ge 0$, taking $\bar u$ as a feasible point of the inner maximization gives
$$
g(\gamma)
=
\max_{\sigma\in\Sigma}(w+M^\top\gamma)^\top u(\sigma)-\gamma^\top c
\ge
w^\top\bar u+\gamma^\top(M\bar u-c)
\ge
w^\top\bar u+s_0\|\gamma\|_1,
$$
where the last step uses $\gamma\ge 0$ and $M\bar u-c\ge s_0\mathbf 1$.

Assumption \ref{asmp:slater} is a Slater condition for the linear program $(P)$ over the bounded polytope $\mathcal U$, so strong duality holds, the dual optimum is attained, and $g(\gamma^*)=\mathrm{OPT}$ at any optimal $\gamma^*$.
Every class utility satisfies $u_i=v_i-(d_iv_i/\lambda_i)L_i^\sigma\le v_i$, so
$$
\mathrm{OPT}\le\max_{u\in\mathcal U}w^\top u\le\|w\|_1\,\bar v .
$$
Since $w\ge 0$ and $\bar u\ge 0$, we have $w^\top\bar u\ge 0$, and the two displays give
$$
s_0\|\gamma^*\|_1\le g(\gamma^*)-w^\top\bar u\le \mathrm{OPT}\le\|w\|_1\,\bar v ,
$$
which is the claim.
\end{proof}

\subsection{Proof of Lemma \ref{lmm:dual_props}}\label{app:proof_dual_props}
\begin{proof}
For each fixed $\sigma$ the function $\gamma\mapsto(w+M^\top\gamma)^\top u(\sigma)-\gamma^\top c$ is affine, so $g$, being their pointwise maximum, is convex. Let $\gamma\in\mathcal G$. Since $w+M^\top\gamma\ge 0$, Theorem \ref{thm:index_optimality} implies that $\sigma_\gamma$ attains the maximum defining $g(\gamma)$, so for every $\gamma'$,
$$
g(\gamma')
\ge
(w+M^\top\gamma')^\top u(\sigma_\gamma)-\gamma'^\top c
=
g(\gamma)+\bigl(Mu(\sigma_\gamma)-c\bigr)^\top(\gamma'-\gamma),
$$
which is the subgradient inequality for $\xi(\gamma)$. Since $0\le u_i(\sigma)\le v_i\le\bar v$, every entry of $Mu(\sigma_\gamma)$ lies in $[0,B_M\bar v]$, so $\|\xi(\gamma)\|_\infty\le B_M\bar v+\|c\|_\infty\le C_g$. Applying the subgradient inequality at $\gamma$ and at $\gamma'$ then gives
$
|g(\gamma)-g(\gamma')|\le C_g\|\gamma-\gamma'\|_1
$
for all $\gamma,\gamma'\in\mathcal G$.
\end{proof}

\subsection{Proof of Lemma \ref{lmm:oracle_acc}}\label{app:proof_oracle_acc}
\begin{proof}
Throughout, $\gamma,\gamma'\in\mathcal G=[0,R+1]^L$, so $\gamma\ge 0$ and $\|\gamma\|_\infty\le R+1$. Since $w\ge 0$ and $M\ge 0$,
\begin{equation}\label{eq:weight_norm}
\|w+M^\top\gamma\|_1
\le
\|w\|_1+\sum_{k=1}^K\sum_{\ell=1}^L M_{\ell,k}\gamma_\ell
\le
\|w\|_1+(R+1)LB_M
=C_w .
\end{equation}

(i) Since $\sigma_\gamma$ is selected exactly and only its evaluation is approximate,
$$
\widehat g(\gamma)-g(\gamma)
=
(w+M^\top\gamma)^\top\bigl(\widehat u(\sigma_\gamma)-u(\sigma_\gamma)\bigr),
$$
so $|\widehat g(\gamma)-g(\gamma)|\le\|w+M^\top\gamma\|_1\,\delta\le C_w\delta$ by \eqref{eq:weight_norm}. Likewise $\widehat\xi(\gamma)-\xi(\gamma)=M(\widehat u(\sigma_\gamma)-u(\sigma_\gamma))$, whose $\ell$-th entry is bounded in absolute value by $\delta\sum_k M_{\ell,k}\le B_M\delta$.

(ii) Combining Lemma \ref{lmm:dual_props} with part (i), for $\gamma'\in\mathcal G$,
\begin{align*}
g(\gamma')
&\ge
g(\gamma)+\xi(\gamma)^\top(\gamma'-\gamma)\\
&\ge
\widehat g(\gamma)-C_w\delta
+\widehat\xi(\gamma)^\top(\gamma'-\gamma)
-\|\widehat\xi(\gamma)-\xi(\gamma)\|_\infty\|\gamma'-\gamma\|_1\\
&\ge
\widehat g(\gamma)+\widehat\xi(\gamma)^\top(\gamma'-\gamma)
-C_w\delta-(R+1)LB_M\,\delta
\ \ge\
\widehat g(\gamma)+\widehat\xi(\gamma)^\top(\gamma'-\gamma)-2C_w\delta.
\end{align*}
\end{proof}

\subsection{Proof of Theorem \ref{thm:algorithm}}\label{app:algo}
\begin{proof}
Write $\delta:=\delta_{\mathrm{orc}}=\epsilon/(48C_w)$, $g_{\min}:=\min_{\gamma\in\mathcal G}g(\gamma)$, and $\mathrm{viol}(u):=\sum_{\ell=1}^L\bigl(c_\ell-(Mu)_\ell\bigr)^+$. Let $\widehat g_{\mathrm{best}}$ denote the smallest value $\widehat g(\gamma_t)$ reported by the oracle during Phase 1.
For $\bar\gamma\in\mathcal G$ let $\mathcal G_\kappa(\bar\gamma):=\bar\gamma+\kappa(\mathcal G-\bar\gamma)$, a copy of $\mathcal G$ scaled by $\kappa$ about $\bar\gamma$; it lies in $\mathcal G$ and has volume $\bigl(\kappa(R+1)\bigr)^L$.

First, a feasibility cut discards no point of $\mathcal G$. If $(\gamma_t)_\ell<0$ the algorithm takes $z_t=-e_\ell$, and the retained halfspace is $\{\gamma\mid\gamma_\ell\ge(\gamma_t)_\ell\}\supseteq\{\gamma\mid\gamma_\ell\ge 0\}\supseteq\mathcal G$; if $(\gamma_t)_\ell>R+1$ it takes $z_t=e_\ell$ and the retained halfspace is $\{\gamma\mid\gamma_\ell\le(\gamma_t)_\ell\}\supseteq\mathcal G$.

Second, $\mathcal G\subseteq\mathcal B_0$, because $\|\gamma-\gamma^\circ\|_2\le\tfrac12(R+1)\sqrt L$ for every $\gamma\in\mathcal G$. The ball $\mathcal B_0$ of radius $\tfrac12(R+1)\sqrt L$ lies in a cube of side $(R+1)\sqrt L$, so $\mathrm{vol}(\mathcal B_0)\le\bigl((R+1)\sqrt L\bigr)^L$. By $\frac{\mathrm{vol}(\mathcal B_{t+1})}{\mathrm{vol}(\mathcal B_t)}<e^{-1/(2L)}$ and the choice of $\bar t$,
$$
\mathrm{vol}(\mathcal B_{\bar t})
<
e^{-\bar t/(2L)}\bigl((R+1)\sqrt L\bigr)^L
<
\Bigl(\frac{\kappa}{\sqrt L}\Bigr)^L\bigl((R+1)\sqrt L\bigr)^L
=
\bigl(\kappa(R+1)\bigr)^L
=
\mathrm{vol}\bigl(\mathcal G_\kappa(\bar\gamma)\bigr),
$$
where the middle inequality is $\bar t>2L^2\ln(\sqrt L/\kappa)$, which holds by the choice of $\bar t$ in both branches of the maximum.

Consequently, if Phase 1 runs all $\bar t$ iterations, some point of $\mathcal G_\kappa(\bar\gamma)$ is discarded along the way: otherwise $\mathcal G_\kappa(\bar\gamma)\subseteq\mathcal B_t$ would hold for every $t\le\bar t$ by induction, contradicting the volume comparison. By the first paragraph, the discarding cut cannot be a feasibility cut, so it is an optimality cut at some iteration $t^*$ with $\gamma_{t^*}\in\mathcal G$, and the discarded point $\gamma\in\mathcal G_\kappa(\bar\gamma)$ satisfies
\begin{equation}\label{eq:discarded}
\widehat\xi(\gamma_{t^*})^\top(\gamma-\gamma_{t^*})>0 .
\end{equation}
In particular, at least one oracle call is made, so $\widehat g_{\mathrm{best}}$ is finite.

Call $\tilde g:\mathcal G\to\mathbb R$ \emph{consistent} if $\tilde g$ is convex, $|\tilde g(\gamma)-\tilde g(\gamma')|\le C_g\|\gamma-\gamma'\|_1$ on $\mathcal G$, and at every point $\gamma_t$ at which the oracle was called during the run, $\tilde g(\gamma_t)=g(\gamma_t)$ and $\xi(\gamma_t)\in\partial\tilde g(\gamma_t)$. We claim that every consistent $\tilde g$ satisfies
\begin{equation}\label{eq:master}
\min_{\gamma\in\mathcal G}\tilde g(\gamma)\ \ge\ \widehat g_{\mathrm{best}}-2C_w\delta-\frac{\epsilon}{16}.
\end{equation}

The argument of Lemma \ref{lmm:oracle_acc}(ii) applies to any consistent $\tilde g$ verbatim: for every queried $\gamma_t$ and every $\gamma'\in\mathcal G$,
\begin{equation}\label{eq:cut_h}
\tilde g(\gamma')\ \ge\ \tilde g(\gamma_t)+\xi(\gamma_t)^\top(\gamma'-\gamma_t)
\ \ge\ \widehat g(\gamma_t)+\widehat\xi(\gamma_t)^\top(\gamma'-\gamma_t)-2C_w\delta,
\end{equation}
where the second inequality uses $\tilde g(\gamma_t)=g(\gamma_t)$, Lemma \ref{lmm:oracle_acc}(i) and $\|\gamma'-\gamma_t\|_1\le(R+1)L$.

If the loop left early at some $\gamma_t$ with $\widehat\xi(\gamma_t)=0$, then \eqref{eq:cut_h} reads $\tilde g(\gamma')\ge\widehat g(\gamma_t)-2C_w\delta$ for every $\gamma'\in\mathcal G$, and $\widehat g_{\mathrm{best}}\le\widehat g(\gamma_t)$, which gives \eqref{eq:master}.

Otherwise, Phase 1 ran all $\bar t$ iterations. Let $\bar\gamma$ minimize $\tilde g$ on $\mathcal G$ and apply the volume argument above with this $\bar\gamma$: there are an optimality iteration $t^*$ and a point $\gamma\in\mathcal G_\kappa(\bar\gamma)$ satisfying \eqref{eq:discarded}. Combining \eqref{eq:discarded} with \eqref{eq:cut_h},
$$
\tilde g(\gamma)\ >\ \widehat g(\gamma_{t^*})-2C_w\delta\ \ge\ \widehat g_{\mathrm{best}}-2C_w\delta .
$$
On the other hand, $\|\gamma-\bar\gamma\|_1\le\kappa\,L(R+1)$ because $\gamma\in\mathcal G_\kappa(\bar\gamma)$, so the Lipschitz property gives
$
\tilde g(\gamma)\le\min_{\mathcal G}\tilde g+C_g\kappa L(R+1)\le\min_{\mathcal G}\tilde g+\epsilon/16
$
by the choice of $\kappa$. Therefore \eqref{eq:master} holds.

Let $\mathcal S$ be the set of index policies produced in Phase 1 and define the restricted dual function
$$
g_{\mathcal S}(\gamma):=\max_{\sigma\in\mathcal S}(w+M^\top\gamma)^\top u(\sigma)-\gamma^\top c,
\qquad
g_{\min}^{\mathcal S}:=\min_{\gamma\in\mathcal G}g_{\mathcal S}(\gamma).
$$
Both $g$ and $g_{\mathcal S}$ are consistent in the sense defined above. For $g$ this is Lemma \ref{lmm:dual_props}. For $g_{\mathcal S}$: it is a maximum of affine functions, hence convex; its subgradients on $\mathcal G$ are of the form $Mu(\sigma)-c$ with $\|Mu(\sigma)-c\|_\infty\le B_M\bar v+\|c\|_\infty\le C_g$, which gives the Lipschitz bound; and at every queried $\gamma_t$ the returned policy $\sigma_{\gamma_t}$ maximizes $(w+M^\top\gamma_t)^\top u(\sigma)$ over all of $\Sigma$ and belongs to $\mathcal S$, so it maximizes over $\mathcal S$ as well, whence $g_{\mathcal S}(\gamma_t)=g(\gamma_t)$ and $\xi(\gamma_t)\in\partial g_{\mathcal S}(\gamma_t)$.

Applying \eqref{eq:master} to $\tilde g=g_{\mathcal S}$, and bounding $\widehat g_{\mathrm{best}}$ from below by Lemma \ref{lmm:oracle_acc}(i), namely $\widehat g_{\mathrm{best}}=\widehat g(\gamma_j)\ge g(\gamma_j)-C_w\delta\ge g_{\min}-C_w\delta$ for the index $j$ attaining the minimum,
\begin{equation}\label{eq:restricted_gap}
g_{\min}^{\mathcal S}\ \ge\ g_{\min}-3C_w\delta-\frac{\epsilon}{16}\ =\ g_{\min}-\frac{\epsilon}{8},
\end{equation}
using $3C_w\delta=\epsilon/16$.

For $\varphi\ge 0$ put $\Phi_\varphi(u):=w^\top u-\varphi\,\mathrm{viol}(u)$, and for a nonempty finite $\mathcal V\subseteq\Sigma_I$ write $u(\alpha)=\sum_{\sigma\in\mathcal V}\alpha_\sigma u(\sigma)$ for $\alpha\in\mathcal A(\mathcal V)$.
Since $\varphi\,(c_\ell-(Mu)_\ell)^+=\max_{0\le\gamma_\ell\le\varphi}\gamma_\ell(c_\ell-(Mu)_\ell)$ and the map $(\alpha,\gamma)\mapsto w^\top u(\alpha)+\gamma^\top(Mu(\alpha)-c)$ is bilinear on the product of the compact convex sets $\mathcal A(\mathcal V)$ and $[0,\varphi]^L$, the minimax theorem gives
\begin{equation}\label{eq:penalty_duality}
\max_{\alpha\in\mathcal A(\mathcal V)}\Phi_\varphi(u(\alpha))
=
\min_{\gamma\in[0,\varphi]^L}\ \Bigl(\max_{\sigma\in\mathcal V}(w+M^\top\gamma)^\top u(\sigma)-\gamma^\top c\Bigr).
\end{equation}

Take $\mathcal V=\Sigma_I$ first. For every $\gamma\ge 0$ the weights $w+M^\top\gamma$ are nonnegative, so by Theorem \ref{thm:index_optimality} the inner maximum over $\Sigma_I$ equals the maximum over all of $\Sigma$, and the right-hand side of \eqref{eq:penalty_duality} is $\min_{\gamma\in[0,\varphi]^L}g(\gamma)$. By Lemma \ref{lmm:dual_radius} the dual optimum $\gamma^*$ satisfies $\|\gamma^*\|_\infty\le\|\gamma^*\|_1\le R$, so $\gamma^*\in[0,R]^L\subseteq\mathcal G$, and by strong duality $g(\gamma^*)=\mathrm{OPT}$. Hence that minimum equals $\mathrm{OPT}$ for both $\varphi=R$ and $\varphi=R+1$; in particular $g_{\min}=\mathrm{OPT}$. By Proposition \ref{prop:achievable_set}, $\mathcal U$ is the set of nonnegative vectors dominated by some element of $\mathrm{conv}\{u(\sigma):\sigma\in\Sigma_I\}$, and $\Phi_\varphi$ is nondecreasing in $u$ because $w\ge 0$ and $M\ge 0$; hence maximizing $\Phi_\varphi$ over $\mathcal A(\Sigma_I)$ is the same as maximizing it over $\mathcal U$, so
\begin{equation}\label{eq:exact_penalty}
\max_{u\in\mathcal U}\Phi_\varphi(u)=\mathrm{OPT}
\qquad\text{for }\varphi\in\{R,R+1\}.
\end{equation}
Taking $\mathcal V=\mathcal S$ and $\varphi=R+1$ in \eqref{eq:penalty_duality}, and using \eqref{eq:restricted_gap} together with $g_{\min}=\mathrm{OPT}$,
\begin{equation}\label{eq:restricted_primal}
\max_{\alpha\in\mathcal A(\mathcal S)}\Phi_{R+1}(u(\alpha))
=
g_{\min}^{\mathcal S}
\ \ge\
\mathrm{OPT}-\frac{\epsilon}{8}.
\end{equation}

Phase 2 maximizes not $\alpha\mapsto\Phi_{R+1}(u(\alpha))$ but its numerical counterpart $\widehat\Phi$, obtained by replacing $u(\sigma)$ with $\widehat u(\sigma)$, and it does so only to accuracy $\epsilon/16$, which is what its last line requires. Since $\|\widehat u(\sigma)-u(\sigma)\|_\infty\le\delta_{\mathrm{rec}}$ for every $\sigma\in\mathcal S$, the same bound holds for any mixture, so for every $\alpha\in\mathcal A(\mathcal S)$,
$$
\bigl|\widehat\Phi(\alpha)-\Phi_{R+1}(u(\alpha))\bigr|
\le
\|w\|_1\delta_{\mathrm{rec}}+(R+1)LB_M\delta_{\mathrm{rec}}
=
C_w\,\delta_{\mathrm{rec}}
\le
\frac{\epsilon}{16}.
$$
If $\alpha$ is the returned point and $\alpha^\dagger$ maximizes $\Phi_{R+1}(u(\cdot))$ over $\mathcal A(\mathcal S)$, then
$$
\Phi_{R+1}(u(\alpha))
\ge
\widehat\Phi(\alpha)-\frac{\epsilon}{16}
\ge
\Bigl(\widehat\Phi(\alpha^\dagger)-\frac{\epsilon}{16}\Bigr)-\frac{\epsilon}{16}
\ge
\Phi_{R+1}(u(\alpha^\dagger))-\frac{3\epsilon}{16}
\ \ge\
\mathrm{OPT}-\frac{5\epsilon}{16},
$$
using \eqref{eq:restricted_primal}. Writing $u^*=u(\alpha)$, this says
\begin{equation}\label{eq:final_lower}
w^\top u^*-(R+1)\,\mathrm{viol}(u^*)\ \ge\ \mathrm{OPT}-\frac{5\epsilon}{16}.
\end{equation}
Since $u^*\in\mathcal U$, \eqref{eq:exact_penalty} with $\varphi=R$ gives $w^\top u^*-R\,\mathrm{viol}(u^*)\le\mathrm{OPT}$. Subtracting \eqref{eq:final_lower} from this inequality yields $\mathrm{viol}(u^*)\le 5\epsilon/16\le\epsilon$, which is the second claim, and substituting $\mathrm{viol}(u^*)\ge 0$ into \eqref{eq:final_lower} gives $w^\top u^*\ge\mathrm{OPT}-5\epsilon/16\ge\mathrm{OPT}-\epsilon$, which is the first.

Phase 1 performs at most $\bar t=\max\{1,\lceil 2L^2\ln(\sqrt L/\kappa)\rceil+1\}$ iterations, and
$$
\ln\frac{\sqrt L}{\kappa}
=
\tfrac12\ln L+\max\left\{\ln 2,\ \ln\frac{16\,C_gL(R+1)}{\epsilon}\right\}
=
O\bigl(\log L+\log C_g+\log(R+1)+\log(1/\epsilon)\bigr),
$$
which is the bound stated in the theorem. Each iteration performs one $\mathrm{cut}_L$ update, costing $O(L^2)$ operations, and at most one oracle call. By Lemma \ref{lmm:index_eval} an oracle call at tolerance $\delta_{\mathrm{orc}}=\epsilon/(48C_w)$ costs
$$
O\left(KL+K\log K+\frac{K^2\bar\lambda}{d}+K\log\frac{48\,C_w\,\mu\bar v}{\underline\lambda\,\epsilon}\right)
$$
arithmetic operations. In Phase 2, $P=|\mathcal S|\le\bar t$, and re-evaluating $\mathcal S$ at tolerance $\delta_{\mathrm{rec}}$ costs $P$ evaluations of the same order. The data of $\widehat\Phi$ is formed in $O(PL)$ operations, and meeting the $\epsilon/16$ requirement needs nothing beyond the recursion of Definition \ref{def:cut}.

Put $n=P-1$, identify $\mathcal A(\mathcal S)$ with $\mathcal A'=\{\alpha'\in\mathbb R^{n}_+\mid\mathbf 1^\top\alpha'\le 1\}$ through $\alpha_{\sigma_j}=\alpha'_j$ for $j<P$ and $\alpha_{\sigma_P}=1-\mathbf 1^\top\alpha'$, and set
$$
C_\Phi=2C_w(\bar v+\delta_{\mathrm{rec}}),
\quad
\kappa_{\mathrm{rec}}=\min\Bigl\{\tfrac12,\ \frac{\epsilon}{32\,C_\Phi}\Bigr\},
\quad
\bar t_{\mathrm{rec}}=\bigl\lceil 2n^2\bigl(\ln(4/\kappa_{\mathrm{rec}})+\ln P\bigr)\bigr\rceil+1 .
$$
Run the recursion from $\alpha'_0=\mathbf 1/P$ and $D'_0=4I$ for $\bar t_{\mathrm{rec}}$ steps, cutting with a violated constraint of $\mathcal A'$ where there is one and with $-\zeta$ for a supergradient $\zeta$ of $\widehat\Phi$ otherwise, and return the best point visited inside $\mathcal A'$. The argument made for Phase 1 then repeats. Writing $a_j=w^\top\widehat u(\sigma_j)$ and $b_j=M\widehat u(\sigma_j)$, each entry of a supergradient has the form $(a_j-a_P)+(R+1)\sum_{\ell\in A}(b_{j\ell}-b_{P\ell})$ for the set $A$ of constraints violated there, hence is at most $C_\Phi$ in absolute value, so $\widehat\Phi$ is $C_\Phi$-Lipschitz for the $\ell_1$ norm. The two feasibility cuts discard no point of $\mathcal A'$. Every $\alpha'\in\mathcal A'$ has $\|\alpha'-\alpha'_0\|_2\le\|\alpha'\|_1+\|\alpha'_0\|_1\le 2$, so $\mathcal A'\subseteq\mathcal B(\alpha'_0,D'_0)$ and $\mathrm{vol}(\mathcal B(\alpha'_0,D'_0))\le 4^{n}$, while the copy of $\mathcal A'$ scaled by $\kappa_{\mathrm{rec}}$ about a maximizer has volume $\kappa_{\mathrm{rec}}^{n}/n!$; since $\ln(n!)\le n\ln P$, the choice of $\bar t_{\mathrm{rec}}$ makes the final ellipsoid the smaller of the two. Some point of that copy is therefore discarded, necessarily by a supergradient cut at a visited $\alpha'_t$, and concavity gives $\widehat\Phi(\alpha'_t)$ above the value there, which in turn is within $2C_\Phi\kappa_{\mathrm{rec}}\le\epsilon/16$ of the maximum.

Each of the $\bar t_{\mathrm{rec}}=O\bigl(P^2(\log P+\log C_w+\log\bar v+\log(1/\epsilon))\bigr)$ iterations costs $O(PL+P^2)$ operations, and $P\le\bar t$, so Phase 2 is polynomial in $L$ and the logarithms listed in the theorem.

Finally,
$
\log C_w=O(\log\|w\|_1+\log\bar v+\log(1/s_0)+\log L+\log B_M)
$
and
$
\log C_g=O(\log B_M+\log\bar v),
$
and $\log(R+1)=O(\log\|w\|_1+\log\bar v+\log(1/s_0))$, so every logarithm above is bounded by a polynomial in the quantities listed in the theorem. Summing the three contributions gives the stated bound.
\end{proof}

\subsection{Proof of Theorem \ref{thm:tiered}}\label{app:proof_tiered}

We use the notation of Section \ref{subsec:implementation}: $\mathcal T(\lambda)=\mathcal T^\rho(\lambda,\mu)$, the cumulative arrival rates $\Lambda_k=\sum_{j\le k}\phi_j\lambda_{\mathrm{cl}(j)}$ of the tiers of a tier list, and $|\mathcal I|$, $\tau(\mathcal I)$ and the windows $W_s(\mathcal I,\beta)$ for a finite union $\mathcal I$ of intervals.
By Lemma \ref{lmm:T_props}(a) below, $\tau(\mathcal I)=\int_{\mathcal I}\mathcal T'(s)\,ds$.

\begin{lemma}\label{lmm:T_props}
\begin{enumerate}
\item[(a)] $\mathcal T$ is infinitely differentiable on $[0,\infty)$, with $\mathcal T(0)=0$ and $0\le \mathcal T(\lambda)\le\lambda$.
\item[(b)] In any instance of the model with the same $\mu$, $\rho$ and common abandonment rate $d$, every stationary policy $\sigma$ and every set $G$ of classes satisfy $\sum_{i\in G}\lambda_iu_i(\sigma)/v_i\le \mathcal T\bigl(\sum_{i\in G}\lambda_i\bigr)$.
\item[(c)] $\mathcal T$ is concave and nondecreasing, and $\mathcal T'$ is nonincreasing with $0\le \mathcal T'\le1$.
\end{enumerate}
\end{lemma}
\begin{proof}
(a) Let $a_0=1$ and $a_n=\prod_{m=1}^n(dm+\mu g_\rho(m))^{-1}$, so that $\Pi^\rho_n(\lambda,\mu)=a_n\lambda^n/Z(\lambda)$ with $Z(\lambda)=\sum_{n\ge0}a_n\lambda^n$ and $\mathcal T(\lambda)=\mu\sum_{n\ge1}a_ng_\rho(n)\lambda^n/Z(\lambda)$.
Since $a_n\le1/(d^nn!)$, both power series converge everywhere, and $Z\ge1$ on $[0,\infty)$; hence $\mathcal T$ is infinitely differentiable there, $\mathcal T(0)=0$ and $\mathcal T\ge0$.
Detailed balance, $(dn+\mu g_\rho(n))\Pi^\rho_n(\lambda,\mu)=\lambda\Pi^\rho_{n-1}(\lambda,\mu)$, summed over $n\ge1$ gives
$$
\mathcal T(\lambda)\le\sum_{n\ge1}(dn+\mu g_\rho(n))\Pi^\rho_n(\lambda,\mu)=\lambda\sum_{n\ge0}\Pi^\rho_n(\lambda,\mu)=\lambda .
$$

(b) Take the weights $w_i=\lambda_i/v_i$ for $i\in G$ and $w_i=0$ otherwise.
Then the weighted welfare is $w^\top u=\sum_{i\in G}\lambda_iu_i/v_i$, and $b_i=1$ for $i\in G$ and $b_i=0$ otherwise.
By Theorem \ref{thm:index_optimality}, $w^\top u$ is maximized by an index policy that ranks the classes of $G$ first, and by Lemma \ref{lmm:index_util} this policy attains $w^\top u=\mathcal T\bigl(\sum_{i\in G}\lambda_i\bigr)$.

(c) Fix $a,h>0$ and consider three classes with arrival rates $(a,h,h)$ and unit values, ranked $1\succ3\succ2$.
By Lemma \ref{lmm:index_util}, classes 1 and 2 together receive the throughput $\mathcal T(a)+\mathcal T(a+2h)-\mathcal T(a+h)$, which by (b) with $G=\{1,2\}$ is at most $\mathcal T(a+h)$.
Thus $\mathcal T(a+h)\ge\frac12\bigl(\mathcal T(a)+\mathcal T(a+2h)\bigr)$ for all $a,h>0$, and a continuous function that is midpoint-concave on $(0,\infty)$ is concave on $[0,\infty)$.
A concave function on $[0,\infty)$ that is bounded below is nondecreasing.
Hence $\mathcal T'$ is nonincreasing and nonnegative, and $\mathcal T'\le \mathcal T'(0)\le1$ because $\mathcal T(0)=0$ and $\mathcal T(\lambda)\le\lambda$.
\end{proof}

\begin{lemma}\label{lmm:tier_formula}
Under a tiered policy, the average utility of class $i$ is
$$
u_i=\frac{v_i}{\lambda_i}\sum_{k:\,\mathrm{cl}(k)=i}\tau\bigl([\Lambda_{k-1},\Lambda_k]\bigr).
$$
\end{lemma}
\begin{proof}
By the thinning property of Poisson processes, arrivals to the tiers are independent Poisson processes with rates $\phi_k\lambda_{\mathrm{cl}(k)}=\Lambda_k-\Lambda_{k-1}$.
All patients abandon at rate $d$ and are compatible with an arriving organ independently with probability $\rho$.
The system with the tiers as classes is therefore an instance of the model, and the tiered policy is its index policy with ranking $1,\dots,n$.
By Lemma \ref{lmm:index_util}, a patient of tier $k$ has average utility $v_{\mathrm{cl}(k)}\,\tau([\Lambda_{k-1},\Lambda_k])/(\phi_k\lambda_{\mathrm{cl}(k)})$.
A class-$i$ patient belongs to tier $k$ with probability $\phi_k$, and averaging over the tiers of class $i$ gives the claim.
\end{proof}

\begin{lemma}\label{lmm:tier_achievable}
The utility vector of every tiered policy lies in $\mathcal U$.
\end{lemma}
\begin{proof}
Let $u$ be the utility vector of a tiered policy and fix weights $w\in\mathbb R^K_+$.
In the instance with the tiers as classes, give tier $k$ the weight $w_{\mathrm{cl}(k)}\phi_k$.
By Lemma \ref{lmm:tier_formula}, the weighted welfare of the tiered policy in that instance equals $w^\top u$.
The priority parameter of tier $k$ is $w_{\mathrm{cl}(k)}\phi_kv_{\mathrm{cl}(k)}/(\phi_k\lambda_{\mathrm{cl}(k)})=b_{\mathrm{cl}(k)}$, so tiers of the same class tie, and by Theorem \ref{thm:index_optimality} the weighted welfare is maximized by an index policy that ranks the tiers in decreasing order of $b_{\mathrm{cl}(k)}$ with the tiers of each class adjacent.
By Lemma \ref{lmm:tier_formula}, this policy has the same class-level utilities as the index policy $\sigma$ of the original model that ranks the classes by $b_i$, because the throughputs of adjacent tiers telescope.
Hence $w^\top u\le w^\top u(\sigma)\le\max_{u'\in\mathcal U}w^\top u'$ for every $w\in\mathbb R^K_+$.
By Proposition \ref{prop:achievable_set}, $\mathcal U=(H-\mathbb R^K_+)\cap\mathbb R^K_+$, where $H$ is the convex hull of the utility vectors of the index policies.
If $u\notin\mathcal U$, then, since $u\ge0$, $u$ lies outside the closed convex set $H-\mathbb R^K_+$, and a separating hyperplane gives $w\neq0$ with $w^\top u>\sup\{w^\top u'\mid u'\in H-\mathbb R^K_+\}$; the supremum is finite only if $w\ge0$, and then it equals $\max_{u'\in\mathcal U}w^\top u'$, a contradiction.
\end{proof}

The construction rests on dividing a set in a given proportion of both length and throughput.
Let $\mathcal I\subseteq[0,\Lambda]$ be a union of at most $p$ intervals with $|\mathcal I|>0$, let $\beta\in(0,1)$, and write $W_s=W_s(\mathcal I,\beta)$ for $s\in[0,(1-\beta)|\mathcal I|]$; formally, $W_s$ is the set of points $z\in \mathcal I$ with $s\le|\mathcal I\cap[0,z]|\le s+\beta|\mathcal I|$.
It is a union of at most $p$ intervals, and $\mathcal I\setminus W_s$ is a union of at most $p+1$ intervals; intervals of zero length are discarded throughout.

\begin{lemma}\label{lmm:split}
The function $s\mapsto\tau(W_s)$ is nonincreasing, satisfies $|\tau(W_s)-\tau(W_{s'})|\le|s-s'|$, and
$$
\tau(W_0)\ \ge\ \beta\,\tau(\mathcal I)\ \ge\ \tau\bigl(W_{(1-\beta)|\mathcal I|}\bigr).
$$
Consequently $\tau(W_{s^*})=\beta\,\tau(\mathcal I)$ for some $s^*$.
\end{lemma}
\begin{proof}
Let $s<s'$.
The sets $A=W_s\setminus W_{s'}$ and $B=W_{s'}\setminus W_s$ have the same length, which is at most $s'-s$, and $A$ lies to the left of $B$.
Since $\mathcal T'$ is nonincreasing with values in $[0,1]$ (Lemma \ref{lmm:T_props}), $\tau(B)\le\tau(A)\le s'-s$, and $\tau(W_{s'})-\tau(W_s)=\tau(B)-\tau(A)\in[-(s'-s),0]$.
Next, $W_0$ lies to the left of $\mathcal I\setminus W_0$, so the average of $\mathcal T'$ over $W_0$ is at least its average over $\mathcal I\setminus W_0$ and hence at least its average over $\mathcal I$: $\tau(W_0)/(\beta|\mathcal I|)\ge\tau(\mathcal I)/|\mathcal I|$.
The other inequality is symmetric, and the last assertion follows from continuity.
\end{proof}

\begin{lemma}\label{lmm:tier_exact}
Let $\mathcal S$ be a finite set of index policies and $\alpha\in\mathcal A(\mathcal S)$. There exists a tier list with at most $K|\mathcal S|^3$ tiers whose tiered policy has the utility vector $u^\alpha=\sum_{\sigma\in\mathcal S}\alpha_\sigma u(\sigma)$.
\end{lemma}
\begin{proof}
For $\sigma\in\mathcal S$ let $\Lambda^\sigma_k$ be the total arrival rate of its $k$ highest-priority classes, so that by Lemma \ref{lmm:index_util} the class ranked $k$-th occupies $[\Lambda^\sigma_{k-1},\Lambda^\sigma_k]$ and has throughput $\tau([\Lambda^\sigma_{k-1},\Lambda^\sigma_k])$.

\emph{Cells.}
Let $0=\omega_0<\omega_1<\dots<\omega_N=\Lambda$ be the distinct values of $\Lambda^\sigma_k$ over $\sigma\in\mathcal S$ and $k=0,\dots,K$; then $N\le K|\mathcal S|$.
On each cell $C_j=[\omega_{j-1},\omega_j]$ every $\sigma\in\mathcal S$ places a single class.
Let $\beta_{ij}$ be the total weight $\alpha_\sigma$ of the policies $\sigma\in\mathcal S$ that place class $i$ on $C_j$, so that $\sum_i\beta_{ij}=1$.
The interval of class $i$ under $\sigma$ is the union of the cells on which $\sigma$ places $i$.
Summing its throughput and its length over $\sigma$ with weights $\alpha_\sigma$ therefore gives, for every class $i$,
\begin{equation}\label{eq:cells}
\sum_{\sigma\in\mathcal S}\alpha_\sigma\frac{\lambda_iu_i(\sigma)}{v_i}=\sum_{j=1}^N\beta_{ij}\,\tau(C_j),
\qquad
\lambda_i=\sum_{j=1}^N\beta_{ij}\,|C_j|.
\end{equation}

\emph{Dividing a cell.}
Fix a cell $C_j$ and let $i_1,\dots,i_{m_j}$ be the classes with $\beta_{ij}>0$; write $\beta_k=\beta_{i_k,j}$ and $m=\max_jm_j$, so that $m\le|\mathcal S|$.
Set $\mathcal I_0=C_j$.
For $k=1,\dots,m_j-1$, apply Lemma \ref{lmm:split} to $\mathcal I_{k-1}$ with the proportion $\beta_k/(\beta_k+\dots+\beta_{m_j})$ to obtain a window $W_k\subseteq \mathcal I_{k-1}$ that has this proportion of both $|\mathcal I_{k-1}|$ and $\tau(\mathcal I_{k-1})$, and let $\mathcal I_k=\mathcal I_{k-1}\setminus W_k$; finally let $W_{m_j}=\mathcal I_{m_j-1}$.
By induction, $\mathcal I_k$ has the proportion $\beta_{k+1}+\dots+\beta_{m_j}$ of both $|C_j|$ and $\tau(C_j)$, and therefore
\begin{equation}\label{eq:pieces}
|W_k|=\beta_k\,|C_j|,\qquad\tau(W_k)=\beta_k\,\tau(C_j),\qquad k=1,\dots,m_j.
\end{equation}
Since $\mathcal I_{k-1}$ consists of at most $k$ intervals, so does $W_k$, and the cell is divided into at most $1+2+\dots+m_j\le m_j^2\le m^2$ intervals; label those of $W_k$ with the class $i_k$.

\emph{Tier list.}
Divide every cell in this way and merge adjacent intervals with the same label, and let the resulting intervals $[a_k,b_k]$, in increasing order, be the tiers, with $\mathrm{cl}(k)$ the label of $[a_k,b_k]$ and $\phi_k=(b_k-a_k)/\lambda_{\mathrm{cl}(k)}$, so that $[a_k,b_k]=[\Lambda_{k-1},\Lambda_k]$.
Merging changes neither the total length nor, since $\tau$ is additive over adjacent intervals, the total throughput of the intervals with a given label.
By \eqref{eq:pieces} and the second identity in \eqref{eq:cells}, the intervals of class $i$ have total length $\lambda_i$, so the shares of its tiers sum to one and this is a tier list, with at most $Nm^2\le K|\mathcal S|^3$ tiers.
By Lemma \ref{lmm:tier_formula}, \eqref{eq:pieces} and the first identity in \eqref{eq:cells}, its utility vector is
$u_i=\frac{v_i}{\lambda_i}\sum_j\beta_{ij}\,\tau(C_j)=\sum_{\sigma\in\mathcal S}\alpha_\sigma u_i(\sigma)$.
\end{proof}

\begin{lemma}\label{lmm:T_eval}
For every $\lambda\in[0,\Lambda]$, $\delta>0$ and integer $\hat N\ge\max\bigl\{2eK\bar\lambda/d,\ \log_2(2\mu/\delta)\bigr\}$, the truncated recursion of Appendix~\ref{app:proof_index_eval} at level $\hat N$ returns a value $\widehat{\mathcal T}(\lambda)$ with $|\widehat{\mathcal T}(\lambda)-\mathcal T(\lambda)|\le\delta$, using $O(\hat N)$ arithmetic operations.
\end{lemma}
\begin{proof}
The proof of Lemma \ref{lmm:index_eval} shows that the truncation error at level $\hat N$ is at most
$$
2\mu\sum_{n>\hat N}\frac{(\lambda/d)^n}{n!},
\qquad\text{and that}\qquad
\sum_{n>\hat N}\frac{a^n}{n!}\le2^{-\hat N}\quad\text{whenever } \hat N\ge2ea .
$$
Since $\lambda\le\Lambda\le K\bar\lambda$, this applies with $a=K\bar\lambda/d$, and $2\mu\,2^{-\hat N}\le\delta$ for $\hat N\ge\log_2(2\mu/\delta)$.
The recursion computes $\hat N$ terms, each from the previous one in a bounded number of operations.
\end{proof}

\begin{proof}[Proof of Theorem \ref{thm:tiered}]
\emph{Procedure.}
Algorithm \ref{algo:tier} follows the construction in the proof of Lemma \ref{lmm:tier_exact}, with the cells $C_j$, the numbers $\beta_{ij}$ and the sets $\mathcal I_{k-1}$ and $W_k$ defined there, except that each window $W_k$ is located by bisection using $\widehat{\mathcal T}$ in place of $\mathcal T$.
By Lemma \ref{lmm:T_eval}, $|\widehat{\mathcal T}-\mathcal T|\le\delta$ on $[0,\Lambda]$ for the constants $\delta$ and $\hat N$ of the algorithm, so on a union of at most $m$ intervals $|\widehat\tau-\tau|\le2m\delta$.
A window has the required length wherever it is placed, so the output is a tier list, and as in Lemma \ref{lmm:tier_exact} it has at most $Nm^2\le K|\mathcal S|^3$ tiers.

\emph{Accuracy.}
Throughout a bisection, the left endpoint $\underline s$ satisfies $\tau(W_{\underline s})\ge\beta\tau(\mathcal I)-4m\delta$ and the right endpoint $\bar s$ satisfies $\tau(W_{\bar s})\le\beta\tau(\mathcal I)+4m\delta$.
Initially this holds without the error terms, by Lemma \ref{lmm:split}.
Afterwards an endpoint moves to a midpoint $s$ only if $\widehat\tau(W_s)\ge\beta\widehat\tau(\mathcal I)$ (left endpoint) or $\widehat\tau(W_s)<\beta\widehat\tau(\mathcal I)$ (right endpoint), and both approximate values are within $2m\delta$ of the truth.
At termination $\bar s-\underline s\le\eta_2$, so the Lipschitz bound of Lemma \ref{lmm:split} gives
$$
\bigl|\tau(W_k)-\beta\,\tau(\mathcal I_{k-1})\bigr|\le\eta_3:=\eta_2+4m\delta .
$$
Let $\Delta_k=\tau(\mathcal I_k)-(\beta_{k+1}+\dots+\beta_{m_j})\tau(C_j)$, so $\Delta_0=0$.
Since $\tau(\mathcal I_k)=\tau(\mathcal I_{k-1})-\tau(W_k)$, the last display yields $|\Delta_k|\le|\Delta_{k-1}|+\eta_3$, so that $|\Delta_k|\le k\eta_3$ by induction, and $|\tau(W_k)-\beta_k\tau(C_j)|\le|\Delta_{k-1}|+\eta_3$, hence $|\tau(W_k)-\beta_k\tau(C_j)|\le m\eta_3$ for every $k\le m_j$ (for $k=m_j$ the difference equals $\Delta_{m_j-1}$).
Summing over the $N$ cells and using Lemma \ref{lmm:tier_formula} and \eqref{eq:cells}, the throughput $\lambda_iu_i/v_i$ of each class differs from that of the mixture by at most
$Nm\eta_3=Nm\eta_2+4Nm^2\delta\le\eta_1$.
Hence $|u_i-u^\alpha_i|\le(v_i/\lambda_i)\,\eta_1\le\eta$ for every class $i$.

\emph{Cost.}
Each bisection takes at most $\lceil\log_2(\Lambda/\eta_2)\rceil$ steps, each step at most $2m$ evaluations of $\widehat{\mathcal T}$, and computing $\widehat\tau(\mathcal I)$ at most $2m$ more.
There are at most $m$ bisections in each of the $N$ cells, so the number of evaluations is at most $2Nm^2\bigl(\lceil\log_2(2Nm\Lambda/\eta_1)\rceil+1\bigr)$, which is at most $\bar E$ because $N\le K|\mathcal S|$, $m\le|\mathcal S|$, $\Lambda\le K\bar\lambda$ and $\eta_1=\underline\lambda\,\eta/\bar v$.
By Lemma \ref{lmm:T_eval}, each evaluation costs $O\bigl(K\bar\lambda/d+\log(\mu/\delta)\bigr)=O\bigl(K\bar\lambda/d+\log(K|\mathcal S|\mu\bar v/(\underline\lambda\eta))\bigr)$ operations.
Sorting the at most $K|\mathcal S|$ cut points, computing the numbers $\beta_{ij}$, and updating the at most $m$ intervals of a window at each bisection step take $O(K|\mathcal S|\log(K|\mathcal S|)+K|\mathcal S|^2)$ and $O(m)$ operations respectively, which the bound on the evaluations dominates.
\end{proof}

\section{Calibration to U.S. Regional Kidney-Transplant Data}
\label{app:calibration_us}

This appendix describes how the parameters used in Section~\ref{subseq:experiments_calibration} are constructed from publicly available U.S.\ kidney-transplant data.

We restrict attention to a one-type environment,
\[
|\Theta|=1, \qquad v=1,
\]
and identify regions with the eleven OPTN regions. The organ is restricted to kidneys from deceased donors. The main data sources are the OPTN public database and the OPTN/SRTR Annual Data Report \citep{optn_public_data, schladt2025optn}. The calibrated parameters are the regional patient-arrival rates $\lambda^r$, the regional organ-arrival rates $\mu^r$, a common abandonment rate $d$, and a common compatibility parameter $\rho$.

\subsection{Flow parameters}

For each region $r$, let $A^r$ denote annual kidney waitlist additions in 2023. We set
\[
\lambda^r=\frac{A^r}{365}.
\]
Thus, $\lambda^r$ is calibrated using publicly reported kidney waitlist additions by OPTN region \citep{optn_public_data}.

Next, let $O^r_{rec}$ denote the annual number of recovered deceased-donor kidneys attributed to region $r$ in 2023. We set
\[
\mu^r=\frac{O^r_{rec}}{365}.
\]

The abandonment rate is assumed to be common across regions. Let $D^{\mathrm{US}}$ denote the total number of kidney waitlist removals in 2023, coded as either death or too sick to transplant, and let $W^{\mathrm{US}}$ denote the total number of kidney candidates on the waiting list during 2023. We set
\[
d=\frac{D^{\mathrm{US}}}{W^{\mathrm{US}}\cdot 365}.
\]
The removal counts are taken from the OPTN public database, and the aggregate kidney waiting-list population is taken from the OPTN/SRTR Annual Data Report \citep{optn_public_data, schladt2025optn}. This should be interpreted as a coarse proxy for the common daily abandonment rate in the model.

We do not attempt to calibrate $\rho$ separately from public aggregate data. The natural empirical moments that would inform $\rho$, such as discard rates, transplant volumes, waiting times, or queue lengths, are all policy-dependent. Without specifying and defending a particular reference policy, these moments do not identify $\rho$ in a policy-invariant way.

We therefore set $\rho=0.2$ in the baseline. This value should be read as a reduced-form average medical compatibility parameter in the one-type model, not as a direct estimate of biological compatibility for any particular candidate group. We take $\rho=0.2$ as a conservative baseline and report robustness for $\rho\in\{0.1,0.2,0.3,0.4,0.5,0.6,0.7,0.8,0.9,1.0\}$.

The baseline calibration is therefore
\begin{align}
\lambda^r
&=
\frac{\text{2023 kidney waitlist additions in region }r}{365},
\\
\mu^r
&=
\frac{\text{2023 recovered deceased-donor kidneys in region }r}{365},
\\
d
&=
\frac{\text{2023 U.S.\ kidney waitlist removals coded as death or too sick}}
{\text{2023 U.S.\ kidney waiting-list population}\cdot 365},
\\
\rho
&=0.2.
\end{align}

Table~\ref{tab:us_calibrated_primitives} reports the calibrated regional primitives used in the baseline simulation. Table~\ref{tab:us_common_parameters} reports the common parameters $d$ and $\rho$.

\begin{table}[htbp]
    \centering
    \begin{tabular}{cccc}
        \toprule
        \textbf{OPTN region}
        & $\lambda^r$
        & $\mu^r$
        & $\mu^r/\lambda^r$ \\
        \midrule
        1  &  5.1836 &  2.8438 & 0.549 \\
        2  & 14.8192 &  8.5151 & 0.575 \\
        3  & 18.1288 & 13.1260 & 0.724 \\
        4  & 14.5890 &  7.7507 & 0.531 \\
        5  & 21.0329 & 12.2521 & 0.583 \\
        6  &  3.0658 &  3.9479 & 1.288 \\
        7  & 11.0137 &  5.7479 & 0.522 \\
        8  &  7.0438 &  6.3178 & 0.897 \\
        9  &  8.9863 &  4.1753 & 0.465 \\
        10 &  8.1014 &  9.3370 & 1.153 \\
        11 & 16.7890 & 10.3753 & 0.618 \\
        \bottomrule
    \end{tabular}
    \caption{Calibrated regional primitives. Rates are daily rates, constructed from 2023 OPTN data.}
    \label{tab:us_calibrated_primitives}
\end{table}

\begin{table}[htbp]
    \centering
    \small
    \begin{tabular}{lcc}
        \toprule
        \textbf{Common parameter} & $d$ & $\rho$ \\
        \midrule
        Calibrated parameter value & $0.000167$ & 0.2 \\
        \bottomrule
    \end{tabular}
    \caption{Common parameters used in the baseline calibration. The abandonment rate $d$ is a daily rate, calibrated from aggregate U.S.\ kidney waitlist removals coded as death or too sick to transplant and the aggregate kidney waiting-list population.}
    \label{tab:us_common_parameters}
\end{table}

\subsection{Policy construction}

In the one-type calibration, each region corresponds to a single patient class. A pooled static index policy is therefore equivalent to a priority ordering of the eleven OPTN regions, and, since every index policy is nonwasteful, all such policies are utilitarian. Let $\Pi(\mathcal R)$ denote the set of all such orderings. For each $\pi\in\Pi(\mathcal R)$, let $\sigma_\pi$ be the static index policy that assigns each arriving kidney to the highest-priority compatible region according to $\pi$, and let $u_\pi$ be the corresponding regional utility vector.

The autarky benchmark, denoted by $sep$, is constructed by evaluating each region as a separate queue with patient-arrival rate $\lambda^r$ and organ-arrival rate $\mu^r$. Since $v=1$, the autarky utility $u^r_{sep}$ is the transplant probability of an arriving patient from region $r$.

The pooled benchmark, denoted $SL$, is the \emph{single-list benchmark}. Let
\[
\Lambda=\sum_{r\in\mathcal R}\lambda^r,
\]
and let $T_{pool}$ denote the aggregate transplant throughput produced by any nonwasteful pooled policy. In the one-type model with common abandonment, the aggregate queue-length process does not depend on the regional allocation rule, provided the rule does not waste a compatible kidney. All nonwasteful pooled policies therefore produce the same aggregate throughput $T_{pool}$.

The single-list benchmark assigns the same transplant probability to every patient:
\[
u^r_{SL}=\frac{T_{pool}}{\Lambda}
\qquad
\forall r\in\mathcal R.
\]
Equivalently, region $r$ receives transplant flow
\[
\lambda^r u^r_{SL}
=
\frac{\lambda^r}{\Lambda}T_{pool}.
\]
An actual first-come-first-served single list uses patient arrival order and is not stationary in the reduced queue-length state, but the utility vector $u_{SL}$ lies in the convex hull of stationary static index policies. A region-blind stationary rule that treats compatible waiting patients symmetrically gives every patient the same transplant probability, and the index-policy characterization of the achievable set implies that the same utility vector can be implemented by a mixture of static index policies.

The participation-constrained policy is computed by CPFS. In this calibration, CPFS solves
\[
\max_{u\in\mathcal U}
\sum_{r\in\mathcal R}\lambda^r u^r
\qquad
\text{s.t.}
\qquad
u^r\ge u^r_{sep}
\quad \forall r\in\mathcal R.
\]
The output is a mixture of static index policies, denoted by $C$, with utility vector $u_C$.

For each policy $\sigma$, we report aggregate transplant throughput
\[
T_\sigma=\sum_{r\in\mathcal R}\lambda^r u^r_\sigma,
\]
the relative gain over autarky,
\[
\frac{T_\sigma-T_{sep}}{T_{sep}},
\]
the minimum participation slack,
\[
\min_{r\in\mathcal R}\{u^r_\sigma-u^r_{sep}\},
\]
and the number of violated regional participation constraints,
\[
\#\{r\in\mathcal R:u^r_\sigma<u^r_{sep}\}.
\]

\subsection{Regional utility outcomes}

Table~\ref{tab:us_regional_utilities} reports the regional utilities and participation slacks underlying Figure~\ref{fig:us_participation_slack}.
Because the one-type objective is the aggregate throughput, which is the same under every nonwasteful policy, the constrained program has a flat objective and its optimal face is not a singleton: the $u^r_C$ column reports one optimal utility vector, and the aggregate figures of Table~\ref{tab:us_policy_comparison} are the same for all of them.

\begin{table}[htbp]
    \centering
    \begin{tabular}{cccccc}
        \toprule
        \textbf{OPTN region}
        & $u^r_{sep}$
        & $u^r_{SL}$
        & $u^r_C$
        & $u^r_{SL}-u^r_{sep}$
        & $u^r_C-u^r_{sep}$ \\
        \midrule
        1  & 0.5486 & 0.6554 & 0.5486 & 0.1068 & 0.0000 \\
        2  & 0.5746 & 0.6554 & 0.5746 & 0.0808 & 0.0000 \\
        3  & 0.7240 & 0.6554 & 0.8409 & $-0.0686$ & 0.1168 \\
        4  & 0.5313 & 0.6554 & 0.5313 & 0.1242 & 0.0000 \\
        5  & 0.5825 & 0.6554 & 0.5825 & 0.0729 & 0.0000 \\
        6  & 0.9995 & 0.6554 & 1.0000 & $-0.3441$ & 0.0005 \\
        7  & 0.5219 & 0.6554 & 0.5219 & 0.1335 & 0.0000 \\
        8  & 0.8969 & 0.6554 & 0.8969 & $-0.2415$ & 0.0000 \\
        9  & 0.4646 & 0.6554 & 0.4646 & 0.1908 & 0.0000 \\
        10 & 0.9997 & 0.6554 & 1.0000 & $-0.3443$ & 0.0003 \\
        11 & 0.6180 & 0.6554 & 0.6180 & 0.0374 & 0.0000 \\
        \bottomrule
    \end{tabular}
    \caption{Regional utilities and participation slacks.}
    \label{tab:us_regional_utilities}
\end{table}

\subsection{Robustness to the compatibility parameter}

Table~\ref{tab:rho_robustness} reports a simple robustness check for $\rho\in\{0.1,0.2,0.3,0.4,0.5,0.6,0.7,0.8,0.9,1.0\}$. The table reports the aggregate throughput of the three policies, and the smallest participation slack and the number of violated constraints under the single-list benchmark. This robustness exercise is not intended to estimate $\rho$, but to show that the qualitative participation pattern is insensitive to the baseline compatibility value.

\begin{table}[htbp]
    \centering
    \begin{tabular}{cccccc}
        \toprule
        $\rho$
        & $T_{sep}$
        & $T_{SL}$
        & $T_{C}$
        & $\min_r (u^r_{SL}-u^r_{sep})$
        & \textbf{Violations under $SL$} \\
        \midrule
        0.1 & 82.2647 & 84.389 & 84.389 & $-0.3441$ & 4 \\
        0.2 & 82.2676 & 84.389 & 84.389 & $-0.3443$ & 4 \\
        0.3 & 82.2685 & 84.389 & 84.389 & $-0.3444$ & 4 \\
        0.4 & 82.2689 & 84.389 & 84.389 & $-0.3444$ & 4 \\
        0.5 & 82.2692 & 84.389 & 84.389 & $-0.3444$ & 4 \\
        0.6 & 82.2693 & 84.389 & 84.389 & $-0.3444$ & 4 \\
        0.7 & 82.2694 & 84.389 & 84.389 & $-0.3444$ & 4 \\
        0.8 & 82.2695 & 84.389 & 84.389 & $-0.3444$ & 4 \\
        0.9 & 82.2695 & 84.389 & 84.389 & $-0.3444$ & 4 \\
        1.0 & 82.2696 & 84.389 & 84.389 & $-0.3444$ & 4 \\
        \bottomrule
    \end{tabular}
    \caption{Robustness to the compatibility parameter.}
    \label{tab:rho_robustness}
\end{table}

The calibration is deliberately parsimonious. It is not intended to replicate the institutional details of U.S.\ kidney allocation; the role is to discipline the regional flow parameters used in the numerical illustration and to show how participation constraints arise in a data-informed environment.

\end{document}